\documentclass[11pt, onecolumn, draftcls, journal]{IEEEtran}

\usepackage{tikz}
\usetikzlibrary{arrows,positioning,plotmarks}
\usepackage{verbatim}
\usepackage{etex}
\usepackage[ngerman]{babel}
\usepackage[T1]{fontenc}
\usepackage{lmodern}
\usepackage{pgfplots}
\usepackage{epsfig,graphics,subfig,psfrag,amsmath,float}
\usepackage{latexsym,amssymb,amsfonts}
\usepackage{cite}
\usepackage{multirow}

\title{Beamforming on the MISO interference channel with multi-user decoding capability\thanks{This work has been performed in the framework of the European research
project SAPHYRE, which is partly funded by the European Union under 
its FP7 ICT Objective 1.1 - The Network of the Future. This work is 
also supported in part by the Deutsche Forschungsgemeinschaft (DFG) under grant Jo 801/4-1.}}
\author{K.~M.~Ho$^\dagger$, D.~Gesbert$^\dagger$, E.~Jorswieck$^\ast$, R.~Mochaourab$^\ast$\\ \IEEEaftertitletext{\vspace{-2\baselineskip}}}
{\newtheorem{Thm}{Theorem}}
\newtheorem{Lem}{Lemma}
\newtheorem{Cor}{Corollary}
\newtheorem{Def}{Definition}

\newtheorem{Remark}{Remark}
\DeclareMathOperator{\mrt}{\scriptscriptstyle{MRT}} 
\DeclareMathOperator*{\argmax}{arg\,max}

\begin{document}
\maketitle

\begin{abstract}
 This paper considers the multiple-input-single-output interference channel (MISO-IC) with interference decoding capability (IDC), so that the interference signal can be decoded and subtracted from the received signal. On the MISO-IC with single user decoding, transmit beamforming vectors are classically designed to reach a compromise between mitigating the generated interference (zero forcing of the interference) or maximizing the energy at the desired user. The particularly intriguing problem arising in the multi-antenna IC with IDC is that transmitters may now have the incentive to amplify the interference generated at the non-intended receivers, in the hope that Rxs have a better chance of decoding the interference and removing it. This notion completely changes the previous paradigm of balancing between maximizing the desired energy and reducing the generated interference, thus opening up a new dimension for the beamforming design strategy. 

Our contributions proceed by proving that the optimal rank of the transmit precoders, optimal in the sense of Pareto optimality and therefore sum rate optimality, is rank one. Then, we investigate suitable transmit beamforming strategies for different decoding structures and characterize the Pareto boundary. As an application of this characterization, we obtain a \emph{candidate set of the maximum sum rate point} which at least contains the set of sum rate optimal beamforming vectors. We derive the Maximum-Ratio-Transmission (MRT) optimality conditions. Inspired by the MRT optimality conditions, we propose a simple algorithm that achieves maximum sum rate in certain scenarios and suboptimal, in other scenarios comparing to the maximum sum rate. \footnote{This work was partially presented in \cite{Ho2010}.}
 \end{abstract}

\section{Introduction}
Despite efforts since pioneering work such as \cite{Carleial1975, Costa1985}, the capacity region of interference channel is still an open problem. Numerous work have attempted to compute achievable rate regions and outer bounds on  the Single-Input-Single-Output Interference Channel (SISO-IC). It is proved in \cite{Carleial1975} that the capacity of any two-user SISO-IC is the same as in a corresponding IC in standard form: direct gain as unity and interference gain as a real positive scalar. Even since, the capacity region of two-user SISO-IC has been studied extensively (see e.g. \cite{Shang2009, Motahari2009, Sason2004, Etkin2008, Weng2008} and the references therein) : 
\begin{itemize}
\item In the weak interference regime where the cross interference gain is much weaker than the direct channel gain, the sum rate capacity is shown to be achievable by treating interference as thermal noise at the receiver which requires no feedback communication between Rx $j$ and Tx $i$ \cite{Shang2009, Motahari2009}.
\item At the other extreme, in the strong and very strong interference regime, both users should decode the interference signal while treating the desired signal as noise. The decoded interference signal is then subtracted from the received signal allowing the  desired signal to be decoded without any interference. \cite{Han1981, Carleial1975, Sato1981}.  
\item In the mixed interference regime, where one cross interference gain is stronger than direct channel gain and the other link is weaker, the sum rate capacity is shown to be attained by one user decoding interference and the other user treating interference as noise \cite{Motahari2009, Weng2008}.
\item The deterministic channel approach offers a good approximation of  the sum capacity of interference channel. In the deterministic channel approach, the input-output relationship of the channel is modeled as a bit-shifting operation \cite{Etkin2008, Cadambe2009, Jafar2010}. In \cite{Etkin2008}, the two-user SISO-IC sum capacity is approximated to within one bit using the deterministic channel approach.
\end{itemize}

To extend the above results, the conditions in which treating interference as noise achieving capacity on the vector Gaussian interference channel is studied in \cite{Bandemer2008}. The capacity region of a specific class of MIMO interference channels is characterized in \cite{Vishwanath2004}. Against intuition, the optimal rank of input covariance matrices remains inconclusive, unlike in single user detection (SUD) case where single mode beamforming attains capacity \cite{Shang2009}. The authors in \cite{Vishwanath2004} showed that the optimal input covariance matrix attaining capacity  of MISO-IC has rank less than the number of users in the IC.

The frontier of the achievable rate region, also known as the Pareto boundary, holds importance to the understanding of IC. Any rate points on the Pareto boundary are operating points such that one user cannot increase its rate without decreasing other users rates. Assuming perfect CSIT, the Pareto boundary of SISO-IC and MISO-IC with SUD are characterized in \cite{Charafeddine2007, Jorswieck2008}  respectively. In \cite{Lindblom2010}, the authors extended the results to partial CSIT. In this paper, we assume simple single user encoding transmitters and interference decoding capability at receivers, which yield a simpler scheme comparing to the Han-Kobayashi scheme \cite{Han1981}.  We study the effects of transmit beamforming on the achievable rate region and to characterize the Pareto boundary. We limit ourselves to the two transmitter-receiver (Tx-Rx) pairs interference channel with IDC. We assume each receiver can choose to fully decode interference (D) or treat interference as noise (N). No rate splitting-based averaging is considered between these two modes. The transmit beamforming vectors are optimized to achieve an operating point as close to the Pareto boundary as possible. In IC with SUD, interference mitigation may seem to be a reasonable strategy. Yet, with the IDC which we address in this paper, it is possible for the user to manipulate it's beamforming vector such that the generated interference is amplified for easier interference removal and yield a better operating point. The fundamental question becomes: \emph{when should the beamforming vectors be designed to amplify interference to improve performance and when to mitigate interference?}

The main contributions of this paper are:
\begin{itemize}
\item In Section \ref{section:definitions}, we describe an achievable rate region of the MISO-IC with IDC, with the assumption of linear precoding, taking into account of receivers choice of actions, D or N. 
\item We study and characterize its Pareto boundary in Section \ref{section:pareto}, in terms of beamforming vectors design and power allocation. We characterize the set $\Omega$ of tuples of beamforming vectors and power allocation which attain the Pareto boundary. 
\item As a special case, in Section \ref{section:sumrate}, we characterize the set of beamforming vectors which attain the maximum sum rate point in the form of a candidate set  $\tilde{\Omega}$. As the maximum sum rate problem is non-convex, conventional solutions rely on different searching techniques. Note that $\tilde{\Omega} \subset \Omega$. The cardinality of $\tilde{\Omega}$ is much smaller than the cardinality of $\Omega$ which provides a significant reduction of searching complexity. Further, we prove that with IDC full power must be used at each Tx to attain the maximum sum rate point. This result is interesting as non-full power should be employed in some Txs to achieve the maximum sum rate point in the SISO-IC-SUD \cite{Kiani2009, Charafeddine2007}.
\item In Section \ref{section:mrtopt}, we investigate the conditions of channel parameters for which simple strategies are sum rate optimal. In particular, we study the matched filter (MF)  with respect to the desired channel  and the MF with respect to the interference channel, which are termed as the maximum-ratio-transmission (MRT) schemes.  
\item Inspired by the MRT optimality conditions, we propose a suboptimal but very low complexity beamforming design in Section \ref{sec:sim_algo}. The suboptimal algorithm shows a promising tradeoff between complexity and performance, as illustrated by simulation results.
\item In Section \ref{section:simulation}, we provide simulations and discussions which illustrate cases where interference decoding is most beneficial to sum rate performances.
\end{itemize}

\subsection{Notations} The lower case bold face letter represents a vector. The conjugate transpose is denoted by $(.)^H$. $\mathbb{R}$ represents the set of real numbers. The projection matrix on vector $\mathbf{x}$ is $\Pi_{\mathbf{x}}=\mathbf{x} \mathbf{x}^H / ||\mathbf{x}||^2$ and the orthogonal projection matrix is $\mathbf{I}-\Pi_{\mathbf{x}}$ where $\mathbf{I}$ is the identity matrix. Denote a boolean statement by $B_i$. The complement of the statement $B_i$ is $\bar{B}_i$. The OR operation is denoted as $\cup$; AND operation as $\cap$. $\boldsymbol{\nu}(\mathbf{A})$ returns the dominant eigenvector of matrix $\mathbf{A}$. $tr(\mathbf{A})$ is the trace of matrix $\mathbf{A}$. The matrix $\mathbf{A}$ is positive semi-definite if $\mathbf{A} \succeq 0$. The symbol $\Leftrightarrow$ represents the \emph{if-and-only-if} relationship between two statements. $Re(z)$ and $Im(z)$ give the real and imaginary part of complex number $z$. The function $\arg(z)$ gives the phase of the complex number $z$. The operator $\times$ is the Cartesian product operator between two sets.

\section{Channel model}
We assume a system of two transmitter-receiver (Tx-Rx) pairs in which each Tx has $N$ transmit antennas and each Rx has only one receive antenna. This results in a two-user Multiple-Input-Single-Output Interference Channel (MISO-IC), which is illustrated in Fig.  \ref{fig:channelmodel} as an example with $N=3$. We assume that the Txs are using commonly known codebooks and therefore the Rx, if the channel qualities allow, can decode the interference and subtract it from the received signal. Also, we assume that the interference is successfully decoded if the rate of the interference signal is smaller than the Shannon capacity of the interference channel.

\tikzstyle{box}=[rectangle,draw=black!50,thick,minimum size=4mm]
\tikzstyle{ball}=[circle,draw=black!50,inner sep=0.25ex]
\begin{figure}
\begin{center}
\resizebox{0.5 \columnwidth}{!}{
\begin{tikzpicture}[node distance=5cm,auto,>=latex']
\node at (0,0) [box] (Rx1) {Rx1};
\node at (0,4) [box] (Tx1) {Tx1};

\node at (4,0) [box] (Rx2) {Rx2};
\node at (4,4) [box] (Tx2) {Tx2};

\draw[-open triangle 45 reversed] (0.2,4.25) -- (0.2,4.6);
\draw[-open triangle 45 reversed] (0,4.25) -- (0,4.6);
\draw[-open triangle 45 reversed] (-0.2,4.25) -- (-0.2,4.6);

\draw[-open triangle 45 reversed] (4.2,4.25) -- (4.2,4.6);
\draw[-open triangle 45 reversed] (4,4.25) -- (4,4.6);
\draw[-open triangle 45 reversed] (3.8,4.25) -- (3.8,4.6);



\draw (Tx1) to node {$\mathbf{h}_{11}$} (Rx1) ;
\draw (Tx1.south) -- (Rx2.north)  ;
\node at (2.9,1) {$\mathbf{h}_{21}$};
\draw (Tx2.south) -- (Rx1.north);
\node at (2.9,3.2) {$\mathbf{h}_{12}$} ;
\draw (Tx2) to node  {$\mathbf{h}_{22}$} (Rx2);
\end{tikzpicture}
}
\caption{The 2 users MISO-IC where Txs are equipped with 3 antennas.\label{fig:channelmodel}}

\end{center}
\end{figure}
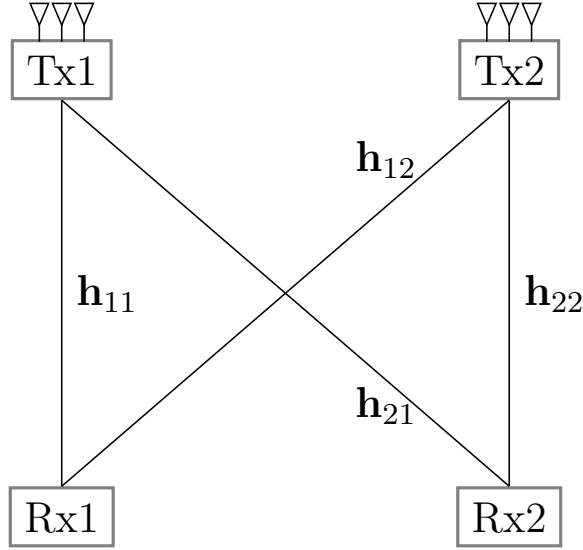

In the MISO-IC-SUD, it has been shown that the optimal transmit precoders are rank 1 and therefore beamforming attains the Pareto boundary. However, whether his conclusion holds in the MISO-IC-IDC is not known yet. We answer this question in the following by starting with a general transmit covariance matrix. Denote the transmit covariance matrix of Tx $i$ by $\mathbf{S}_i$ and the channel from Tx $i$ to Rx $\bar{i}$, where $i \in \left\{ 1, 2\right\}, \bar{i} \neq i$, $\mathbf{h}_{\bar{i}i}\in \mathbb{C}^{N \times 1}$. Note that the channel gains are proper i.i.d complex Gaussian coefficients with zero mean and unit variance. The received signal at Rx $i$ is therefore
\begin{equation}
y_i= \mathbf{h}_{ii}^H \mathbf{S}_i^{1/2} \mathbf{x}_i  + \mathbf{h}_{i\bar{i}}^H\mathbf{S}_{\bar{i}}^{1/2} \mathbf{x}_{\bar{i}} + n_i.
\end{equation}
The noise $n_i$ is a complex Gaussian random variable with zero mean and unit variance.  The symbol $\mathbf{x}_i$ is the transmit symbol at Tx $i$ with unit power. Denote the set of the transmit covariance matrices that satisfy the power constraint $tr(\mathbf{S}_i) \leq P_{max} $ to be  
\begin{equation}
\mathcal{S}= \left\{ \mathbf{S}\in \mathbb{C}^{N \times N}: \mathbf{S} \succeq 0,  tr(\mathbf{S}) \leq P_{max}\right\}, \; i=\{1,2\}.
\end{equation}

\section{Achievable Rate Region}\label{section:definitions}
We propose the following four decoding structures corresponding to the Rxs. actions: (N,N), (N,D), (D,N) and (D,D) \cite{Carleial1978}, with ``N" stands for treating interference as noise and ``D" stands for decoding and removing interference. Thus, (D,N) means Rx 1 decodes and removes interference and Rx 2  treats interference as noise. In \cite{Carleial1978}, these four decoding structures are proposed and its corresponding rate points are shown to be achievable in the SISO-IC. We extend the concept to the MISO-IC and define the following important quantities:
\begin{equation}\label{eqt:ch4_rate_def}
\begin{aligned}
C_1(\mathbf{S}_1) & \triangleq \log_2(1+\mathbf{h}_{11}^H \mathbf{S}_1 \mathbf{h}_{11}),\\ 
C_2(\mathbf{S}_2)  &\triangleq \log_2(1+\mathbf{h}_{22}^H \mathbf{S}_2 \mathbf{h}_{22}),\\
D_1(\mathbf{S}_1,\mathbf{S}_2)  &\triangleq \log_2 \left(1+ \frac{\mathbf{h}_{11}^H \mathbf{S}_1 \mathbf{h}_{11}}{ \mathbf{h}_{12}^H \mathbf{S}_2 \mathbf{h}_{12}+ 1} \right), \\
D_2(\mathbf{S}_1,\mathbf{S}_2)  &  \triangleq \log_2 \left(1+ \frac{\mathbf{h}_{22}^H \mathbf{S}_2 \mathbf{h}_{22}}{ \mathbf{h}_{21}^H \mathbf{S}_{1} \mathbf{h}_{21} + 1}\right),\\
T_2(\mathbf{S}_1,\mathbf{S}_2)    &\triangleq \log_2 \left(1+ \frac{\mathbf{h}_{12}^H \mathbf{S}_2 \mathbf{h}_{12}}{ \mathbf{h}_{11}^H \mathbf{S}_1 \mathbf{h}_{11}+ 1}\right), \\
T_1(\mathbf{S}_1,\mathbf{S}_2)    &\triangleq \log_2 \left(1+ \frac{\mathbf{h}_{21}^H \mathbf{S}_1 \mathbf{h}_{21}}{ \mathbf{h}_{22}^H \mathbf{S}_2  \mathbf{h}_{22} + 1}\right).
\end{aligned}
\end{equation} $C_1$ and $C_2$ are the \emph{single user rates}, the largest rate user 1 and 2 can achieve without the influence of interference. $D_1$ and $D_2$ are the rates corresponding to decoding the desired signal while treating interference as thermal noise and $T_1$  and $T_2$ are the rate corresponding to decoding the interference while treating the desired signals as noise.

Consequently, if both receivers decode interference, user $i$ must transmit at a rate that ensures interference decoding at Rx $\bar{i}$, thus we have the following:
\begin{equation} \label{eqt:rdd}
\begin{aligned}
R_1 & \leq \min \big\{ C_1(\mathbf{S}_1), T_1(\mathbf{S}_1,\mathbf{S}_2)  \big\}\\
R_2 & \leq  \min \big\{C_2(\mathbf{S}_2), T_2(\mathbf{S}_1,\mathbf{S}_2)  \big\}.
\end{aligned}
\end{equation}

Denote the rate region with interference decoding at both receivers by the Decode-Decode (DD) region:
\begin{equation}\label{def:rdd}
\mathcal{R}^{dd}= \bigcup_{ (\mathbf{S}_{1},\mathbf{S}_{2}) \in \mathcal{S} \times \mathcal{S} } \bigg\{ (R_1,R_2) \leq \bigg( \min \big\{ C_1(\mathbf{S}_1), T_1(\mathbf{S}_1,\mathbf{S}_2)  \big\},\min \big\{C_2(\mathbf{S}_2), T_2(\mathbf{S}_1,\mathbf{S}_2)  \big\} \bigg) \bigg\}.
\end{equation}

\begin{Remark}
For each selected pair of transmit beamformers, a corresponding rate region which satisfies the inequalities \eqref{eqt:rdd} is obtained. The achievable rate region $\mathcal{R}^{dd}$ is defined as the union of all regions achieved by all possible transmit beamformers.
\end{Remark}

On the other hand, if both Rxs choose to treat interference as noise, we obtain the NN region,
\begin{equation}\label{def:rnn}
 \mathcal{R}^{nn}= \bigcup_{\mathbf{S}_{1},\mathbf{S}_{2} \in \mathcal{S} } \bigg\{ (R_1,R_2) \leq \big(D_1(\mathbf{S}_1,\mathbf{S}_2) , D_2(\mathbf{S}_1,\mathbf{S}_2) \big) \bigg\}.
\end{equation}

If Rx 1 decodes interference but Rx 2 treats interference as noise,  Tx 2 must transmit at a rate that ensures interference decoding at Rx 1. Thus, the DN region is obtained as,
\begin{equation}\label{eqt:rdn_def}
\mathcal{R}^{dn}= \bigcup_{\mathbf{S}_{1},\mathbf{S}_{2} \in \mathcal{S} } \bigg\{ (R_1,R_2) \leq \bigg(C_1(\mathbf{S}_1) ,  \min \big\{D_2(\mathbf{S}_1,\mathbf{S}_2), T_2(\mathbf{S}_1,\mathbf{S}_2) \big\} \bigg) \bigg\}.
\end{equation}

\begin{Remark}
$\mathcal{R}^{dn}(\mathbf{S}_1,\mathbf{S}_2)$ is the rate region that the inequalities in \eqref{eqt:rdn_def} are satisfied for specific transmit covariance matrices $(\mathbf{S}_1,\mathbf{S}_2)$. It can be an empty region if the inequalities cannot be satisfied at the same time. This corresponds to the situation where the data rate of Tx 2 is too high for Rx 1 to decode.
\end{Remark}

Similarly,  we have for the ND region,
\begin{equation}\label{eqt:rnd_def}
\mathcal{R}^{nd}=\bigcup_{ \mathbf{S}_{1},\mathbf{S}_{2} \in \mathcal{S} } \bigg\{ (R_1, R_2) \leq \bigg( \min \big\{D_1(\mathbf{S}_1,\mathbf{S}_2),T_1(\mathbf{S}_1,\mathbf{S}_2) \big\},    C_2(\mathbf{S}_1) \bigg)  \bigg\}.
\end{equation}

Finally, one achievable rate region of the MISO-IC with interference decoding capability is therefore the union of the above regions:
\begin{equation}\label{eqt:ach_rate_region_union}
 \mathcal{R}= \mathcal{R}^{nn} \cup \mathcal{R}^{dd} \cup \mathcal{R}^{dn} \cup \mathcal{R}^{nd}. 
\end{equation}

We now turn our attention to the Pareto boundary of the rate region. To find the boundary achieving solutions, we proceed by identifying a set of smaller dimension than $\mathcal{S} \times \mathcal{S}$ which is guaranteed to contain the Pareto optimal solutions. We refer to such set as \emph{candidate set}. The main practical value of a candidate set is to offer a subtantial reduction of complexity compared with the exhaustive search over the full set $\mathcal{S} \times \mathcal{S}$.

\begin{Def}
Denote the set of points on the Pareto boundary of $\mathcal{R}$ by $\mathcal{B}(\mathcal{R})$. If the rate pair $(r_1,r_2) \in \mathcal{R}$ is on the boundary, $(r_1,r_2) \in \mathcal{B}(\mathcal{R})$, then there does not exist a rate pair $(r_1', r_2') \in \mathcal{R}$ such that $(r_1', r_2') \geq(r_1,r_2)$, with one strict inequality. Using $\mathcal{R}$ in \eqref{eqt:ach_rate_region_union},
\begin{equation}\label{def:pareto}
 \mathcal{B}(\mathcal{R}) \subset \mathcal{B}(\mathcal{R}^{nn}) \cup  \mathcal{B}(\mathcal{R}^{dd}) \cup  \mathcal{B}( \mathcal{R}^{dn}) \cup  \mathcal{B}(\mathcal{R}^{nd} )
\end{equation}
\end{Def}

\begin{Def}
The transmit covariance matrices $\mathbf{S}_1, \mathbf{S}_2$ are Pareto optimal in the rate region $\mathcal{R}$ if 
\begin{equation}
\bigg( R_1(\mathbf{S}_1, \mathbf{S}_2), R_2 (\mathbf{S}_1, \mathbf{S}_2)\bigg) \in \mathcal{B} (\mathcal{R}).
\end{equation}
\end{Def}

\begin{Def}
The candidate set $\Omega^{xy}$ of $\mathcal{B}(\mathcal{R}^{xy})$, $x,y \in \{ n,d\}$, is a set of transmit covariance matrices that contains the transmit covariance matrices that attain the Pareto boundary of $\mathcal{R}^{xy}$. If $(\mathbf{S}_1, \mathbf{S}_2)$ are Pareto optimal, then $(\mathbf{S}_1, \mathbf{S}_2) \in \Omega^{xy}$. Similarly, the candidate set of $\mathcal{B}(\mathcal{R})$ is $\Omega$ which contains all pairs of $(\mathbf{S}_1, \mathbf{S}_2)$ which are Pareto optimal in the region $\mathcal{R}$.
\end{Def}


\section{The Pareto optimal transmit covariance matrices}\label{section_bf}
In this section, we study the transmit covariance matrices that attain the Pareto boundary and prove that they are rank one.

\begin{Thm}\label{thm:transmit_cov}
The Pareto boundaries of the NN region, the DN region and the DD region are attained by rank 1 matrices. Consequently, the Pareto boundary of MISO-IC-IDC, defined in \eqref{def:pareto}, is attained by rank one transmit covariance matrices, or transmit beamforming. 
\end{Thm}
\begin{proof}
Here, we provide a sketch of the proof. For details, please refer to Appendix \ref{app:transmit_cov}. We first show that the boundaries of rate region $\mathcal{R}^{nd}$ and $\mathcal{R}^{dd}$ are attained by rank one matrices. By exchanging the roles of the transmitters, we obtain that $\mathcal{B}(\mathcal{R}^{dn})$ is attained by rank one matrices. From \cite{Shang2009, Mochaourab2010}, it is shown that the boundary in the NN region is attained by rank one matrices. Hence,  the boundaries of all decoding structures are attained by rank one transmit covariance matrices. Since the Pareto boundary of the proposed achievable rate region in the MISO-IC-IDC, defined in \eqref{def:pareto}, is a subset of the union of the above boundaries, we conclude that this Pareto boundary is attained by rank one transmit covariance matrices, or transmit beamforming.
\end{proof}

From Theorem \ref{thm:transmit_cov}, we have established that the Pareto boundary is attained by transmit beamforming vectors. To facilitate the following discussions, we define the transmit beamforming vectors $\mathbf{w}_i$ and transmit power $P_i$, for $i=1,2$,
\begin{equation}
\mathbf{S}_i= \mathbf{w}_i \mathbf{w}_i^H P_i
\end{equation} with $\| \mathbf{w}_i \|^2=1$. As an abuse of notation, we write $\mathcal{S}$ as the set of all possible beamforming vectors,
\begin{equation}
\mathbf{w}_i \in \mathcal{S}, \hspace{1cm} \mathcal{S}= \left\{ \mathbf{w} \in \mathbb{C}^{N \times 1}: \| \mathbf{w}\|=1 \right\}.
\end{equation} Consequently, we redefine the candidate sets in terms of transmit power allocations and beamforming vectors. The candidate set $\Omega^{xy}$ of $\mathcal{B}(\mathcal{R}^{xy}), x,y \in \{ n,d\}$ contains the Pareto optimal beamforming vectors and transmit power allocations.
\begin{equation}
\Omega^{xy} \supset \bigg\{ (\mathbf{w}_1, \mathbf{w}_2, P_1, P_2) : \big(R_1(\mathbf{w}_1,\mathbf{w}_2, P_1, P_2), R_2(\mathbf{w}_1, \mathbf{w}_2, P_1, P_2) \big) \in \mathcal{B}(\mathcal{R}^{xy}) \bigg\}
\end{equation}
and the candidate set of $\mathcal{B}(\mathcal{R})$ is 
\begin{equation}
\Omega \supset \bigg\{ (\mathbf{w}_1, \mathbf{w}_2, P_1, P_2)  : \big(R_1(\mathbf{w}_1, \mathbf{w}_2, P_1, P_2) , R_2(\mathbf{w}_1, \mathbf{w}_2, P_1, P_2)  \big) \in  \mathcal{B}(\mathcal{R}) \bigg\}.
\end{equation}

In the following sections, we study the Pareto boundary in terms of power allocation and transmit beamforming vectors in different decoding structures namely $\mathbf{R}^{nd}$ and $\mathbf{R}^{dd}$. $\mathbf{R}^{nn}$ is the case of MISO-IC-SUD and is well studied in \cite{Jorswieck2008}. $\mathbf{R}^{dn}$ is symmetric to $\mathbf{R}^{nd}$ and is therefore omitted here. Then as a special case, we discuss the characterization of the maximum sum rate point in each decoding structures. 

\section{The received power region}
The received power region was first proposed in \cite{Mochaourab2010} as a powerful tool to illustrate the dependancy between the received power tuple and the Pareto boundary on the $K$-user MISO-IC-SUD.  The tuple at one receiver includes the received power from the desired signal and the received power from the interference signal(s). In a two-user MISO-IC, we can illustrate the received power region as a two-dimensional plot, as shown in Fig. \ref{fig:powergain}. Mathematically, the received power region of user $i$ for the two-user MISO-IC is defined as:
\begin{equation}
\Phi_i= \left\{ \left(| \mathbf{h}_{ii}^H \mathbf{w}_i|^2 P_i, |\mathbf{h}_{ji}^H \mathbf{w}_i|^2 P_i \right): \mathbf{w}_i \in \mathcal{S}, 0 \leq P_i \leq P_{max}\right\}
\end{equation}
 where the desired channel power of user $i$ is $|\mathbf{h}_{ii}^H \mathbf{w}_i|^2 P_i$ and the interference channel power of user $i$ is $ |\mathbf{h}_{ji}^H \mathbf{w}_i|^2 P_i$.

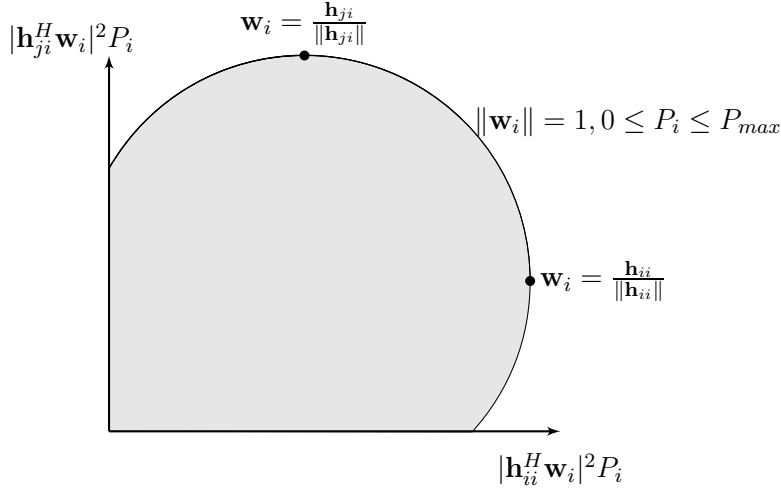
\begin{figure}
\begin{center}
\begin{tikzpicture}[node distance=5cm,auto,>=latex']

\draw[fill=black!10] (0,3.5) arc (150:-42:3cm) -- (0,0);

\draw[->,thick] (0,0) -- (6,0);
\node at (6,-0.5) {$|\mathbf{h}_{ii}^H \mathbf{w}_i|^2 P_i$};

\draw[->,thick] (0,0) -- (0,5);
\node at (-0.5,5.2) {$|\mathbf{h}_{ji}^H \mathbf{w}_i|^2 P_i$};

\draw (0,3.5) arc (150:90:3cm)  node[circle, draw=black, fill=black, inner sep=0.25ex]  {}  node [above] {$\mathbf{w}_i=\frac{\mathbf{h}_{ji}}{\|\mathbf{h}_{ji}\|}$} ;
\draw (0,3.5) arc (150:45:3cm)  node [right] {$\|\mathbf{w}_i\|=1, 0 \leq P_i \leq P_{max}$};
\draw (0,3.5) arc (150:0:3cm) node[circle, draw=black, fill=black, inner sep=0.25ex]  {}  node [right]{$\mathbf{w}_i=\frac{\mathbf{h}_{ii}}{\|\mathbf{h}_{ii}\|}$};

\end{tikzpicture}
\caption{The received  power region for Tx $i$. \label{fig:powergain}}

\end{center}
\end{figure}

The importance and relevance of the received power region can be summarized in the following.

\begin{itemize}
\item \emph{The boundary of received power region and the Pareto boundary}: The received power region is a convex and compact region with respect to the received power values. The Pareto boundary of the MISO-IC-SUD with linear pre-coding, the NN region here, is shown to be attained by the received power values \emph{on the boundary of the received power region} \cite{Jorswieck2008}. 
\item \emph{Monotonicity of rates}: The rate metrics defined in \eqref{eqt:ch4_rate_def} are either \emph{monotonically increasing or decreasing} with the channel powers. The optimization of such rates can be simplified by first computing the optimized channel powers and then the corresponding beamforming vectors that achieve such channel powers. 
\end{itemize}

Now we define the received power region achieved by beamforming vectors inside the Pareto boundary candidate set, $\Omega$,
\begin{equation}
\Phi(\Omega)= \bigg\{ \bigg( | \mathbf{h}_{11}^H \mathbf{w}_1|^2 P_1, |\mathbf{h}_{21}^H \mathbf{w}_1|^2 P_1, |\mathbf{h}_{22}^H \mathbf{w}_2|^2 P_2, |\mathbf{h}_{12}^H \mathbf{w}_2|^2 P_2 \bigg): \left( \mathbf{w}_1, \mathbf{w}_2, P_1, P_2\right) \in \Omega \bigg\}.
\end{equation}

Immediately, we have the following relations: the Pareto boundary candidate set achieves a received power region which is a subset of the Cartesian product of the received power regions for Rx 1 and 2, $\Phi_1$ and $\Phi_2$,
\begin{equation}
\Phi(\Omega) \subset \Phi_1 \times \Phi_2.
\end{equation}  The beamforming vectors and power allocations in $\Phi_1 \times \Phi_2$ contribute to the whole achievable rate region whereas the tuples in $\Phi(\Omega)$ only attain the Pareto boundary. This means that if we know $\Phi(\Omega)$, we can achieve the Pareto boundary without searching over all beamforming vectors in the remaining space in $\Phi_1 \times \Phi_2$. This reduces the search space from $\Phi_1 \times \Phi_2$ significantly.

In the following sections, we compute $\Omega^{nn},\Omega^{dn},\Omega^{nd},\Omega^{dd}$ which are candidate sets of the Pareto boundary of the corresponding regions: NN, DN, ND and DD. We define a candidate set of the overall Pareto boundary, $\Omega$, as  the union of the candidate sets mentioned above:
\begin{equation}\label{eqt:candidateset}
\Omega=  \bigcup_{x,y \in \{ n,d \} }\Omega^{xy}.
\end{equation}

\section{Pareto boundary characterization}

\subsection{ Pareto boundary characterization in the ND region}\label{section:pareto}

With decoding structure $\mathbf{R}^{nd}$, Rx 1 treats interference as noise and Rx 2 decodes and subtracts the interference signal from the received signal before decoding the desired signal.
\begin{Thm}\label{thm:pareto_rnd}
The Pareto boundary $\mathcal{B}(\mathcal{R}^{nd})$ is attained by candidate set $\Omega^{nd}$
\begin{equation}\label{eqt:omega_nd}
\Omega^{nd}= \left\{ \mathcal{W}_1, \mathcal{W}_2, P_1=P_{max}, 0 \leq P_2 \leq P_{max} \right\}
\end{equation} where $\mathcal{W}_1, \mathcal{W}_2$ defined in \eqref{eqt:bf_set_w}, are sets of beamforming vectors composed of linear combinations of two channel vectors; to attain the Pareto boundary, Tx 1 transmits with full power $P_{max}$ whereas Tx 2 transmits with less than full power $P_2 \leq P_{max}$.
\begin{figure*}[!t]
 \begin{equation}\label{eqt:bf_set_w}
 \mathcal{W}_i= \left\{ \mathbf{w}_i : \mathbf{w}_i= \sqrt{\lambda_i} \frac{\Pi_{ji}\mathbf{h}_{ii}}{\|\Pi_{ji}\mathbf{h}_{ii} \|} + \sqrt{1-\lambda_i} \frac{\Pi_{ji}^\perp \mathbf{h}_{ii}}{\| \Pi_{ji}^\perp \mathbf{h}_{ii}\|} \; ; \; 0 \leq \lambda_i \leq 1 \right\}, \hspace{0.5cm} i,j=1,2, i\neq j.
\end{equation}
\end{figure*}
\end{Thm}
\vspace{0.5cm}
\begin{proof}
See Appendix \ref{app:rnd_pareto}. 
\end{proof}

In the ND region, Rx 1 treats interference as noise and Rx 2 decodes interference. As described by Thm. \ref{thm:pareto_rnd}, the Pareto optimal transmit power for Tx 1 is to transmit at full power $P_{max}$ and less than full power for Tx 2. The interpretation is that Tx 1's transmit power does not affect the rate performance of Rx 2 as the interference from Tx 1 is decoded and removed. On the other hand, Tx 2 is not advised to transmit at full power because its increase of power will increase the interference power at Rx 1 and hence reduce the achievable rate of Rx 1.

The Pareto optimal transmit beamforming vectors are parameterized in the sets $\mathcal{W}_1, \mathcal{W}_2$ as positive linear combinations of two orthogonal vectors. These two vectors are the desired channel projection onto the span and the null space of the interference channel. As shown in Fig. \ref{fig:pow_gain_region_pareto_rnd}, the vectors in $\mathcal{W}_1$ and $\mathcal{W}_2$ are represented by blue regions. The blue regions cover from the point of zero interference power (Point A in Fig. \ref{fig:pow_gain_region_pareto_rnd}) to the point of maximum desired channel power (Point B) and the point of maximum interference power (Point C). Moving from point A to B and C on the Pareto boundary in Fig. \ref{fig:pow_gain_region_pareto_rnd}, the interference power increases monotonically. On the left figure of Figure \ref{fig:pow_gain_region_pareto_rnd}, we show the received power region of Tx 1, the channel powers between A and B correspond to a strong desired channel power of Tx 1 and a relatively small interference channel power from Tx 1 to Rx 2. These points may attain the Pareto boundary if the desired channel power of Tx 2 is weak. On the other hand, if Tx 2's desire channel power is large, the interference power from Tx 1 to Rx 2 must be increased to increase to interference rate $T_1$ which limits rate rate of $R_1$ as $R_1=\min(T_1, D_1)$. This is a novel concept comparing to the conventional single user decoding interference channel, the increase in interference power here is beneficial as it facilitates interference decoding and removal.

\begin{figure}
\begin{center}
\begin{tikzpicture}[node distance=5cm,auto,>=latex']
\draw (4.5,0) arc (-40:137:3cm) -- (0,0);

\draw[->,thick] (0,0) -- (6,0);
\node at (6,-0.5) {$|\mathbf{h}_{11}^H \mathbf{w}_1|^2 P_1$};

\draw[->,thick] (0,0) -- (0,5);
\node at (-0.5,5.2) {$|\mathbf{h}_{21}^H \mathbf{w}_1|^2 P_1$};

\draw[ultra thick, draw=black!20] (4.5,0) arc (-40:90:3cm)  node[circle, draw=black, fill=black, inner sep=0.25ex]  {}  node [above] {$\mathbf{w}_1=\frac{\mathbf{h}_{21}}{\|\mathbf{h}_{21}\|}$} node[below] {C};
\draw[ draw=black!20] (4.5,0) arc (-40:0:3cm) node[circle, draw=black, fill=black, inner sep=0.25ex]  {}  node [right]{$\mathbf{w}_1=\frac{\mathbf{h}_{11}}{\|\mathbf{h}_{11}\|}$}
node [left] {B};
\draw[ draw=black!20] (4.5,0) node[circle, draw=black, fill=black, inner sep=0.25ex] {} node [left=2pt, above=1pt] {A};

\draw ++(8,0);
\draw (12.5,0) arc (-40:137:3cm) -- (8,0);

\draw[ultra thick, draw=black!20,fill=black!20] (8,0) -- (12.5,0) arc (-40:90:3cm)  node[circle, draw=black, fill=black, inner sep=0.25ex]  {}  node [above] {$\mathbf{w}_2=\frac{\mathbf{h}_{12}}{\|\mathbf{h}_{12}\|}$} node[below] {C};
\draw[ draw=black!20] (12.5,0) arc (-40:0:3cm) node[circle, draw=black, fill=black, inner sep=0.25ex]  {}  node [right]{$\mathbf{w}_2=\frac{\mathbf{h}_{22}}{\|\mathbf{h}_{22}\|}$}
node [left] {B};
\draw[ draw=black!20] (12.5,0) node[circle, draw=black, fill=black, inner sep=0.25ex] {} node [left=2pt, above=1pt] {A};

\draw[->,thick] (8,0) -- (8,5);
\node at (7.5,5.2) {$|\mathbf{h}_{12}^H \mathbf{w}_2|^2 P_2$};

\draw[->,thick,draw=black] (8,0) -- (14,0);
\node at (14,-0.5) {$|\mathbf{h}_{22}^H \mathbf{w}_2|^2 P_2$};

\node at (7,-1.5) {$\Omega^{nd}=\left\{P_1=P_{max}, \mathbf{w}_1 \in \mathcal{W}_1, \mathbf{w}_2 \in \mathcal{W}_2, 0 \leq P_2 \leq P_{max} \right\}$};

\end{tikzpicture}
\caption{The graphical illustration of the candidate set $\Omega^{nd}$ as the shaded area which is a subset of the received power regions $\Phi_i$.\label{fig:pow_gain_region_pareto_rnd}}

\end{center}
\end{figure}
In the next section, we investigate the Pareto boundary attaining beamforming vectors in the DD region.

\subsection{The Pareto boundary characterization in the DD region}\label{section:pareto_rdd}
With decoding structure $\mathbf{R}^{dd}$, both Rx's decode interference. The Pareto boundary attaining solutions are:

\begin{Thm}\label{Thm:ParetoRdd}
The Pareto boundary $\mathcal{B}(\mathcal{R}^{dd})$ is attained by candidate set 
\begin{equation}
\Omega^{dd}=\left\{  \mathbf{w}_1\in \mathcal{V}_1, \mathbf{w}_2 \in \mathcal{V}_2, 0 \leq P_1, P_2 \leq P_{max}\right\},
\end{equation} with Pareto boundary attaining beamforming vectors composed of two orthogonal channel vectors, specifically for $i=1,2$,
\begin{equation}
\mathcal{V}_i= \left\{ \mathbf{w} \in \mathcal{S}: \mathbf{w}=\sqrt{\lambda_i} \frac{\Pi_{ii}\mathbf{h}_{ji}}{\|\Pi_{ii}\mathbf{h}_{ji} \|} + \sqrt{1- \lambda_i} \frac{\Pi_{ii}^\perp \mathbf{h}_{ji}}{\|\Pi_{ii}^\perp \mathbf{h}_{ji} \|}, 0 \leq \lambda_i \leq 1\right\},
\end{equation} and both Txs transmit at less than full power.
\end{Thm}
\begin{proof}
See Appendix \ref{app:ParetoRdd}.
\end{proof}
\begin{Remark}
Note that $\mathcal{W}_i$ in Thm. \ref{thm:pareto_rnd} and $\mathcal{V}_i$ defined here are different candidate sets. In particular, $\mathcal{W}_i$ is a set of vectors that are  the positive linear combinations of $\Pi_{ji}\mathbf{h}_{ii}$ and $\Pi_{ji}^\perp \mathbf{h}_{ii}$ whereas $\mathcal{V}_i$ is a set of vectors that are the positive linear combinations of $\Pi_{ii} \mathbf{h}_{ji}$ and $\Pi_{ii}^\perp \mathbf{h}_{ji}$. This difference is shown graphically in Fig. \ref{fig:pow_gain_region_pareto_rnd} and \ref{fig:pow_gain_region_pareto_rdd}.
\end{Remark}

In the DD region, both Tx 1 and 2 decode interference and their choice of actions are symmetric. As described by Thm. \ref{Thm:ParetoRdd}, the Pareto optimal transmit power for Tx 1 and 2 are to transmit at  less than full power. The Pareto optimal transmit beamforming vectors are parameterized in the sets $\mathcal{V}_1, \mathcal{V}_2$ as positive linear combinations of two orthogonal vectors. These two vectors are the interference channel projection onto the span and the null space of the desired channel. As shown in Fig. \ref{fig:pow_gain_region_pareto_rdd}, the vectors in $\mathcal{V}_1$ and $\mathcal{V}_2$ are represented by blue regions. The blue regions cover from the point where the point of maximum desired channel power (Point B) to the point of maximum interference power (Point C) and the point of zero desired channel power (Point D). 

Notice that the points where the interference channel powers are zero is not Pareto optimal. It is because minimizing the interference power in the DD region makes decoding interference more difficult. This choice of action is not Pareto optimal. From the received power region representation in Fig. \ref{fig:pow_gain_region_pareto_rdd}, we can see that the interference channel power should be maximized despite the values of the direct channel gain. It means that for each achievable desired channel power value, the interference channel power should be increased for easy interference decoding and removal.

\begin{figure}
\begin{center}
\begin{tikzpicture}[node distance=5cm,auto,>=latex']
\draw (0,4.5) arc (137:-55:3cm) -- (0,0);

\draw[->,thick] (0,0) -- (6,0);
\node at (6,-0.5) {$|\mathbf{h}_{11}^H \mathbf{w}_1|^2 P_1$};


\draw[ draw=black!20,fill=black!20] (0,0) -- (0,4.5) arc (137:0:3cm) node[circle, draw=black, fill=black, inner sep=0.25ex]  {}  node [right]{$\mathbf{w}_1=\frac{\mathbf{h}_{11}}{\|\mathbf{h}_{11}\|}$}
node [left] {B};
\draw[ultra thick, draw=black!20] (0,4.5) arc (137:90:3cm)  node[circle, draw=black, fill=black, inner sep=0.25ex]  {}  node [above] {$\mathbf{w}_1=\frac{\mathbf{h}_{21}}{\|\mathbf{h}_{21}\|}$} node[below] {C};
\draw[ draw=black!20] (0,4.5) node[circle, draw=black, fill=black, inner sep=0.25ex] {} node [left] {D};

\draw[->,thick] (0,0) -- (0,5);
\node at (-0.5,5.2) {$|\mathbf{h}_{21}^H \mathbf{w}_1|^2 P_1$};

\draw (8,4.5) arc (137:-55:3cm) -- (8,0);

\draw[ draw=black!20,fill=black!20] (8,0) -- (8,4.5) arc (137:0:3cm) node[circle, draw=black, fill=black, inner sep=0.25ex]  {}  node [right]{$\mathbf{w}_1=\frac{\mathbf{h}_{11}}{\|\mathbf{h}_{11}\|}$}
node [left] {B};
\draw[ultra thick, draw=black!20] (8,4.5) arc (137:90:3cm)  node[circle, draw=black, fill=black, inner sep=0.25ex]  {}  node [above] {$\mathbf{w}_1=\frac{\mathbf{h}_{21}}{\|\mathbf{h}_{21}\|}$} node[below] {C};
\draw[ draw=black!20] (8,4.5) node[circle, draw=black, fill=black, inner sep=0.25ex] {} node [left] {D};
\draw[->,thick] (8,0) -- (8,5);
\node at (7.5,5.2) {$|\mathbf{h}_{12}^H \mathbf{w}_2|^2 P_2$};

\draw[->,thick,draw=black] (8,0) -- (14,0);
\node at (14,-0.5) {$|\mathbf{h}_{22}^H \mathbf{w}_2|^2 P_2$};

\node at (7,-1.5) {$\Omega^{dd}=\left\{0 \leq P_1,P_2 \leq P_{max}, \mathbf{w}_1 \in \mathcal{V}_1,\mathbf{w}_2 \in \mathcal{V}_2 \right\}$};

\end{tikzpicture}
\caption{The graphical illustration of candidate set $\Omega^{dd}$ as the shaded area which is a subset of the received power regions $\Phi_i$.\label{fig:pow_gain_region_pareto_rdd}}

\end{center}
\end{figure}

\subsection{The Pareto boundary characterization}
From Thm. \ref{thm:pareto_rnd} and \ref{Thm:ParetoRdd}, we have presented the Pareto boundary characterization in ND and DD region. We can easily obtain the candidate set  $\Omega^{dn}$ by reversing the role of Tx 1 and 2 from $\Omega^{nd}$ in Thm. \ref{thm:pareto_rnd}. Also, the candidate set $\Omega^{nn}$ is shown to be the following \cite{Jorswieck2008}:
\begin{equation}
\Omega^{nn}= \left\{ \mathcal{W}_1, \mathcal{W}_2,   P_1 = P_2 = P_{max} \right\}
\end{equation} where $\mathcal{W}_1, \mathcal{W}_2$ are defined in Thm. \ref{thm:pareto_rnd}.

By definition in \eqref{eqt:candidateset} , the candidate set of the Pareto boundary $\mathcal{B}(\mathcal{R})$ is the union of the candidate sets in each decoding region. Hence, we have characterized the Pareto boundary $\mathcal{B}(\mathcal{R})$. In the candidate sets of the Pareto boundary, the beamforming vectors are parameterized with positive real scalars $0 \leq \lambda_1, \lambda_2 \leq 1$. By varying $\lambda_1, \lambda_2$ from zero to one and $P_2$ from zero to $P_{max}$, we obtain all beamforming vectors that may attain the the boundary in each decoding region and in turn the overall Pareto boundary. Intuitively, it means that the boundary attaining beamforming vectors in each decoding region exist only in a two-dimensional subspace, spanned by the direct channel and the interference channel, in a $N$-dimensional signal space.

As a direct application of the Pareto boundary characterization, we characterize the maximum sum rate point in the following section. Since the maximum sum rate point is always on the Pareto boundary, the candidate set of the maximum sum rate point is therefore a subset of the candidate set derived above. We reduce the size of the candidate set by eliminating beamforming vectors in the candidate set that achieve a smaller sum rate than other vectors in the set. 

\section{The maximum sum rate point characterization}\label{section:sumrate}
In this section, we characterize the candidate sets of the maximum sum rate point by first illustrating that Txs. should always transmit with full power, in Section \ref{section:sumrate_fullpow}. Then, we study the candidates sets of the boundaries of ND and DD regions by eliminating vectors that attain a smaller sum rate than other vectors in the candidate sets and obtain the candidate sets of maximum sum rate point in $\mathcal{B}(\mathcal{R}^{nd})$ and $\mathcal{B}(\mathcal{R}^{dd})$ respectively, in Section \ref{section:sumrate_rnd} and \ref{section:sumrate_rdd}.

\subsection{Full power transmission}\label{section:sumrate_fullpow}
We observe that the maximum sum rate point is attained by maximum transmit power at each transmitter. To see this, we combine the power constraints and beamformer norm constraints:
\begin{equation}
\|\mathbf{w}_i \|^2 \leq P_i.
\end{equation}
Assume that the sum rate optimal beamformer is not transmitting at maximum power: $\| \mathbf{w}_i\|^2=p < P_i$. We can choose a beamformer $\mathbf{w}_i'= \mathbf{w}_i + \epsilon e^{j \phi} \Pi_{ji}^\perp \mathbf{h}_{ii}$ where $\epsilon$ is  chosen such that $\|\mathbf{w}_i' \|^2=P_i$ and $\phi= \arg(\mathbf{h}_{ii}^H \mathbf{w}_i)$. Notice that $|\mathbf{h}_{ii}^H \mathbf{w}_i'|^2 \geq |\mathbf{h}_{ii}^H \mathbf{w}_i|^2$ and $|\mathbf{h}_{ji}^H \mathbf{w}_i'|^2=|\mathbf{h}_{ji}^H \mathbf{w}_i|^2$. Or, we can choose $\mathbf{w}_i''=\mathbf{w}_i + \epsilon' e^{j \phi'} \Pi_{ii}^\perp \mathbf{h}_{ji}$ with $\phi'=\arg(\mathbf{h}_{ji}^H \mathbf{w}_i)$ to increase $|\mathbf{h}_{ji}^H \mathbf{w}_i|^2$
and keep $|\mathbf{h}_{ii}^H \mathbf{w}_i|^2$ constant. Thus, it contradicts that $\mathbf{w}_i$ is on the Pareto boundary. From now on, we set $P_i=P_{max}, i=1,2.$ Note that the argument above is limited to non-parallel channels, for parallel channels (e.g. $\mathbf{h}_{ji}^H \mathbf{h}_{ii}=0$), it reduces to SISO-IC where the maximum sum rate point is attained by one Tx transmitting with full power whereas the other Txs. transmit at less than full power \cite{Bandemer2008}.

In the following sections, we characterize the candidate sets that attain the maximum sum rate point. Note that, the candidate sets attaining the maximum sum rate point is a strict subset of those attaining the Pareto boundary. The sum rate metric does not distinguish between Tx 1 and 2's rate and therefore we can identify a much smaller candidate set as illustrated in the following sections. This is particularly useful for system optimization which does not put emphasis on user fairness.

Note that the computation of the global optimal solution 
\begin{equation}
\boldsymbol{\omega}^*=\argmax_{\mathbf{w}_1,\mathbf{w}_2 \in \mathcal{S}} \bar{R}^{nd}(\mathbf{w}_1, \mathbf{w}_2)
\end{equation} is in general NP-hard \cite{Liu2010}, even though there exist channels for which the solution is easily obtained (e.g. orthogonal channels). Here, we would like to reduce the search space and characterize the solutions set. 

\subsection{The maximum sum rate point characterization in the ND region}\label{section:sumrate_rnd}

\begin{Thm}\label{thm:sumrate_rnd}
The candidate set of maximum sum rate $\bar{R}^{nd}$ denoted as $\tilde{\Omega}^{nd}$, hence $\boldsymbol{\omega}^* \subset \tilde{\Omega}^{nd}\subset \Omega^{nd}$, is given by
\begin{equation}
\tilde{\Omega}^{nd}=\left\{\tilde{\mathcal{W}}_1, \tilde{\mathcal{W}}_2, P_{max}, P_{max} \right\}
\end{equation} where $\Omega^{nd}$ is the candidate set of  Pareto boundary $\mathcal{B}(\mathcal{R}^{nd})$ in \eqref{eqt:omega_nd}. In particular, $\tilde{\mathcal{W}}_1$  is the following set with cardinality three:
\begin{equation}
\tilde{\mathcal{W}}_1= \left\{ \frac{\mathbf{h}_{11}}{||\mathbf{h}_{11}||}, \frac{\mathbf{h}_{21}}{||\mathbf{h}_{21} ||}, \mathbf{w}_1(\lambda_1^{(b)}) \right\}
\end{equation} with $\lambda_1^{(b)}=\frac{c_1 ||\Pi_{21}^\perp \mathbf{h}_{11} ||^2}{c_2||\mathbf{h}_{21}||^2- 2 \sqrt{c_1c_2} |\mathbf{h}_{21}^H\mathbf{h}_{11} | + c_1 ||\mathbf{h}_{11} ||^2}$. The candidate set $\tilde{\mathcal{W}}_2$ is a set of beamforming vectors characterized by a parameter $\lambda_2$ in a smaller range than the range in $\mathcal{W}_2$:
\begin{equation}
\tilde{\mathcal{W}}_2= \left\{ \mathbf{w}_2 \in \mathcal{S}: \mathbf{w}_2= \sqrt{\lambda_2} \frac{\Pi_{12}\mathbf{h}_{22}}{\|\Pi_{12}\mathbf{h}_{22} \|} + \sqrt{1-\lambda_2} \frac{\Pi_{12}^\perp \mathbf{h}_{22}}{\| \Pi_{12}^\perp \mathbf{h}_{22} \|} \; ; \; \lambda_2^{(b)} \leq \lambda_2 \leq \lambda_2^{\mrt} \right\}
\end{equation} 
where $\lambda_2^{\mrt}=\frac{|\mathbf{h}_{12}^H \mathbf{h}_{22}|}{||\mathbf{h}_{12} || ||\mathbf{h}_{22} ||}$ is a parameter that gives the beamforming solution towards channel $\mathbf{h}_{22}$ and 
$\mathbf{w}_2(\lambda_2^{(b)})=\frac{\tilde{b}}{\sqrt{\tilde{a}+\tilde{b}}}\mathbf{v}_a + \frac{e^{j \phi}\tilde{a}}{\sqrt{\tilde{a}+\tilde{b}}} \mathbf{v}_b$ for some eigenvectors $\mathbf{v}_a, \mathbf{v}_b$ and positive scalars $\tilde{a}, \tilde{b}$. The vectors $\mathbf{v}_{a}, \mathbf{v}_b$ are the most and least dominant eigenvectors of the matrix $\mathbf{S}=\mathbf{h}_{22}\mathbf{h}_{22}^H -\frac{g_{21}}{g_{11}} \mathbf{h}_{12}\mathbf{h}_{12}^H$.
\end{Thm}
\begin{proof}
See Appendix \ref{app:sumrate_rnd}.
\end{proof}

\begin{Remark}
 Note that for some channel realizations and chosen $\mathbf{w}_2$, $\mathbf{w}_1(\lambda_1^{(b)})$ may be equal to the maximum ratio transmission solutions $\frac{\mathbf{h}_{11}}{\| \mathbf{h}_{11}\|}$ or $\frac{\mathbf{h}_{21}}{\|\mathbf{h}_{21} \|}$. But we distinguish between them in the candidate sets to illustrate that for most channel realizations and $\mathbf{w}_2$, $\mathbf{w}_1(\lambda_1^{(b)}) \notin \left\{ \frac{\mathbf{h}_{11}}{\| \mathbf{h}_{11}\|} , \frac{\mathbf{h}_{21}}{\| \mathbf{h}_{21}\|} \right\}$.
 
 It is interesting to see that the sum rate optimal beamforming vector of Tx 1 is either the beamforming vector towards the desired channel or the beamforming vector towards the interference channel or a beamforming vector that balances the interference decoding rate $T_i$ and the treating interference as noise rate $D_i$ in a weighted manner with weights $c_1, c_2$ which depend on the choice of the beamforming vector at Tx 2.
\end{Remark}

Comparing the candidate set of the Pareto boundary to the candidate set of the maximum sum rate point of the ND region, namely $\tilde{\Omega}^{nd}$ and $\Omega^{nd}$,  we observe the following:
\begin{itemize}
\item For each Tx $i$, the candidate set $\tilde{\Omega}^{nd}$ consists of only \emph{three} closed-form beamforming vectors whereas $\Omega^{nd}$ consists of a set of beamforming vectors characterized by a real-valued parameter spanned between zero and one, as shown in Fig. \ref{fig:sumrate_pow_gain_rnd}. 
\item An interesting question rises: what are the conditions of each of these potential sum rate optimal solutions being sum rate optimal? We give the discussion in Section \ref{section:mrtopt}.
\end{itemize}

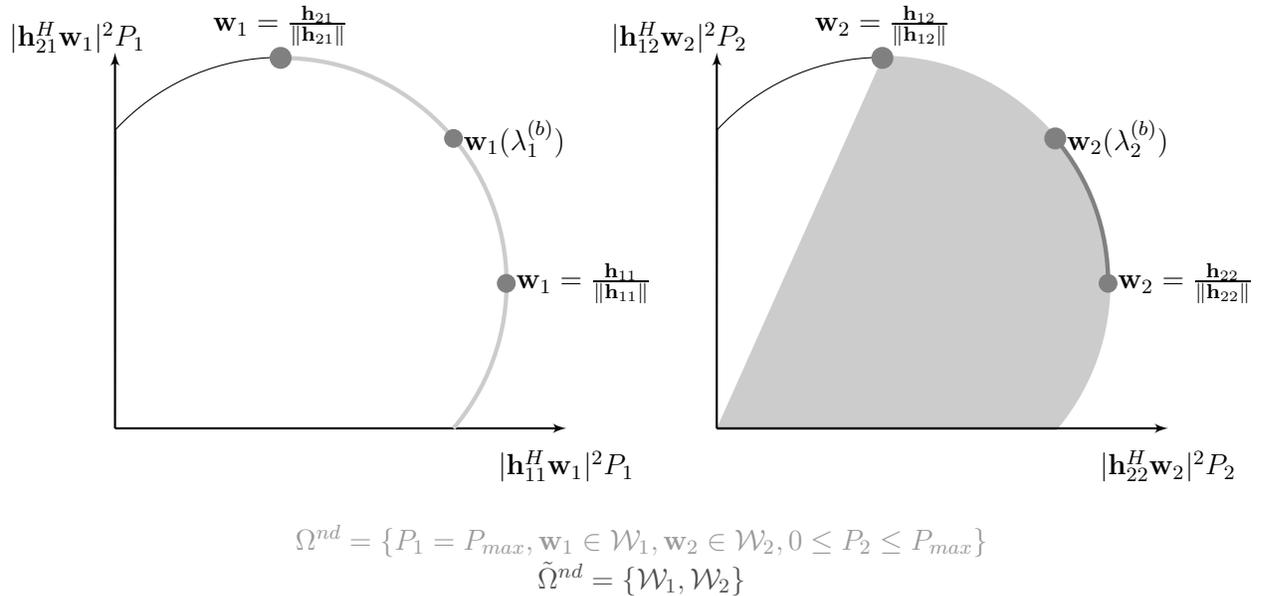
\begin{figure}
\begin{center}
\begin{tikzpicture}[node distance=5cm,auto,>=latex']
\draw (4.5,0) arc (-40:137:3cm) -- (0,0);

\draw[->,thick] (0,0) -- (6,0);
\node at (6,-0.5) {$|\mathbf{h}_{11}^H \mathbf{w}_1|^2 P_1$};

\draw[->,thick] (0,0) -- (0,5);
\node at (-0.5,5.2) {$|\mathbf{h}_{21}^H \mathbf{w}_1|^2 P_1$};

\draw[ultra thick, draw=black!20] (4.5,0) arc (-40:90:3cm)  node[circle, draw=black!50, fill=black!50, inner sep=0.5ex]  {}  node [above] {$\mathbf{w}_1=\frac{\mathbf{h}_{21}}{\|\mathbf{h}_{21}\|}$} ;
\draw[ draw=black!20] (4.5,0) arc (-40:40:3cm) node[circle, draw=black!50, fill=black!50, inner sep=0.5ex] {} node [right] {$\mathbf{w}_1(\lambda_1^{(b)})$};
\draw[ draw=black!20] (4.5,0) arc (-40:0:3cm) node[circle, draw=black!50, fill=black!50, inner sep=0.5ex]  {}  node [right]{$\mathbf{w}_1=\frac{\mathbf{h}_{11}}{\|\mathbf{h}_{11}\|}$};

\draw ++(8,0);
\draw (12.5,0) arc (-40:137:3cm) -- (8,0);

\draw[ultra thick, draw=black!20,fill=black!20] (8,0) -- (12.5,0) arc (-40:90:3cm)  node[circle, draw=black!50, fill=black!50, inner sep=0.5ex]  {}  node [above] {$\mathbf{w}_2=\frac{\mathbf{h}_{12}}{\|\mathbf{h}_{12}\|}$} ;
\draw[ draw=black!20] (12.5,0) arc (-40:0:3cm) node[circle, draw=black!50, fill=black!50, inner sep=0.5ex]  {}  node [right]{$\mathbf{w}_2=\frac{\mathbf{h}_{22}}{\|\mathbf{h}_{22}\|}$} coordinate (t);
\draw[ draw=black!50, ultra thick] (t) arc (0:40:3cm) node[circle, draw=black!50, fill=black!50, inner sep=0.5ex] {} node [right] {$\mathbf{w}_2(\lambda_2^{(b)})$} ;

\draw[->,thick] (8,0) -- (8,5);
\node at (7.5,5.2) {$|\mathbf{h}_{12}^H \mathbf{w}_2|^2 P_2$};

\draw[->,thick,draw=black] (8,0) -- (14,0);
\node at (14,-0.5) {$|\mathbf{h}_{22}^H \mathbf{w}_2|^2 P_2$};

\node at (7,-1.5) {\color{black!40} $\Omega^{nd}=\left\{P_1=P_{max}, \mathbf{w}_1 \in \mathcal{W}_1, \mathbf{w}_2 \in \mathcal{W}_2, 0 \leq P_2 \leq P_{max} \right\}$};
\node at (7,-2) {\color{black!70}$\tilde{\Omega}^{nd}=\left\{ \mathcal{W}_1,\mathcal{W}_2 \right\}$};

\end{tikzpicture}
\caption{The illustration of the candidate set of the maximum sum rate point of ND region  in dark grey and the candidate set of the Pareto boundary $\mathcal{B}(\mathcal{R}^{nd})$ in light grey. The cardinality of the candidate set for $\mathbf{w}_1$ of the maximum sum rate point is only three, conditioned on $\mathbf{w}_2$. \label{fig:sumrate_pow_gain_rnd}}

\end{center}
\end{figure}

\subsection{The maximum sum rate point characterization in the DD region}\label{section:sumrate_rdd}
In this section, we compute the candidate set that attains the maximum sum rate point of the DD region. \begin{Thm}\label{thm:rddsumrate}
The  candidate set of the maximum sum rate in $\mathbf{R}^{dd}$ is
\begin{equation}
\tilde{\Omega}^{dd}=\left\{ \mathcal{V}^{dd}_1, \mathcal{V}^{dd}_2, P_{max}, P_{max}\right\} 
\end{equation}
where for user $i=1,2$, the sum rate optimal beamforming vectors are either a linear combination of two orthogonal vectors or maximizing the desired channel power or a specific vector:
\begin{equation}
\mathcal{V}^{dd}_i=\left\{ \tilde{\mathcal{V}}_i, \frac{\mathbf{h}_{ii}}{\|\mathbf{h}_{ii} \|}, \mathbf{w}_i(\lambda_i^A)\right\}
\end{equation}
\begin{equation}
\tilde{\mathcal{V}}_i=\left\{ \mathbf{w}_i: \sqrt{\lambda_i} \frac{\Pi_{ii}\mathbf{h}_{ji}}{\|\Pi_{ii}\mathbf{h}_{ji} \|}+ \sqrt{1-\lambda_i} \frac{\Pi_{ii}^\perp \mathbf{h}_{ji}}{\|\Pi_{ii}^\perp \mathbf{h}_{ji} \|} , \lambda_i^{A} \leq \lambda_i \leq \lambda_i^{\mrt} \right\} 
\end{equation} 
where $\lambda_i^{\mrt}=\frac{|\mathbf{h}_{ii}^H\mathbf{h}_{ji}|^2}{\|\mathbf{h}_{ii}\|^2 \|\mathbf{h}_{ji}\|^2}$ and $\lambda_i^{A}=\frac{\|\Pi_{ii}^\perp \mathbf{h}_{ji} \|}{\|\mathbf{h}_{ji}\|^2 + (1+g_{jj})\|\mathbf{h}_{ii} \|^2-2 |\mathbf{h}_{ii}^H \mathbf{h}_{ji}| \sqrt{1+g_{jj}}}$. 
\end{Thm}
\begin{proof}
see Appendix \ref{appendix:sumraterdd}.
\end{proof}
\begin{Remark}
Note that $\mathbf{w}_i(\lambda_i^A)$ may not be a element of $\mathcal{V}_i$ because $\lambda^A_i$ may not be smaller than $\lambda_i^{mrt}$ and in this case $\tilde{\mathcal{V}}_i$ is empty. The vector $\mathbf{w}_i(\lambda_i^A)$ is a beamforming vector that balance the interference decoding rate $T_i$ and the treating interference as noise rate $D_i$ in a weighted manner. See Appendix  \ref{appendix:sumraterdd} for more details.

In Fig. \ref{fig:sumrate_pow_gain_rdd}, we illustrate the reduction of the candidate set of the maximum sum rate point of the DD region, in red,  comparing to the candidate set of the Pareto boundary of the DD region, in blue. As shown in Fig. \ref{fig:sumrate_pow_gain_rdd}, the beamforming vectors in $\tilde{\mathcal{V}}$ achieve channel powers that are in the direction of minimizing the direct channel power while maximizing the interference channel power.
\end{Remark}

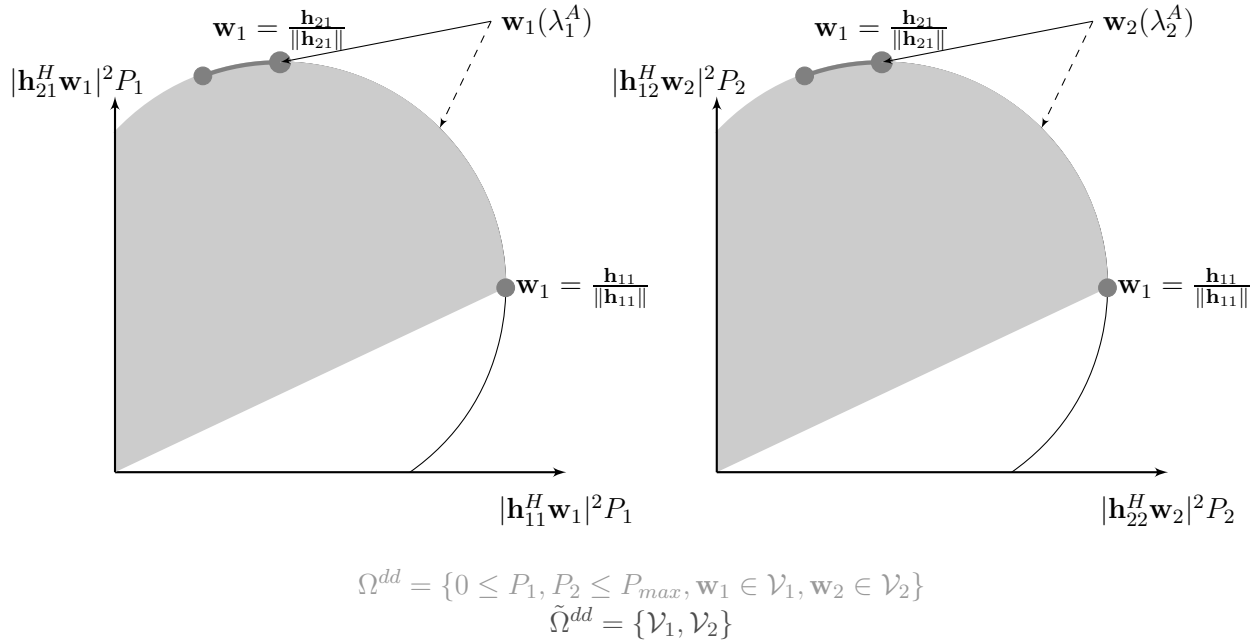
\begin{figure}
\begin{center}
\begin{tikzpicture}[node distance=5cm,auto,>=latex']
\draw (0,4.5) arc (137:-55:3cm) -- (0,0);

\draw[->,thick] (0,0) -- (6,0);
\node at (6,-0.5) {$|\mathbf{h}_{11}^H \mathbf{w}_1|^2 P_1$};


\draw[ draw=black!20,fill=black!20] (0,0) -- (0,4.5) arc (137:0:3cm) node[circle, draw=black!50, fill=black!50, inner sep=0.5ex]  {}  node [right]{$\mathbf{w}_1=\frac{\mathbf{h}_{11}}{\|\mathbf{h}_{11}\|}$};

\draw[ultra thick, draw=black!20] (0,4.5) arc (137:90:3cm)  node[circle, draw=black!50, fill=black!50, inner sep=0.5ex]  {}  node [above] {$\mathbf{w}_1=\frac{\mathbf{h}_{21}}{\|\mathbf{h}_{21}\|}$};

\draw[ draw=black!20] (0,4.5) arc (137:110:3cm) node [circle, draw=black!50, fill=black!50, inner sep=0.5ex] {} coordinate (s);
\draw[draw=black!50, ultra thick] (s) arc (110:90:3cm) coordinate (s1);
\draw[draw=black!20] (s1) arc (90:45:3cm) coordinate (s2);

\draw (5,6) node [right] {$\mathbf{w}_1(\lambda_1^A)$};
\draw[->] (5,6) -- (s1);
\draw[->, dashed] (5,6) -- (s2);

\draw[->,thick] (0,0) -- (0,5);
\node at (-0.5,5.2) {$|\mathbf{h}_{21}^H \mathbf{w}_1|^2 P_1$};

\draw (8,4.5) arc (137:-55:3cm) -- (8,0);

\draw[ draw=black!20,fill=black!20] (8,0) -- (8,4.5) arc (137:0:3cm) node[circle, draw=black!50, fill=black!50, inner sep=0.5ex]  {}  node [right]{$\mathbf{w}_1=\frac{\mathbf{h}_{11}}{\|\mathbf{h}_{11}\|}$};

\draw[ultra thick, draw=black!20] (8,4.5) arc (137:90:3cm)  node[circle, draw=black!50, fill=black!50, inner sep=0.5ex]  {}  node [above] {$\mathbf{w}_1=\frac{\mathbf{h}_{21}}{\|\mathbf{h}_{21}\|}$};

\draw[ draw=black!20] (8,4.5) arc (137:110:3cm) node[circle, draw=black!50, fill=black!50, inner sep=0.5ex]  {} coordinate (t);
\draw[draw=black!50, ultra thick] (t) arc (110:90:3cm) coordinate (t1);
\draw[draw=black!20] (t1) arc (90:45:3cm) coordinate (t2);

\draw (13,6) node [right] {$\mathbf{w}_2(\lambda_2^A)$};
\draw[->] (13,6) -- (t1);
\draw[->, dashed] (13,6) -- (t2);

\draw[->,thick] (8,0) -- (8,5);
\node at (7.5,5.2) {$|\mathbf{h}_{12}^H \mathbf{w}_2|^2 P_2$};

\draw[->,thick,draw=black] (8,0) -- (14,0);
\node at (14,-0.5) {$|\mathbf{h}_{22}^H \mathbf{w}_2|^2 P_2$};

\node at (7,-1.5) {\color{black!40} $\Omega^{dd}=\left\{0 \leq P_1,P_2 \leq P_{max}, \mathbf{w}_1 \in \mathcal{V}_1,\mathbf{w}_2 \in \mathcal{V}_2 \right\}$};
\node at (7,-2) {\color{black!70}$\tilde{\Omega}^{dd}=\left\{ \mathcal{V}_1,\mathcal{V}_2 \right\}$};

\end{tikzpicture}
\caption{The illustration of the candidate set of the maximum sum rate point of the DD region, in red, and the candidate set of the Pareto boundary $\mathcal{B}(\mathcal{R}^{dd})$ in blue. If $\lambda_i^A \leq \lambda_i^{\mrt}$, then the candidate set consists of the set $\tilde{V}_i$ and $\mathbf{w}_i =\frac{\mathbf{h}_{ii}}{\|\mathbf{h}_{ii} \|}$ as illstrated by the dark grey area in the figure. When $\lambda_i^A > \lambda_i^{\mrt}$, the candidate set becomes 3 beamforming vectors: $\frac{\mathbf{h}_{ii}}{\|\mathbf{h}_{ii} \|}, \mathbf{w}(\lambda_i^A)$ and $\mathbf{w}_i= \frac{\mathbf{h}_{ji}}{\|\mathbf{h}_{ji} \|}$.   \label{fig:sumrate_pow_gain_rdd}}
\end{center}
\end{figure}

To summarize, we obtain the candidate set of the maximum sum rate point in the ND and DD region, in Thm. \ref{thm:sumrate_rnd} and Thm. \ref{thm:rddsumrate} respectively. We can exchange the role of Tx 1 and 2 in Thm. \ref{thm:sumrate_rnd} to obtain the candidate set of maximum sum rate point in the DN region, $\tilde{\Omega}^{dn}$. For the NN region, the candidate set of the maximum sum rate point is identical to the candidate set of the Pareto boudary, $\Omega^{nn}$. Thus, we can have candidate set of the maximum sum rate point of MISO-IC-IDC as $\tilde{\Omega}$:
\begin{equation}
\tilde{\Omega}= \tilde{\Omega}^{nd} \bigcup \tilde{\Omega}^{dn} \bigcup \tilde{\Omega}^{dd} \bigcup \Omega^{nn}.
\end{equation}

In the next section, we apply the results obtained above: from the candidate set of the maximum sum rate point in different decoding structures, we identify the conditions in which the MRT strategies are sum rate optimal. Such strategies are attractive because of their simplicity and the MRT optimality conditions answer the following two interesting questions: \emph{When is selfishness sum rate optimal? When is interference amplification sum rate optimal?}

\section{MRT optimality conditions}\label{section:mrtopt}
In this section, we investigate the conditions in which the MRT strategies at both Tx 1 and 2 are sum rate optimal. For clarification, MRT strategies can mean two strategies, one to beamform to the direct channel $\mathbf{h}_{ii}$ and the other to beamform to the interference channel $\mathbf{h}_{ji}$. We characterize the MRT optimality conditions in terms of the \emph{separation} between the desired channel and the interference channel, $\theta_i$:
\begin{equation}
\theta_i = \cos^{-1} \left( \frac{|\mathbf{h}_{ji}^H \mathbf{h}_{ii} |}{\| \mathbf{h}_{ji}\| \|\mathbf{h}_{ii} \|}\right).
\end{equation}

\begin{Thm}
The MRT optimality conditions for decoding structure $\mathbf{R}^{nd}$ are:
\begin{itemize}
\item  $ \left( \frac{\mathbf{h}_{11}}{\|\mathbf{h}_{11} \|}, \frac{\mathbf{h}_{22}}{\|\mathbf{h}_{22} \|}\right)$  is optimal if and only if
  \begin{equation}
   \frac{c_1 \| \Pi_{21}^\perp \mathbf{h}_{11}\|^2}{c_2 \| \mathbf{h}_{21}\|^2 - 2 \sqrt{c_1 c_2} |\mathbf{h}_{21}^H \mathbf{h}_{11} | + c_1 \| \mathbf{h}_{11}\|^2}  < \cos^2(\theta_1) \leq \frac{(1+ \| \mathbf{h}_{22}\|^2 P_{max}) \| \mathbf{h}_{11}\|^2}{(1+ \| \mathbf{h}_{12} \|^2 \cos^2(\theta_2)P_{max} ) \| \mathbf{h}_{21}\|^2 } 
  \end{equation}
\item $ \left( \frac{\mathbf{h}_{21}}{\|\mathbf{h}_{21} \|}, \frac{\mathbf{h}_{22}}{\|\mathbf{h}_{22} \|}\right) $  is optimal if and only if
\begin{equation}
  \| \mathbf{h}_{21}\|^2 \leq (1+ \| \mathbf{h}_{22}\|^2 P_{max}) \frac{\| \Pi_{21} \mathbf{h}_{11}\|^2 }{1+ \| \mathbf{h}_{12}\|^2 \cos^2(\theta_2) P_{max}}
  \end{equation} where $c_1=\frac{P_{max}}{\| \mathbf{h}_{12}\|^2 \cos^2(\theta_2) P_{max} + 1}$ and $c_2=\frac{P_{max}}{\|\mathbf{h}_{22} \|^2 P_{max}+1}$.
  \end{itemize}
\end{Thm}
\begin{proof}
See Appendix \ref{app:mrtopt}.
\end{proof}

Now, we provide the MRT optimality conditions for $\mathbf{R}^{dd}$.
\begin{Thm}\label{thm:mrtopt_rdd}
The MRT optimality conditions for $\mathbf{R}^{dd}$ are:
\begin{eqnarray}
\mathbf{w}_i=\frac{\mathbf{h}_{ji}}{\|\mathbf{h}_{ji}\|} \text{ is optimal if } && \frac{g_{ij}}{g_{jj}}-1 \leq \|\mathbf{h}_{ii}\|^2 \cos^2(\theta_i)  \leq  \frac{\| \mathbf{h}_{ji}\|^2}{1+g_{jj}}\\
\mathbf{w}_i=\frac{\mathbf{h}_{ii}}{\|\mathbf{h}_{ii}\|} \text{ is optimal if } && (1+ g_{jj}) \|\mathbf{h}_{ii} \|^2 \leq \| \mathbf{h}_{ji}\|^2 \cos^2(\theta_i)
\end{eqnarray} for $i,j=1,2$ and $g_{ij}=|\mathbf{h}_{ij}^H \mathbf{w}_j|^2$. To be more specific:
\begin{eqnarray}
& &\left(\frac{\mathbf{h}_{11}}{\|\mathbf{h}_{11}  \|}, \frac{\mathbf{h}_{22} }{\|\mathbf{h}_{22}  \|} \right) \text{ is optimal if and only if} \left\{ \begin{aligned}
\cos^2(\theta_1) & \geq \frac{(1+ \| \mathbf{h}_{22}\|^2) \|\mathbf{h}_{11} \|^2}{ \| \mathbf{h}_{21}\|^2}\\
\cos^2(\theta_2) & \geq \frac{(1+ \| \mathbf{h}_{11}\|^2)\| \mathbf{h}_{22}\|^2}{ \|\mathbf{h}_{12} \|^2}.
\end{aligned} \right. \\
\nonumber && \left( \frac{\mathbf{h}_{21}}{\|\mathbf{h}_{21} \|}, \frac{\mathbf{h}_{12}}{\|\mathbf{h}_{12} \|}\right) \text{ is optimal if and only if} \\
&& \left\{ \begin{aligned}
\| \mathbf{h}_{12}\|^2 - \| \mathbf{h}_{22}\|^2 \cos^2 (\theta_2) &\leq \| \mathbf{h}_{11}\|^2 \| \mathbf{h}_{22}\|^2 \cos^2(\theta_1) \cos^2(\theta_2) &\leq \frac{\|\mathbf{h}_{21} \|^2 \| \mathbf{h}_{22}\|^2 \cos^2(\theta_2)}{ 1+ \|\mathbf{h}_{22} \|^2 \cos^2(\theta_2)}\\
\| \mathbf{h}_{21}\|^2 - \| \mathbf{h}_{11}\|^2 \cos^2 (\theta_1) &\leq \| \mathbf{h}_{11}\|^2 \| \mathbf{h}_{22}\|^2 \cos^2(\theta_1) \cos^2(\theta_2) &\leq \frac{\|\mathbf{h}_{12} \|^2 \| \mathbf{h}_{11}\|^2 \cos^2(\theta_1)}{ 1+ \|\mathbf{h}_{11} \|^2 \cos^2(\theta_1)}.
\end{aligned}\right.
\end{eqnarray}
\end{Thm}
\vspace{0.5cm}
\begin{proof}
see Appendix \ref{app:mrtopt_rdd}.
\end{proof}
%

\section{Simulation Results}\label{section:simulation}
In this section, we provide simulation results for the proposed parameterization. By varying the beamforming vectors and power allocation, according to the proposed parameterization, we plot the achievable rate region for each decoding structure for a particular channel realization in Section \ref{section:sim_rate_region}. The maximum sum rate point and the MRT points in each decoding structure are plotted on the corresponding achievable rate region. In Section \ref{section:sim_empirical_freq}, we compute the empirical frequency of MRT strategies in $\mathbf{R}^{nd}$ and $\mathbf{R}^{dd}$ averaged over 500 channel realizations. In Section \ref{section:sim_theta_ch}, we allow the channels to be correlated and we see that the sum rate optimal decoding structure changes with the strength of the interference channel, agreeing with the observations for SISO-IC.

\subsection{Achievable rate region and maximum sum rate point}\label{section:sim_rate_region}
In Fig. \ref{fig:ach_region}, we plot the achievable rate region of the decoding structure $\mathbf{R}^{nd}$, $\mathbf{R}^{dn}$, $\mathbf{R}^{dd}$ and $\mathbf{R}^{nn}$ in Fig. \ref{fig:rnd}, \ref{fig:rdn}, \ref{fig:rdd} and \ref{fig:rnn}, respectively. We assume $N=3$ transmit antennas and SNR=0dB. We exhaust $\lambda_1$ and $\lambda_2$ to take  20 values between zero and one, inclusively. For each pair of $(\lambda_1, \lambda_2)$, beamforming vectors $\mathbf{w}_1, \mathbf{w}_2$ are generated and the corresponding rates with transmit powers $P_1,P_2$ are plotted. Depending on the candidate set in each decoding structure, the transmit powers can be less than maximum power or equal to the maximum power $P_{max}$. For example, in $\mathbf{R}^{nn}$, maximum power is used: $P_1=P_2=P_{max}$, whereas in $\mathbf{R}^{nd}$, $P_1=P_{max}$ and $0 \leq P_2 \leq P_{max}$ and in $\mathbf{R}^{dd}$, $0 \leq P_1,P_2 \leq P_{max}$.
For simulation purposes, we allow the transmit powers to take 10 values between 0 and $P_{max}$, inclusively. The rate points plotted are achieved by the proposed Pareto boundary parameterization and the red asterisk is the maximum sum rate point by employing the maximum sum rate point parameterization where as the red square is the MRT strategies: $\left(\mathbf{w}_1=\frac{\mathbf{h}_{21}}{\|\mathbf{h}_{21} \|}, \mathbf{w}_2=\frac{\mathbf{h}_{22}}{\|\mathbf{h}_{22} \|} \right)$ in $\mathbf{R}^{nd}$; $\left(\mathbf{w}_1=\frac{\mathbf{h}_{11}}{\|\mathbf{h}_{11} \|}, \mathbf{w}_2=\frac{\mathbf{h}_{12}}{\|\mathbf{h}_{12} \|} \right)$ in $\mathbf{R}^{dn}$  and $\left(\mathbf{w}_1=\frac{\mathbf{h}_{11}}{\|\mathbf{h}_{11} \|}, \mathbf{w}_2=\frac{\mathbf{h}_{22}}{\|\mathbf{h}_{22} \|} \right)$ in $\mathbf{R}^{dd}$ and $\mathbf{R}^{nn}$. 

\begin{figure}
  \centering
  \subfloat[Achievable rate region of $\mathbf{R}^{nd}$: proposed parameterization achieves the Pareto Boundary and maximum sum rate point.]{\label{fig:rnd}\includegraphics[width=7cm, height=6cm,keepaspectratio]{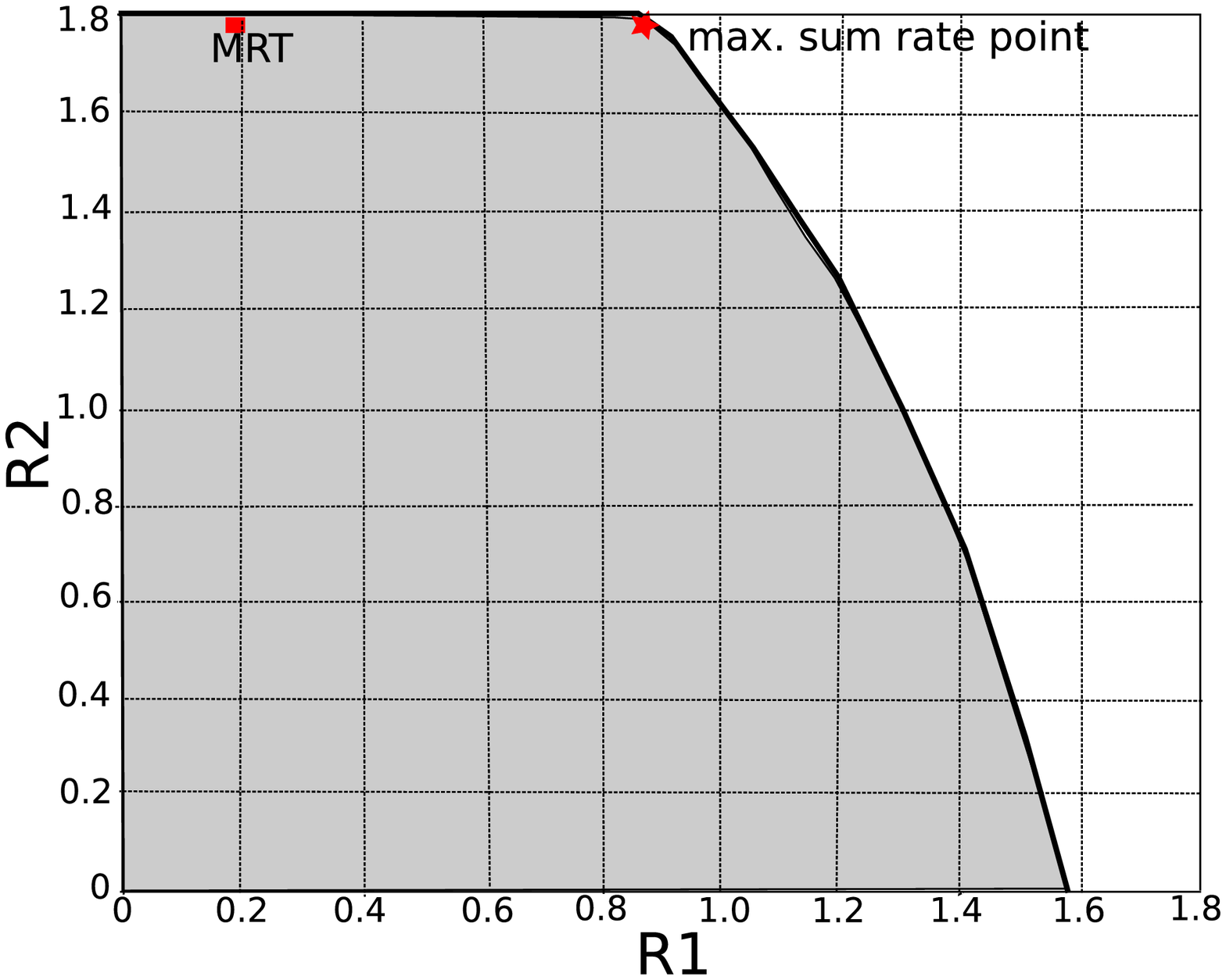}} 
  \subfloat[Achievable rate region of $\mathbf{R}^{dn}$: proposed parameterization achieves the Pareto Boundary and maximum sum rate point. \label{fig:rdn}]{\includegraphics[width=7cm, height=6cm,keepaspectratio]{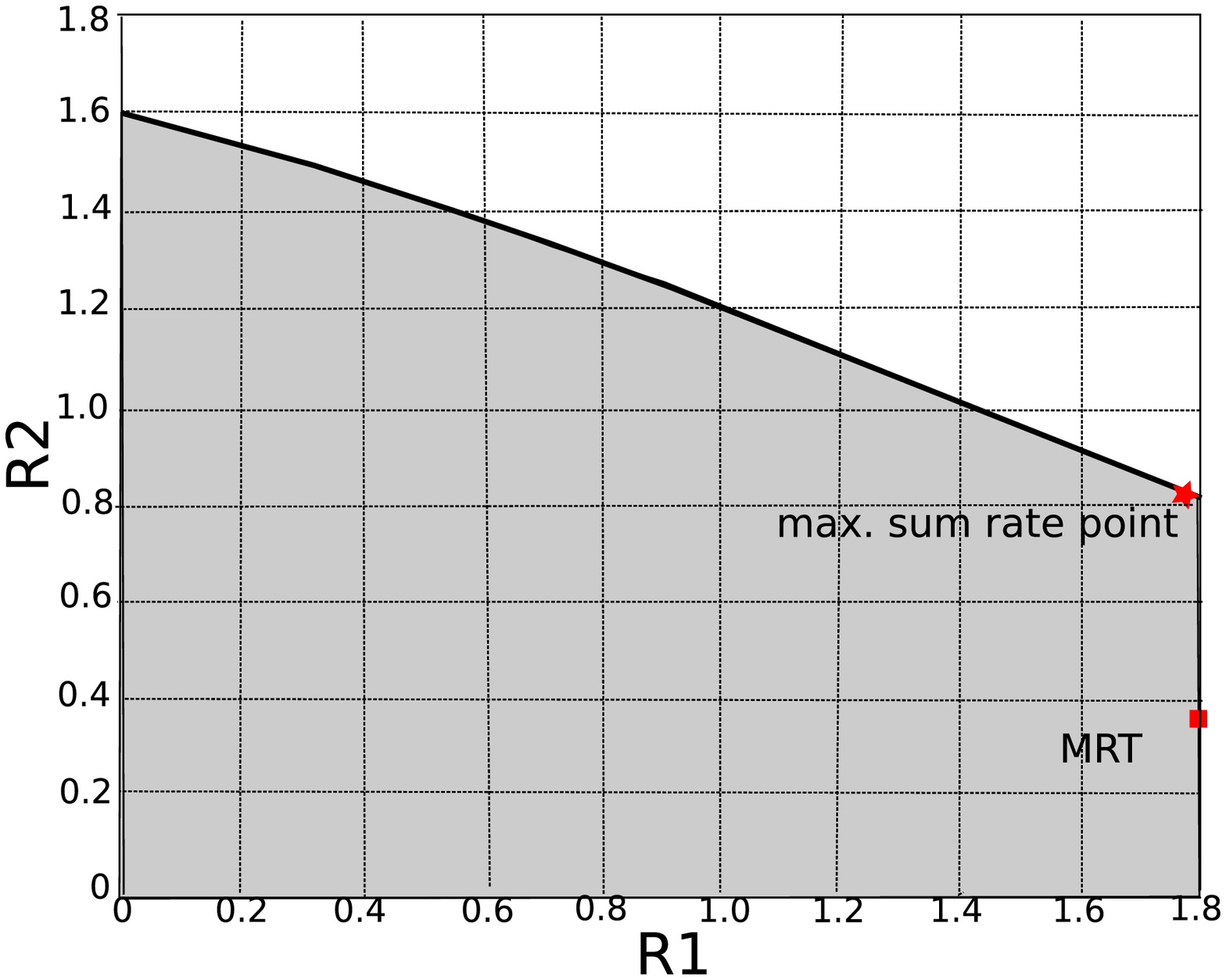}}\\
  \subfloat[Achievable rate region of $\mathbf{R}^{dd}$: proposed parameterization achieves the Pareto Boundary and maximum sum rate point. \label{fig:rdd}]{\includegraphics[width=7cm, height=6cm,keepaspectratio]{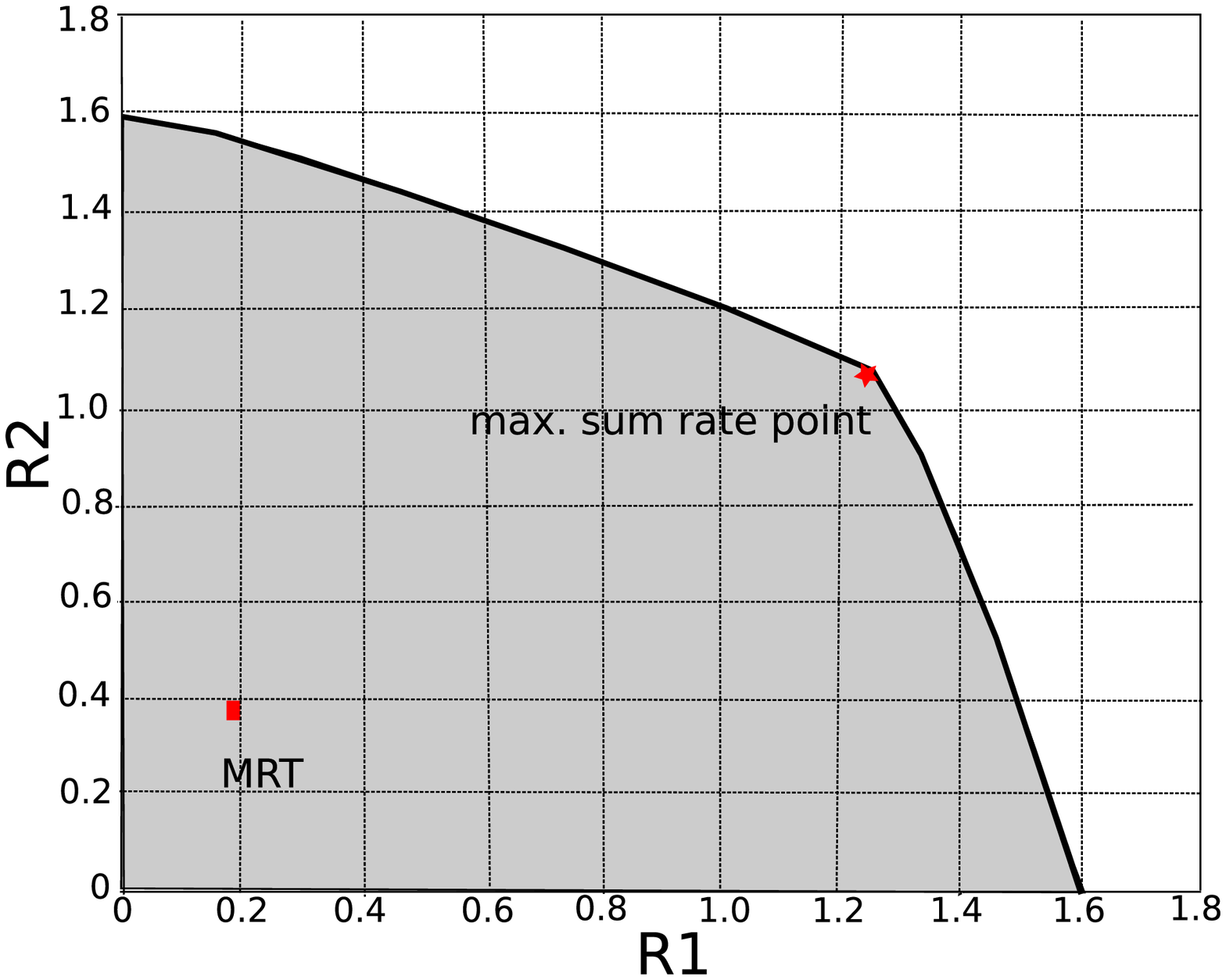}}    \subfloat[Achievable rate region of $\mathbf{R}^{nn}$: proposed parameterization achieves the Pareto Boundary and maximum sum rate point. \label{fig:rnn}]{\includegraphics[width=7.5cm, height=7cm,keepaspectratio]{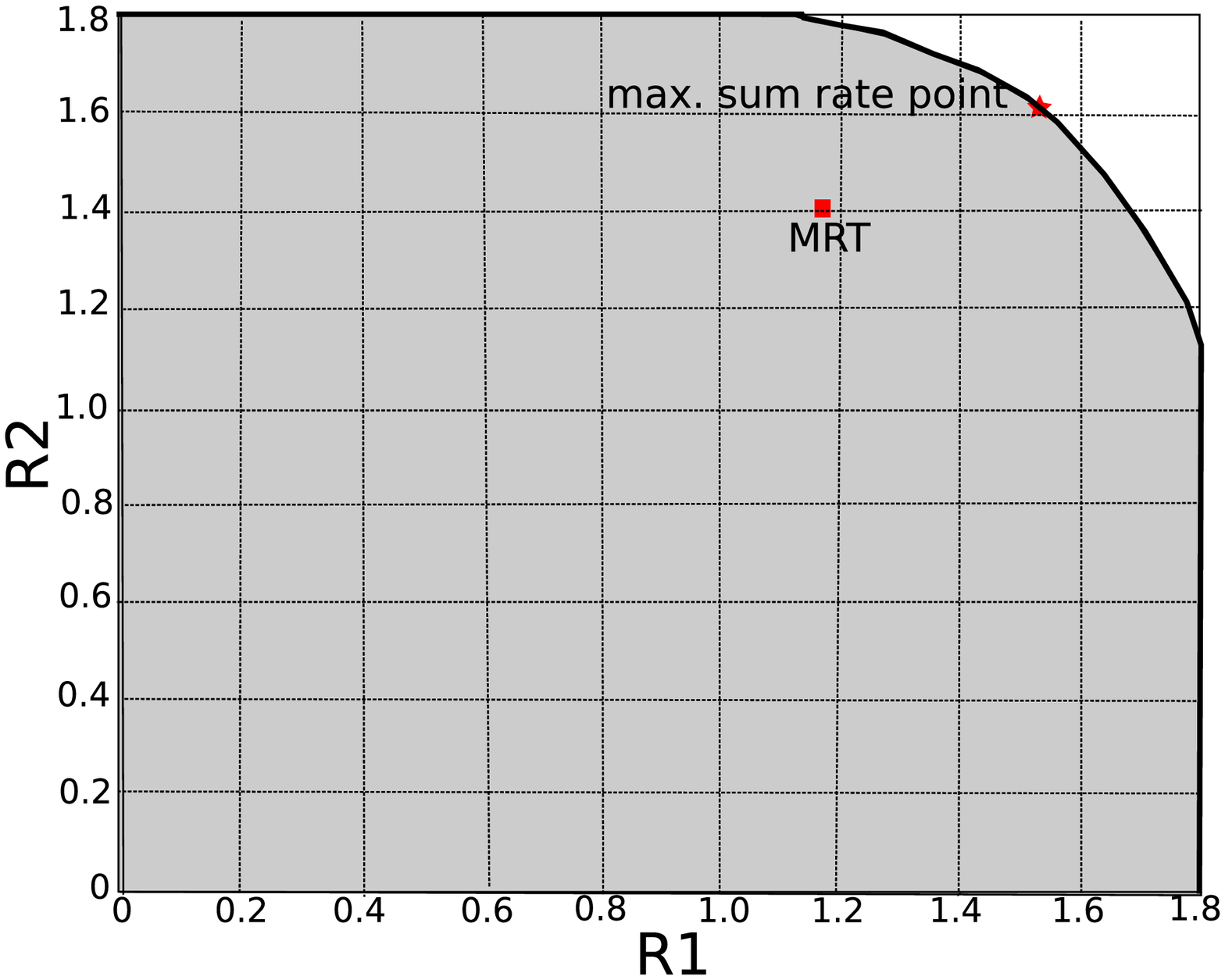}}
  \caption{Achievable rate region of different decoding structures.}
  \label{fig:ach_region}
\end{figure}

\subsection{Empirical Frequency of MRT strategies}\label{section:sim_empirical_freq}
In Fig. \ref{fig:rnd_mrtopt}, we demonstrate the variation of the empirical frequency of  MRT strategies in $\mathbf{R}^{nd}$ when SNR increases. In $\mathbf{R}^{nd}$, the empirical frequency of beamforming vectors pair $\left(\frac{\mathbf{h}_{21}}{\| \mathbf{h}_{21}\|}, \frac{\mathbf{h}_{22}}{\| \mathbf{h}_{22}\|}\right)$ is  50\% when SNR goes to infinity where as the empirical frequency of $\left(\frac{\mathbf{h}_{11}}{\| \mathbf{h}_{11}\|}, \frac{\mathbf{h}_{22}}{\| \mathbf{h}_{22}\|}\right)$ is  10\%. It shows that with $\mathbf{R}^{nd}$, in high SNR, Tx 1 should amplify interference signal by beamforming at the interference channel and Tx 2 should amplify desired signal by beamforming at the desired channel and by doing so, it achieves maximum sum rate on average 50\% of the channel realizations.

\begin{figure}
\begin{center}
\includegraphics[width=12cm, height=8cm,keepaspectratio]{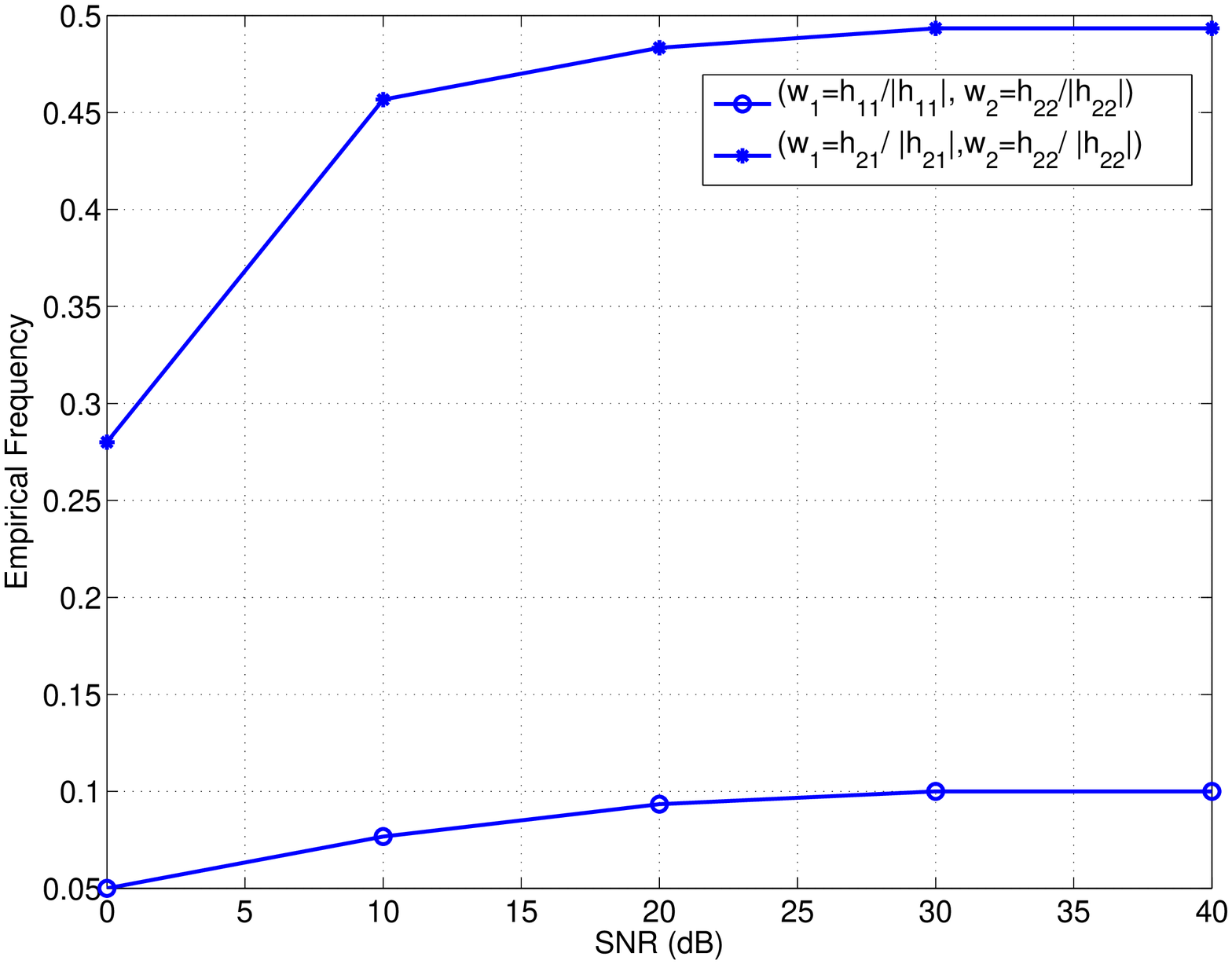}
\end{center}
\caption{MRT optimality in  $\mathbf{R}^{nd}$ when SNR increases: interference should be maximized half of the time when SNR goes to infinity. \label{fig:rnd_mrtopt}}
\end{figure}

In Fig. \ref{fig:rddmrtopt}, we compare the maximum sum rate and the rates achieved by MRT strategies in $\mathbf{R}^{dd}$: $(\mathbf{w}_1=\frac{\mathbf{h}_{11}}{\| \mathbf{h}_{11}\|},  \mathbf{w}_2=\frac{\mathbf{h}_{22}}{\| \mathbf{h}_{22}\|})$ and $(\mathbf{w}_1=\frac{\mathbf{h}_{21}}{\| \mathbf{h}_{21}\|}, \mathbf{w}_2=\frac{\mathbf{h}_{12}}{\| \mathbf{h}_{12}\|})$. In the x-axis, we plot the percentage of the average maximum sum rate whereas the y-axis is the percentage of channel realizations such that the MRT rates are less than a certain percentage of the maximum sum rate. Simulations show that the sum rate achieved by MRT strategies are less than 20\% to 80\% of the maximum sum rate.  Since maximizing direct channel gain or interference gain do not reach maximum sum rate,  it seems to imply that in $\mathbf{R}^{dd}$, interference should not be maximized or minimized and should be balanced instead.
\begin{figure}
\begin{center}
\includegraphics[width=12cm, height=8cm,keepaspectratio]{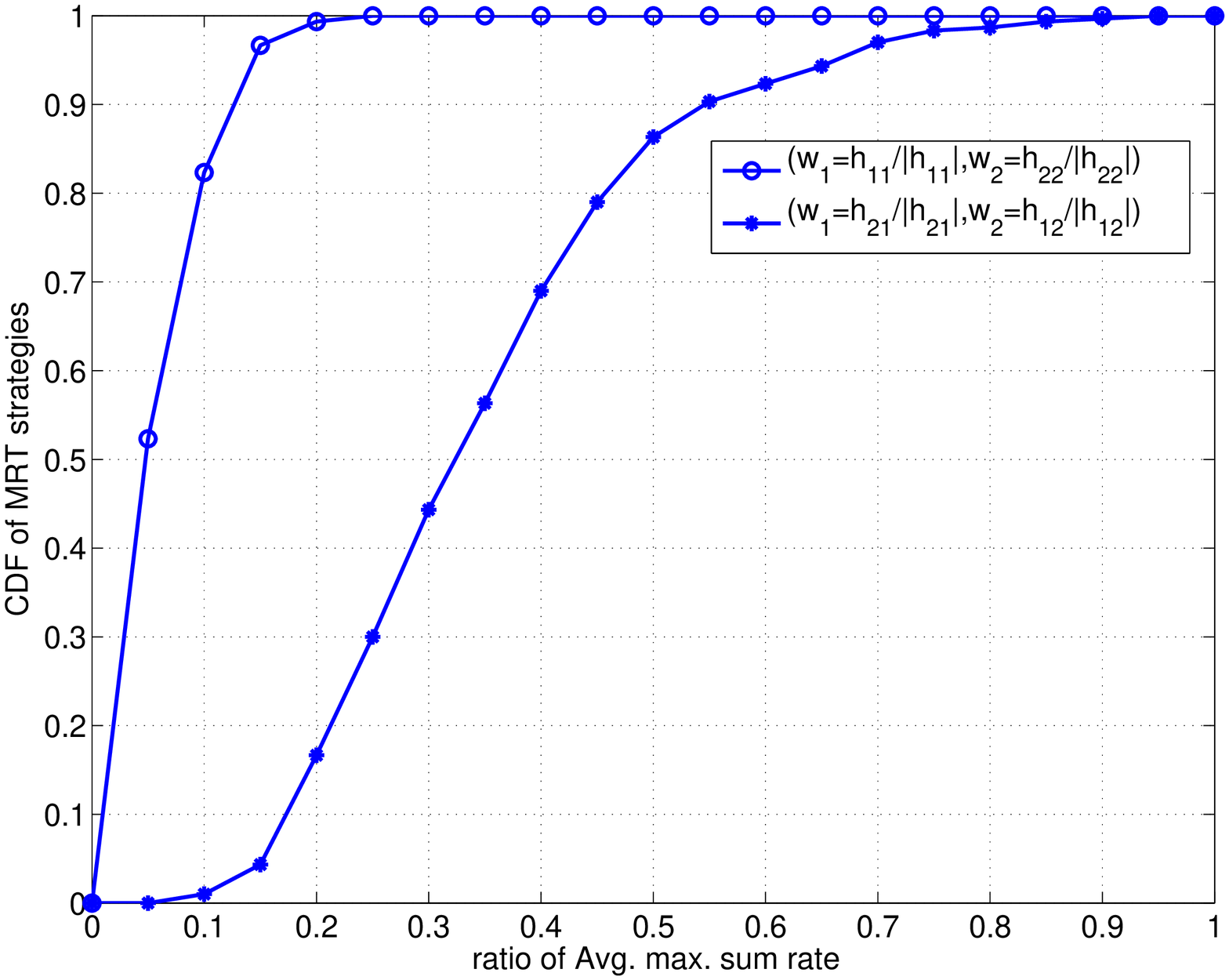}
\caption{Maximum sum rate plots in different channel realizations. MRT strategies are not sum rate optimal in $\mathbf{R}^{dd}$. \label{fig:rddmrtopt}}
\end{center}
\end{figure}

In Fig. \ref{fig:rate_diff}, we plotted the averaged difference between the maximum sum rate and MRT strategies in $\mathbf{R}^{nd}$ and $\mathbf{R}^{dd}$ respectively. From Fig. \ref{fig:rndmrtloss}, the rate difference between MRT strategies and the maximum sum rate decreases with SNR and reaches to about 2\% and 4\% at 40dB SNR. Thus, even if the empirical frequency is about 50\% and 10\%, MRT strategies only lose about 2\% and 4\% of the maximum sum rate of $\mathbf{R}^{nd}$. On the other hand, from Fig. \ref{fig:rddmrtloss}, it shows that the rate difference between MRT strategies and the maximum sum rate in $\mathbf{R}^{dd}$ increases with SNR and we conclude that MRT strategies are not sum rate optimal in $\mathbf{R}^{dd}$.

\begin{figure}
  \centering
  \subfloat[The averaged sum rate difference between maximum sum rate point and MRT strategies in $\mathbf{R}^{nd}$. \label{fig:rndmrtloss}]{\includegraphics[width=10cm, height=8cm,keepaspectratio]{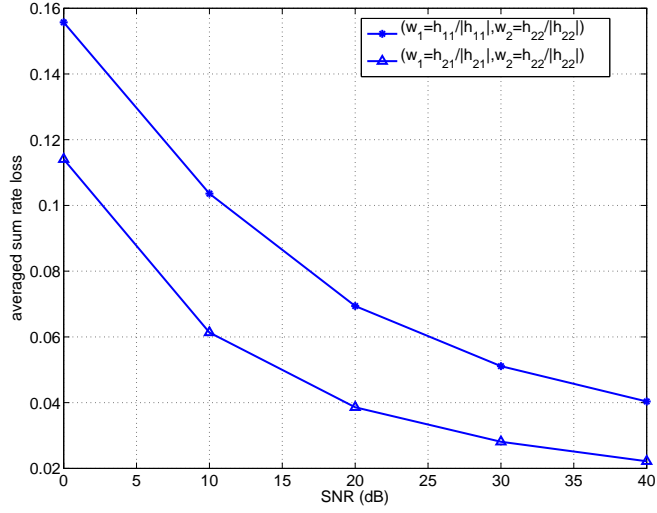}} \hspace{0.5cm} 
  \subfloat[The averaged sum rate difference between maximum sum rate point and MRT strategies in $\mathbf{R}^{dd}$. \label{fig:rddmrtloss}]{\includegraphics[width=10cm, height=8cm,keepaspectratio]{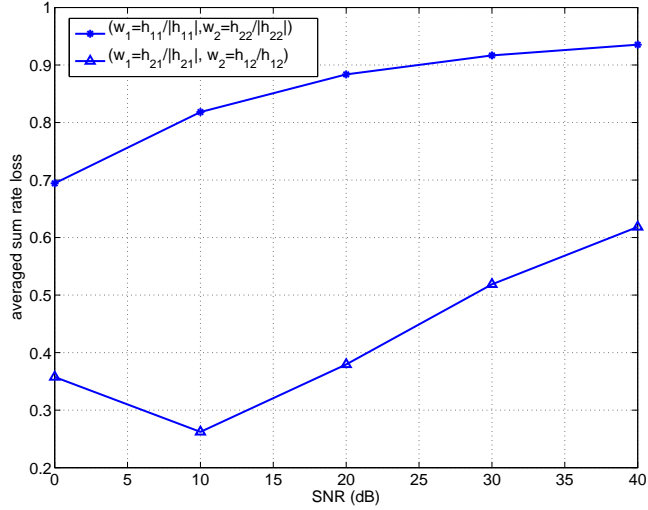}}
  \caption{The averaged sum rate difference between maximum sum rate point and MRT strategies in $\mathbf{R}^{nd}$ and $\mathbf{R}^{dd}$.}
  \label{fig:rate_diff}
\end{figure}

\subsection{Correlated Channels and Sum Rate optimal decoding structures}\label{section:sim_theta_ch}

In this section, we assume a symmetric channel \cite{Annapureddy} in which the direct channels, $\mathbf{h}_{ii}$, are i.i.d complex gaussian vector channels. The interference channel $\mathbf{h}_{ji}$ has a projection angle $\theta_i$ with the direct channel $\mathbf{h}_{ii}$:
\begin{equation}
|\mathbf{h}_{ji}^H \mathbf{h}_{ii}|= \|\mathbf{h}_{ii} \|  \|\mathbf{h}_{ji}\| \cos(\theta_i).
\end{equation}
Moreover, we define the signal to interference ratio SIR as
\begin{equation}
\text{SIR}=  \frac{\| \mathbf{h}_{ii}\|^2}{\| \mathbf{h}_{ji}\|^2}.
\end{equation} 

In Fig. \ref{fig:sumratealpha}, we compare the sum rate achieved by $\mathbf{R}^{nn}$, $\mathbf{R}^{dd}$ and TDMA. When the strengthen of interference channel increases and consequently SIR decreases, there is a transition from $\mathbf{R}^{nn}$ to TDMA to $\mathbf{R}^{dd}$: treating interference as noise is sum rate optimal in low interference regime and then time sharing should be performed and then decoding interference is sum rate optimal in high interference regime. When the angle between the interference channel increases to $\theta=0.15 \pi$, about 27 degrees, there is a direct transition between treating interference as noise and decoding interference. Thus, when the direct channel and the interference channel are more \emph{apart}, time sharing is not sum rate optimal and outperformed by the other decoding structures.

\begin{figure}
\begin{center}
\includegraphics[width=12cm, height=10cm,keepaspectratio]{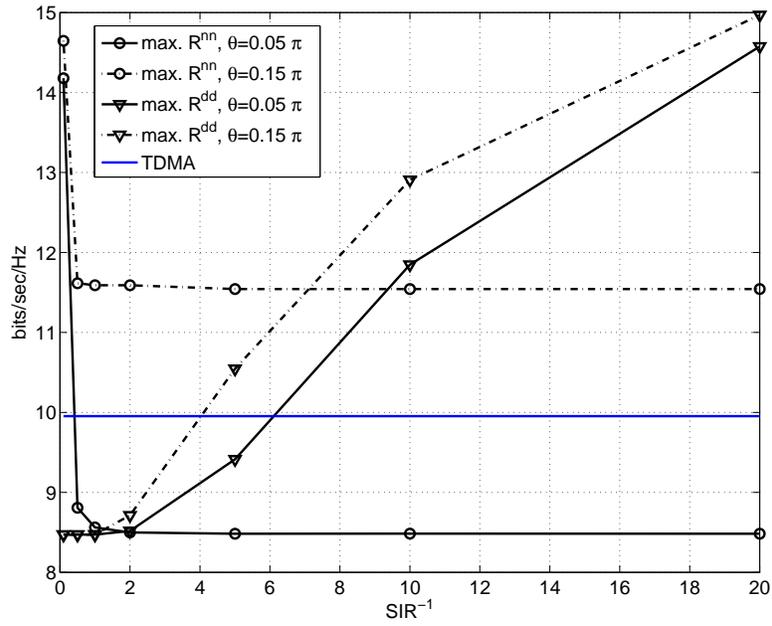}
\caption{The averaged sum rate of different decoding structure comparing to TDMA in symmetric channel when the strength of the interference channel increases. \label{fig:sumratealpha}}
\end{center}
\end{figure}

In Fig. \ref{fig:sumrate_snr}, we compare the maximum sum rate achieved in different decoding structures with TDMA when the system SNR increases. When the interference channel is as strong as the direct channel SIR$=1$, treating interference as noise is sum rate optimal in all SNR range. When the interference channel power increases SIR$^{-1}=5, 10, 20$, both Rxs. decoding interference is sum rate optimal in low SNR whereas one Rx treating interference as noise and one Rx decoding interference is sum rate optimal in high SNR.

\begin{figure}
\begin{center}
\includegraphics[width=12cm, height=10cm,keepaspectratio]{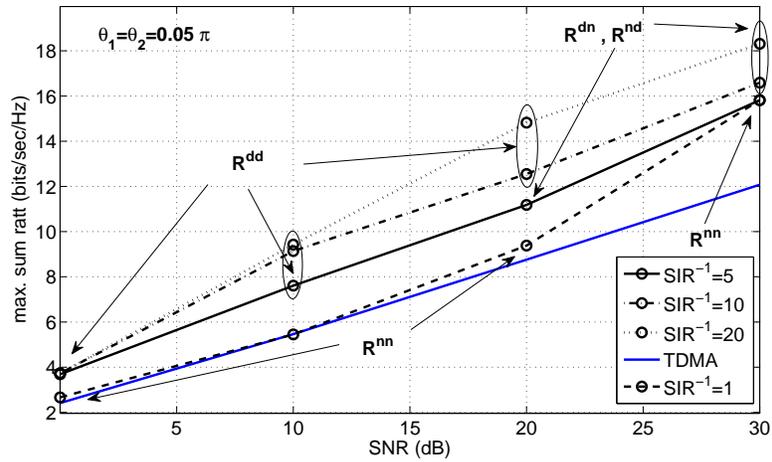}
\caption{The averaged sum rate of different decoding structure comparing to TDMA in symmetric channel when the system SNR increases. \label{fig:sumrate_snr}}
\end{center}
\end{figure}

\subsection{Performance of suboptimal algorithm}\label{sec:sim_algo}

For illustration purposes, we propose a very simple transmission strategy with only finite low number of beamforming vector choices. This transmission strategy is inspired by the parameterizaiton of each decoding structure. We propose to select only two beamforming vectors in each candidate set. Based on channel states information, we compare the sum rate performance of these eight beamforming vectors and choose the beamforming vector and the corresponding decoding structure which achieves the highest sum rate.
\begin{itemize}
\item NN region: $\left(\frac{\Pi_{21}^\perp \mathbf{h}_{11}}{\|\Pi_{21}^\perp \mathbf{h}_{11} \|}, \frac{\Pi_{12}^\perp \mathbf{h}_{22}}{\|\Pi_{12}^\perp \mathbf{h}_{22} \|} \right)$ and $\left( \frac{\mathbf{h}_{11}}{\|\mathbf{h}_{11} \|},   \frac{\mathbf{h}_{22}}{\|\mathbf{h}_{22} \|}\right)$.
\item ND region: $\left(\frac{\mathbf{h}_{21}}{\|\mathbf{h}_{21} \|}, \frac{\mathbf{h}_{22}}{\| \mathbf{h}_{22}\|} \right)$ and $\left(\frac{\mathbf{h}_{11}}{\|\mathbf{h}_{11} \|}, \frac{\mathbf{h}_{22}}{\| \mathbf{h}_{22}\|}\right)$.
\item DN region: $\left(\frac{\mathbf{h}_{11}}{\|\mathbf{h}_{11} \|}, \frac{\mathbf{h}_{12}}{\| \mathbf{h}_{12}\|}\right)$ and $\left(\frac{\mathbf{h}_{11}}{\|\mathbf{h}_{11} \|}, \frac{\mathbf{h}_{22}}{\| \mathbf{h}_{22}\|}\right)$.
\item DD region: $\left(\frac{\mathbf{h}_{21}}{\|\mathbf{h}_{21} \|}, \frac{\mathbf{h}_{12}}{\|\mathbf{h}_{12} \|}\right)$ , $\left(\frac{\mathbf{h}_{11}}{\|\mathbf{h}_{11} \|}, \frac{\mathbf{h}_{22}}{\|\mathbf{h}_{22} \|}\right)$, $\left(\frac{\mathbf{h}_{21}}{\|\mathbf{h}_{21} \|}, \frac{\mathbf{h}_{22}}{\| \mathbf{h}_{22}\|} \right)$ and  $\left(\frac{\mathbf{h}_{11}}{\|\mathbf{h}_{11} \|}, \frac{\mathbf{h}_{12}}{\| \mathbf{h}_{12}\|}\right)$ .
\item TDMA: a time sharing scheme between single user points and  $\mathbf{w}_i=\frac{\mathbf{h}_{ii}}{\| \mathbf{h}_{ii}\|}$.
\end{itemize}

In the NN region, we propose to choose either the interference nulling solution or the desired channel gain maximizing solution. It has been shown in previous literature that in MISO-IC-SUD and low SNR regime, maximizing desired channel gain is sum rate optimal whereas in high SNR regime, interference nulling is sum rate optimal. 

It was shown in SISO-IC \cite{Costa1985, Sato1981}that the DD scheme is sum rate optimal among all four decoding structures when the strength of both interference channels are strong and the DN or ND scheme is sum rate optimal when one interference channel is strong and the other interference channel is weak compared to the desired channel. The interference maximizing beamforming solution and the desired channel channel power beamforming solution are chosen in DN, ND and DD regions in the proposed algorithm to verify the analogy from SISO-IC to MISO-IC.

In Fig. \ref{fig:suboptalgo2}, we plotted the maximum sum rate achieved by different decoding structure and compare it with the proposed simple algorithm when SIR decreases. We see that when the interference is weak, it is sum rate optimal to treat interference as noise and when the interference strength increases, sum rate can be increased by allowing one of the Rx to decode interference and in the strong interference regime, both Rxs. decoding interference achieves the highest sum rate. Depending on the channel coefficients, TDMA may outperform $\mathbf{R}^{nn}$ and $\mathbf{R}^{dd}$ in the medium interference regime. Note that the computation of the maximum sum rate point is NP-hard. However, we see the the proposed simple algorithm achieves nice sum rate performance with only five choices of beamforming vectors.


\begin{figure}
\begin{center}
\includegraphics[width=12cm, height=10cm,keepaspectratio]{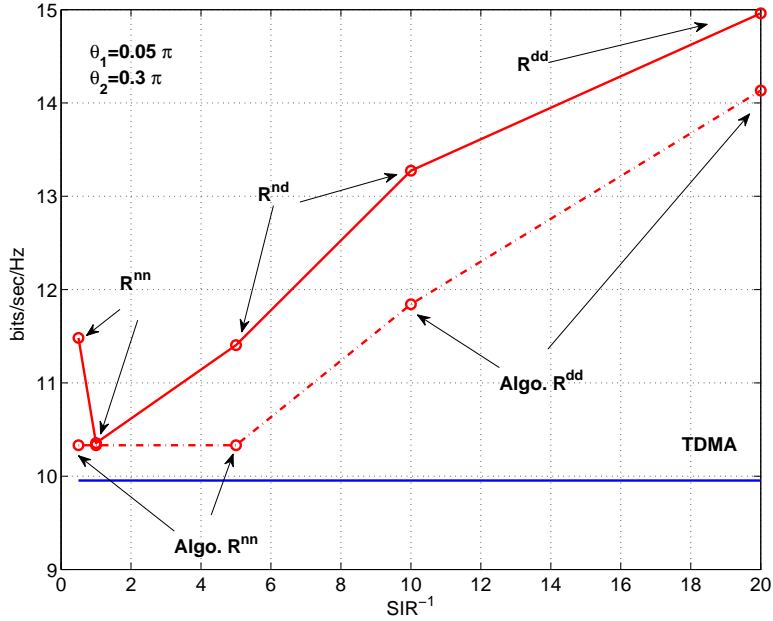}
\caption{Sum rate optimal decoding structures when the strength of interference channel increases. \label{fig:suboptalgo2}}
\end{center}
\end{figure}

\section{Conclusion and future work}
The interference decoding capability brings additional freedom to the Rxs. which either  decode interference or treat interference as noise. However, it is not trivial when Txs. should avoid interference and when should amplify interference power. To answer this question,  we formulate the achievable rate region for a two-user MISO-IC-SUD. We provide an in-depth analysis of the achievable rate region, as a union of different decoding structures. The Pareto boundary is then characterized in terms of both power allocation and beamforming vectors. As a direct application of the Pareto boundary characterization, we characterize the maximum sum rate points. The candidate set to the maximum sum rate point is a strict subset of the the candidate set of the Pareto boundary. With the maximum sum rate characterization, we derive the MRT optimality conditions which describe the conditions in which simple MRT strategies are sum rate optimal. We conclude the paper by providing simulation results which shed some insights into the question `` \emph{When is selfishness sum rate optimal?}''. Results show that MRT strategies have only 2\% to 4\% rate loss comparing to the maximum sum rate point in high SNR if one Rx decodes interference and the other treats interference as noise. On the other hand, MRT is not sum rate optimal if both Rxs. decode interference. In symmetric channels, there is a transition in decoding structure from treating interference as noise to TDMA to decoding interference at both Rxs. when the strength of interference increases.

The achievable rate problem for the $K$-user MISO-IC-IDC is not simple as the number of possible interference decoding order increases exponentially with the number of users. It is not easy to see which decoding order is better than the others and whether this decoding order in this decoding structure is better than other decoding structures. This extension to the $K$-user case is currently under preparation.


%

%
\section{Appendix}

\subsection{Proof of Thm. \ref{thm:transmit_cov}}\label{app:transmit_cov}
In this sequel, we are going to prove that the Pareto boundary in NN region, DN region and DD region are attained by rank 1 matrices.

\begin{Lem}\label{lem:rnn}
In the NN region, where Rx 1 and Rx 2 treat interference as noise, the Pareto boundary attaining transmit covariance matrices are rank one.
\end{Lem}
\begin{proof}
See reference \cite{Shang2009, Mochaourab2010}.
\end{proof}

To prove that the Pareto optimal transmit covariance matrices in DN and DD regions are rank one, we follow and modify slightly the proof from \cite{Mochaourab2010}. To facilitate the discussion, we define the following received channel power region for Tx $i$, assuming transmit power being one, 
\begin{equation}
\Phi_i^s= \left\{(\mathbf{h}_{ii}^H \mathbf{S}_i \mathbf{h}_{ii}, \mathbf{h}_{ji}^H \mathbf{S}_i \mathbf{h}_{ji}): \mathbf{S}_i \in \mathcal{S} \right\}.
\end{equation} 
For each received channel power $g_{ji}= \mathbf{h}_{ji}^H \mathbf{S}_i \mathbf{h}_{ji}$, $i,j=1,2$, there is a set $\mathcal{K}_{ji}^{\uparrow}$ such that  each utility $u_k$, $k \in \mathcal{K}_{ji}^{\uparrow}$,  is monotonically increasing with received channel power $g_{ji}$,
\begin{equation}
u_k(g_{ii}, g_{ji}, g_{ij}, g_{jj}) \leq u_k(g_{ii}, g_{ji}', g_{ij}, g_{jj})
\end{equation}  if  $g_{ji} \leq g_{ji}'$. Similarly, there is a set  $\mathcal{K}_{ji}^{\downarrow}$ such that each utility $u_k$, $k \in \mathcal{K}_{ji}^{\downarrow}$, is monotonically decreasing with received channel power $g_{ji}$. 

\begin{Lem}\label{lem:rank_one}
For an arbitrary fixed Tx $i$, if there exist a received channel power $g_{ji}$, $j=1,2$, such that the number of utilities that are monotonically increasing and decreasing with $g_{ji}$ are both larger than zero and the number of such received channel power $g_{ji}$ is not larger than one, e.g.
\begin{equation}\label{eqt:assume}
\mathcal{N}_i =\left\{k:  \left| \mathcal{K}_{ki}^{\uparrow} \right| >0 , \left| \mathcal{K}_{ki}^{\downarrow} \right|> 0 \right\}, \| \mathcal{N}_i\| \leq 1, \hspace{1cm} i=1,2,
\end{equation} 
then the Pareto optimal transmit covariance matrices with respect to these utilities are rank one and attain the boundary of the corresponding received channel power regions $\Phi_i^s$.
\end{Lem}
\begin{proof}
We proceed by separating the cases where the some utilities are increasing or decreasing with the received channel power $g_{ji}$, for $i,j=1,2$.
\begin{itemize} 
\item If $\left|\mathcal{K}_{ji}^{\uparrow} \right| \neq 0$ and $\left|\mathcal{K}_{ji}^{\downarrow} \right| = 0$, then the Pareto optimal transmit covariance matrix $\mathbf{S}^*$ attains the boundary of the received channel power region $\Phi_i^s$. It is because for each utility $u_k(g_{ii}, g_{ji}, g_{ij}, g_{jj})$ such that $k \in \mathcal{K}_{ji}^{\uparrow}$, if $g_{ji}$ is not on the boundary of $\Phi_i^s$, we can choose a transmit covariance matrix $\mathbf{S}^*$ such that $g_{ji}^*=\mathbf{h}_{ji}^H \mathbf{S}^* \mathbf{h}_{ji} \geq g_{ji}$ which increases the value of the utility $u_k(g_{ii}, g_{ji}, g_{ij}, g_{jj})\leq u_k(g_{ii}, g_{ji}^*, g_{ij}, g_{jj})$. Since $\left| \mathcal{K}_{ji}^{\downarrow}\right|=0$, no utilities are decreased by modifying the transmit covariance matrix from $\mathbf{S}$ to $\mathbf{S}^*$.
\item If $\left|\mathcal{K}_{ji}^{\downarrow} \right| \neq 0$ and $\left|\mathcal{K}_{ji}^{\uparrow} \right| = 0$, then the Pareto optimal transmit covariance matrix $\mathbf{S}^*$ attains the boundary of the received channel power region $\Phi_i^s$. It is because for each utility $u_k(g_{ii}, g_{ji}, g_{ij}, g_{jj})$ such that $k \in \mathcal{K}_{ji}^{\uparrow}$, if $g_{ji}$ is not on the boundary of $\Phi_i^s$, we can choose a transmit covariance matrix $\mathbf{S}^*$ such that $g_{ji}^*=\mathbf{h}_{ji}^H \mathbf{S}^* \mathbf{h}_{ji} \leq g_{ji}$ which increases the value of the utility $u_k(g_{ii}, g_{ji}, g_{ij}, g_{jj})\leq u_k(g_{ii}, g_{ji}^*, g_{ij}, g_{jj})$. Since $\left| \mathcal{K}_{ji}^{\uparrow}\right|=0$, no utilities are decreased by modifying the transmit covariance matrix from $\mathbf{S}$ to $\mathbf{S}^*$.
\item  According to the assumption \eqref{eqt:assume}, for each Tx $i$, there exist at most one $j$ such that $\left|\mathcal{K}_{ji}^{\downarrow} \right| \neq 0$ and $\left|\mathcal{K}_{ji}^{\uparrow} \right| \neq 0$. In this case, the received channel power with transmit power, $g_{ji} P_i$ has value between 0 and $\| \mathbf{h}_{ji}\|^2 P_{max}$ which can be achieved by setting $g_{ji}$ to be on the boundary of $\Phi^s_i$ and letting $P_i$ vary from 0 to $P_{max}$. Hence, we can put the Pareto optimal transmit covariance matrix $\mathbf{S}^*$ on the boundary of the received channel power region $\Phi_i^s$ by allowing power control $0 \leq P_i \leq P_{max}$.
\end{itemize}
From \cite[ Lemma 3]{Mochaourab2010}, the received channel powers on the boundary of the received channel power region $\Psi_i^s$ are attained by rank one transmit covariance matrices, which completes the proof.
\end{proof}

In the following, we are going to show that the Pareto optimality problem in DN region and DD region respectively satisfy the assumption in Lemma \ref{lem:rank_one} and therefore the Pareto optimal transmit covariance matrices are rank one.

\subsubsection{In the DN region}\label{app:rdn_rankone}
Rx 1 decodes interference and Rx 2 treats interference as noise and the Pareto boundary attaining transmit covariance matrices are rank one. To see this, we  recall the achievable rates in $\mathbf{R}^{dn}$: 
\begin{equation}
\begin{aligned}
u_1(g_{11},g_{21},g_{12},g_{22})&=\log_2(1+g_{11}P_1)\\
u_2(g_{11},g_{21},g_{12},g_{22})&=  \min \left( \log_2 \left( 1+ \frac{g_{12}P_2}{g_{11} P_1+1} \right), \log_2 \left( 1+ \frac{g_{22}P_2}{1+ g_{21}P_1}\right)\right).
\end{aligned}
\end{equation} and therefore we have
\begin{equation}
\begin{aligned}
\mathcal{K}_{11}^{\uparrow}&= \{1\}; \hspace{1cm} \mathcal{K}_{11}^{\downarrow}= \{2\}; \hspace{1cm}
\mathcal{K}_{21}^{\uparrow}= \{\emptyset\}; \hspace{1cm} \mathcal{K}_{21}^{\downarrow}=\{ 2\};\\
\mathcal{K}_{12}^{\uparrow}&= \{2\}; \hspace{1cm} \mathcal{K}_{12}^{\downarrow}=\{\emptyset\}; \hspace{1cm}
\mathcal{K}_{22}^{\uparrow}= \{2\}; \hspace{1cm} \mathcal{K}_{22}^{\downarrow}= \{\emptyset\}.
\end{aligned}
\end{equation} Since the assumption \eqref{eqt:assume} is satisfied, by Lemma \ref{lem:rank_one}, we have the Pareto optimal transmit covariance matrices in DN region as rank one.

\subsubsection{In DD region}\label{app:rdd_rankone}
 Rx 1 and 2 decode interference and the Pareto boundary attaining transmit covariance matrices are rank one. To see this, we recall the achivable rates in $\mathbf{R}^{dd}$: 
\begin{equation}
\begin{aligned}
u_1(g_{11},g_{21},g_{12},g_{22}) &= \min\left( \log_2\left( 1+ g_{11}\right), \log_2\left( 1+ \frac{g_{21}}{1+g_{22}}\right)\right)\\
u_2(g_{11},g_{21},g_{12},g_{22}) &=  \min\left( \log_2\left( 1+ g_{22}\right), \log_2\left( 1+ \frac{g_{12}}{1+ g_{11}}\right)\right).
\end{aligned}
\end{equation}  Hence we have
\begin{equation}
\begin{aligned}
\mathcal{K}_{11}^{\uparrow}&= \{1\}; \hspace{1cm} \mathcal{K}_{11}^{\downarrow}= \{2\}; \hspace{1cm}
\mathcal{K}_{21}^{\uparrow}&= \{ 1\}; \hspace{1cm} \mathcal{K}_{21}^{\downarrow}=\{ \emptyset\};\\
\mathcal{K}_{12}^{\uparrow}&= \{2\}; \hspace{1cm} \mathcal{K}_{12}^{\downarrow}=\{\emptyset\}; \hspace{1cm}
\mathcal{K}_{22}^{\uparrow}&= \{2\}; \hspace{1cm} \mathcal{K}_{22}^{\downarrow}= \{ 1\}.
\end{aligned}
\end{equation} Since the assumption \eqref{eqt:assume} is satisfied, by Lemma \ref{lem:rank_one}, we have the Pareto optimal transmit covariance matrices in DD region are rank one.
\subsection{Proof of Thm. \ref{thm:pareto_rnd} }\label{app:rnd_pareto}

 The Pareto optimal beamforming vectors can be presented as the solutions of the following optimization problems, for some feasible value $r_i$, $i,j=1,2, j \neq i$,
\begin{equation}\label{eqt:rnd_pareto_f1}
\begin{aligned}
\max_{\mathbf{w}_1,\mathbf{w}_2, P_1, P_2} & \hspace{1cm} R_j(\mathbf{w}_1, \mathbf{w}_2, P_1, P_2) \\
\text{subject to} & \hspace{1cm} R_i(\mathbf{w}_1, \mathbf{w}_2, P_1, P_2) \geq r_i,\\
 & \hspace{1cm} \|\mathbf{w}_1\|=1, \|\mathbf{w}_2\|=1,\\
& \hspace{1cm} 0 \leq P_1 \leq P_{max}, 0 \leq P_2 \leq P_{max}.
\end{aligned}
\end{equation} With different decoding structure, the achievable rates $R_1, R_2$ are substituted with different rate expressions. However, both the ND and DD case resemble closely to the maximization problem of channel power  and the optimal solution can be characterized in a linear combination of some specific vectors.

\begin{Lem}\label{lem:pow_gain_gen}
Define a maximization problem of channel power in the following:
\begin{equation}\label{eqt:gen_form}
\begin{aligned}
\max_{\mathbf{w}} & \hspace{1cm} t\\
\text{subject to} & \hspace{1cm} |\mathbf{u}^H \mathbf{w}|^2 \geq ut, |\mathbf{v}^H \mathbf{w}|^2 \geq vt\\
& \hspace{1cm} \| \mathbf{w}\|^2 \leq 1
\end{aligned}
\end{equation} for some arbitrary fixed scalars $u \geq 0 $ and $v \geq 0$. The solutions $\mathbf{w}^*$ must be in the following set $\mathcal{W}$,
\begin{equation}\label{eqt:opt_gen_form}
\mathcal{W}=\bigg\{ \mathbf{w} \in \mathbb{C}^{N \times 1}: \mathbf{w}= \sqrt{\mu}\frac{\Pi_{u} \mathbf{v}}{\| \Pi_{u} \mathbf{v}\|} +  \sqrt{1-\mu}\frac{\Pi_{u}^\perp \mathbf{v}}{\| \Pi_{u}^\perp \mathbf{v}\|}, 0 \leq \mu \leq 1\bigg\}.
\end{equation}
\end{Lem}
\begin{Cor}\label{cor:formulation}
Reversing the signs of the inequality in \eqref{eqt:gen_form} to $ |\mathbf{u}^H \mathbf{w}|^2 \leq ut, |\mathbf{v}^H \mathbf{w}|^2 \leq vt$ does not change the characterization of the solutions in \eqref{eqt:opt_gen_form}.
\end{Cor}
\begin{proof}
We proceed by writing the Lagrangians of the problem, with Lagrange multipliers $\boldsymbol{\lambda}= [\lambda_1, \lambda_2, \lambda_3]$,
\begin{equation}
L(\mathbf{w},\boldsymbol{\lambda})= t- \lambda_1 \left( ut - |\mathbf{u}^H \mathbf{w}|^2 \right) - \lambda_2 \left( vt - |\mathbf{v}^H \mathbf{w}|^2\right) -\lambda_3  (\| \mathbf{w}\|^2-1).
\end{equation} Now, we compute the vanishing point of the Lagragian derivative which is a necessary condition of the optimal solution,
$\frac{\partial L(\mathbf{w}_1,\boldsymbol{\lambda})}{\partial \mathbf{w}^H}= \lambda_1 \mathbf{u} \mathbf{u}^H \mathbf{w}+ \lambda_2 \mathbf{v} \mathbf{v}^H \mathbf{w} - \lambda_3 \mathbf{w}=0$. We can write
$\lambda_1 \mathbf{u} \mathbf{u}^H \mathbf{w}+ \lambda_2 \mathbf{v}\mathbf{v}^H \mathbf{w} = \lambda_3 \mathbf{w}$ and adjusting the constant scaling, we have
$\lambda_1 \| \mathbf{u}\|^2 \frac{\mathbf{u}\mathbf{u}^H}{\| \mathbf{u}\|^2} \mathbf{w}+ \lambda_2 \| \mathbf{v} \|^2 \frac{\mathbf{v}\mathbf{v}^H}{\| \mathbf{v} \|^2} \mathbf{w} = \lambda_3 \mathbf{w}$. Therefore, the eigenvector $\mathbf{w}$ is a composition of its projection on $\mathbf{u}$ and $\mathbf{v}$:
\begin{equation}
\lambda_1 \| \mathbf{u}\|^2 \Pi_{u}\mathbf{w}+ \lambda_2 \| \mathbf{v} \|^2\Pi_{v} \mathbf{w} = \lambda_3 \mathbf{w}.
\end{equation} Since $\lambda_i \geq 0, i=1,2,3$,  we can write, for some complex-valued $\mu_1, \mu_2 $, 
\begin{equation}\label{eqt:middle}
\mathbf{w}= \frac{ \mu_1 \frac{\mathbf{u}}{\|\mathbf{u}\|} + \mu_2  \frac{\mathbf{v}}{\|\mathbf{v}\|}}{\left\| \mu_1 \frac{\mathbf{u}}{\|\mathbf{u}\|} + \mu_2  \frac{\mathbf{v}}{\|\mathbf{v}\|} \right\|}.
\end{equation} 
Now we define the set of beamforming vectors that satisfy \eqref{eqt:middle},
$\mathcal{U}=\left\{ \mathbf{w}:\mathbf{w}= \frac{ \mu_1 \frac{\mathbf{u}}{\|\mathbf{u}\|} + \mu_2  \frac{\mathbf{v}}{\|\mathbf{v}\|}}{\left\| \mu_1 \frac{\mathbf{u}}{\|\mathbf{u}\|} + \mu_2  \frac{\mathbf{v}}{\|\mathbf{v}\|} \right\|}, \mu_1, \mu_2 \in \mathbb{C} \right\}
$. Then, we show in the following that $\mathcal{U}$ is a subset of $\mathcal{W}$ in Lemma \ref{lem:pow_gain_gen}.
We start with
\begin{equation}\label{eqt:rnd_pareto_w1_step}
\begin{aligned}
\mathbf{w} &= \mu_1 \frac{\mathbf{u}}{\|\mathbf{u}\|} + \mu_2  \frac{\mathbf{v}}{\|\mathbf{v}\|}\\
&\overset{(a)}{=} \frac{\mu_1}{\|\mathbf{u} \|} \left( \Pi_{v}+\Pi_{v}^\perp \right)\mathbf{u}+ \mu_2  \frac{ \| \mathbf{v}\|^2}{| \mathbf{v}^H \mathbf{u}|} e^{-j \phi_{uv}} \Pi_{v} \mathbf{u} \\
&= \left(\frac{\mu_1}{\|\mathbf{u} \|}  +\mu_2  \frac{ \| \mathbf{v}\|^2}{| \mathbf{v}^H \mathbf{u}|} e^{-j \phi_{uv}}\right) \Pi_{v} \mathbf{u} +\Pi_{v}^\perp \mathbf{u} \\
&= z_1 \frac{\Pi_{v} \mathbf{u}}{\|\Pi_{v} \mathbf{u}\|} + z_2 \frac{\Pi_{v}^\perp \mathbf{u}}{\|\Pi_{v}^\perp \mathbf{u}\|}
\end{aligned}
\end{equation} where (a) is due to $\mathbf{v}= \frac{\Pi_{\mathbf{v}}\mathbf{u} \| \mathbf{v}\|^2}{\mathbf{v}^H \mathbf{u}}$ and $ \phi_{uv}= \arg(\mathbf{v}^H \mathbf{u})$. The parameter $z_1$ is complex and $z_2$ is real; the values of $z_1$ and $z_2$ are scaled such that $\| \mathbf{w}\|=1$.
Notice that $|\mathbf{u}^H \mathbf{w} | = \left|z_1 \| \Pi_{v} \mathbf{u}\| + z_2 \|\Pi_{v}^\perp \mathbf{u} \|\right| \overset{(a)}{\leq} |z_1| \| \Pi_{v} \mathbf{u}\| + z_2 \|\Pi_{v}^\perp \mathbf{u} \|$ and
$|\mathbf{v}^H \mathbf{w} |= \left| z_1 \frac{\mathbf{v}^H \mathbf{u}}{\|\Pi_{v} \mathbf{u} \|}\right| =  |z_1| \frac{|\mathbf{v}^H \mathbf{u}|}{\|\Pi_{v} \mathbf{u} \|}$. Note that equality at $(a)$ when $z_1$ is real and the phase of $z_1$ does not affect $| \mathbf{v}^H \mathbf{w}|$. Hence, $z_1$ can be chosen real and since the two basis are orthogonal, the power constraint of $\mathbf{w}$ is satisfied when
$z_1^2 +c_2^2 =1 $ and thus, we can write $z_1= \sqrt{\mu}$ and $c_2= \sqrt{1-\mu}$ for $0\leq \mu \leq 1$.

If the signs of inequalities are reversed and we have  $ |\mathbf{u}^H \mathbf{w}|^2 \leq ut, |\mathbf{v}^H \mathbf{w}|^2 \leq vt$ in the constraints in \eqref{eqt:gen_form}, then the corresponding signs changes to minus from positive in the Lagrangian but does not affect the discussion above and the characterization of the Pareto optimal solutions holds.
\end{proof}

Substitute the rate definitions \eqref{eqt:ch4_rate_def} into \eqref{eqt:rnd_pareto_f1} and maximizing $R_1$ subjecting to a constraint on $R_2$, we let $z_{22}^2=|\mathbf{h}_{22}^H \mathbf{w}_2|^2 P_2$ and $z_{12}^2= |\mathbf{h}_{12}^H \mathbf{w}_2|^2 P_2$ and we focus on the subproblem that concerns $\mathbf{w}_1$,
\begin{equation}
\begin{aligned}
\max_{\mathbf{w}_1} & \hspace{1cm} t\\
\text{subject to}& \hspace{1cm} |\mathbf{h}_{21}^H \mathbf{w}_1|^2 \geq (z_{22}^2+1)t,\\
& \hspace{1cm} |\mathbf{h}_{11}^H \mathbf{w}_1|^2 \geq (z_{12}^2 +1)t,\\
& \hspace{1cm} \|\mathbf{w}_1\|^2 \leq 1.
\end{aligned} 
\end{equation} This has the same formulation as in Lemma \ref{lem:pow_gain_gen}. By substituting $\mathbf{u}=\mathbf{h}_{21}$ and $\mathbf{v}=\mathbf{h}_{11}$, we obtain the characterization of the Pareto optimal beamforming vectors as a linear combination of the vectors $\Pi_{21} \mathbf{h}_{11}$ and $\Pi_{21}^\perp \mathbf{h}_{11}$.
Now we reverse the optimization order : maximize $R_2$ subject to a constraint on $R_1$. 
After some manipulations,  we obtain for some $z_{11}^2, z_{21}^2$,
\begin{equation}\label{eqt:rnd_pareto_pro2}
\begin{aligned}
\max_{\mathbf{w}_2, P_2} & \hspace{1cm} |\mathbf{h}_{22}^H \mathbf{w}_2|^2 P_2\\
\text{subject to} & \hspace{1cm} |\mathbf{h}_{22}^H \mathbf{w}_2|^2 P_2  \leq \frac{z_{21}^2}{2^{r_1}-1}-1,\\
& \hspace{1cm} |\mathbf{h}_{12}^H \mathbf{w}_2|^2 P_2  \leq \frac{z_{11}^2}{ 2^{r_1}-1}-1,\\
& \hspace{1cm} \| \mathbf{w}_2\| \leq 1, 0 \leq P_2 \leq P_{max}.
\end{aligned}
\end{equation} 
By Corollary \ref{cor:formulation}, we have 
\begin{equation}
\mathbf{w}_2= \sqrt{\mu_2} \frac{\Pi_{12} \mathbf{h}_{22}}{\|\Pi_{12} \mathbf{h}_{22} \|} +\sqrt{1-\mu_2} \frac{\Pi_{12}^\perp \mathbf{h}_{22}}{\|\Pi_{12}^\perp \mathbf{h}_{22} \|} 
\end{equation} for $0 \leq \mu_2 \leq 1$.
 
\subsection{Proof of Theorem \ref{Thm:ParetoRdd}}\label{app:ParetoRdd}
The following proof is similar to the approach in Appendix \ref{app:rnd_pareto}, to avoid repetitions we only highlight the main differences in the following. The Pareto optimality problem in the DD region is written as a maximization of $R_1$ subject to $R_2 \geq r_2$ and after some manipulation, we focus on the subproblem optimizing $\mathbf{w}_1$:
\begin{equation}\label{eqt:rdd_pareto_w1}
\begin{aligned}
\max_{\mathbf{w}_1, P_1} & \hspace{1cm} t\\
\text{subject to} & \hspace{1cm} |\mathbf{h}_{11}^H \mathbf{w}_1|^2 P_1 \geq t,\\
& \hspace{1cm} |\mathbf{h}_{11}^H \mathbf{w}_1|^2 P_1  \leq \frac{z_{12}^2 P_2}{2^{r_2}-1}-1,\\
& \hspace{1cm}  | \mathbf{h}_{21}^H \mathbf{w}_1|^2 P_1 \geq t(z_{22}^2+1),\\
& \hspace{1cm} \| \mathbf{w}_1\|\leq 1, 0 \leq P_1 \leq P_{max}
\end{aligned}
\end{equation} for some $z_{12}^2, z_{22}^2$. Similar to Lemma \ref{lem:pow_gain_gen}, we write the Lagrangian and set the derivative with respect to $\mathbf{w}_1^H$  to zero, and obtian
$\left( P_1 (\lambda_1-\lambda_2) \mathbf{h}_{11}\mathbf{h}_{11}^H + \lambda_3 P_1 \mathbf{h}_{21} \mathbf{h}_{21}^H \right) \mathbf{w}_1= \lambda_4 \mathbf{w}_1$.
We add $P_1 \lambda_2 \|\mathbf{h}_{11} \|^2 \mathbf{w}_1$ to both sides and obtain
\begin{equation}
\left( P_1 \lambda_1 \|\mathbf{h}_{11}\|^2 \Pi_{11} + P_1 \lambda_2 \|\mathbf{h}_{11} \|^2 \Pi_{11}^\perp + \lambda_3 P_1 \| \mathbf{h}_{21}\|^2 \Pi_{21}\right) \mathbf{w}_1= (\lambda_4 + P_1 \lambda_2 \| \mathbf{h}_{11}\|^2) \mathbf{w}_1.
\end{equation}
Hence, we see that the optimal solution is composed of its projection on the subspace spanned by $\mathbf{h}_{11}, \mathbf{h}_{21}$ and the orthogonal subspace of $\mathbf{h}_{11}$, which can be represented by the following:
\begin{equation}
\mathbf{w}_1= \mu_1 \frac{\Pi_{11} \mathbf{h}_{21}}{\|\Pi_{11} \mathbf{h}_{21} \|} +\mu_2 \frac{\Pi_{11}^\perp \mathbf{h}_{21}}{\|\Pi_{11}^\perp \mathbf{h}_{21} \|}
\end{equation} for some $\mu_1,\mu_2 \in \mathbb{C}$ and $| \mu_1|^2+ |\mu_2|^2=1$. 

Similar to the arguments before in Appendix \ref{app:rnd_pareto}, we omit  the details here to avoid repetitions. The values $\mu_1, \mu_2$ can be chosen real-valued and  the Pareto boundary of $\mathbf{R}^{dd}$ attaining beamforming vectors are
\begin{equation}
\mathbf{w}_1 \in \mathcal{V}_1= \left\{ \mathbf{w}= \sqrt{\nu_1} \frac{\Pi_{11}\mathbf{h}_{21}}{\|\Pi_{11}\mathbf{h}_{21} \|} + \sqrt{1-\nu_1} \frac{\Pi_{11}^\perp \mathbf{h}_{21}}{\|\Pi_{11}^\perp \mathbf{h}_{21} \|}, 0 \leq \nu_1 \leq 1\right\}.
\end{equation}

\subsection{Proof of Theorem \ref{thm:sumrate_rnd}}\label{app:sumrate_rnd}
Before we go into details of computing the candidate set of the maximum sum rate in the ND region, we present the following lemma which holds importance in the discussions later.
\begin{Lem}\label{lem:rnd_solset1}
Consider two functions $f_1(x), f_2(x)$ where $f_1(x)$ is concave and $f_2(x)$ is linearly increasing with $x \in \mathcal{X} \subset \mathbb{R}$. Define
\begin{eqnarray}
x_1^* &=& \argmax_{x \in \mathcal{X}} f_1(x) , \hspace{ 0.3cm}
x_2^* = \argmax_{x \in \mathcal{X}} f_2(x) \hspace{ 0.3cm} \mbox{ and } \hspace{ 0.3cm}
\mathbb{X} = \{x \in \mathcal{X} :  f_1(x)= f_2(x) \}.
\end{eqnarray}
If $\hat{x}=\arg \max \min(f_1(x),f_2(x))$, then
\begin{equation}
\hat{x} \subset \left\{ x_1^*, x_2^*, \mathbb{X} \right\}
\end{equation}
\end{Lem}
\begin{proof}
Notice that since $f_1(x)$ is concave, there are at most two intersection points.
If there are no intersection points, then $\hat{x} \subset \{ x_1^*, x_2^*\}$. If there is one intersection point $\bar{x}$, then  one of the following orderings is true:
\begin{itemize}
\item $\bar{x} < x_1^* < x_2^*:$ since there is only one intersection point, $f_2(x_2^*)<f_1(x_2^*)$ and thus $\hat{x}=x_2^*$.
\item $x_1^* < \bar{x} < x_2^*:$ let $x^-<\bar{x}<x^+$ and we have $f_1(x^-)>f_1(\bar{x})>f_1(x^+)$ and $f_2(x^-)<f_2(\bar{x})<f_2(x^+)$. Thus, $\hat{x}=\bar{x}$.
\end{itemize}
Note that $f_2(x)$ is linearly increasing with $x$ and therefore $x_2^*$ is the boundary of $\mathcal{X}$ and thus $x_2^*> x_1^*$ and $x_2^*>\bar{x}$. If there are two intersection points, then we have $\hat{x} \subset \mathbb{X}$.
\end{proof}

Note that the global optimal solution $\boldsymbol{\omega}^*$ must be in the solutions set of $\mathcal{B}(\mathbf{R}^{nd})$:
\begin{equation}
\boldsymbol{\omega}^* \subset \Omega^{nd}
\end{equation} where $\Omega^{nd}$ is defined in Thm. \ref{thm:pareto_rnd}. Thus, we can refine the constraint set in the maximum sum rate problem in the ND region to
\begin{equation}\label{eqt:rnd_maxsumrate_def2}
\bar{R}^{nd}(\mathbf{w}_1, \mathbf{w}_2)=\max_{(\mathbf{w}_1,\mathbf{w}_2 ) \in \Omega^{nd}}C_2(\mathbf{w}_2) + \min \{ T_1(\mathbf{w}_1, \mathbf{w}_2), D_1(\mathbf{w}_1, \mathbf{w}_2)\}
\end{equation} 
which can be decomposed to the following:
\begin{equation}\label{eqt:maximin}
\max_{\mathbf{w}_2 \in \mathcal{W}_2} \left\{C_2(\mathbf{w}_2) + \max_{\mathbf{w}_1 \in \mathcal{W}_1} \min \{ T_1(\mathbf{w}_1, \mathbf{w}_2), D_1(\mathbf{w}_1, \mathbf{w}_2)\} \right\}
\end{equation} where the inner maximin problem is maximized over $\mathbf{w}_1$ for \emph{each} $\mathbf{w}_2$ . With each given $\mathbf{w}_2$, we define $c_1=\frac{P_{max}}{|\mathbf{h}_{12}^H \mathbf{w}_2|^2P_{max}+1}$ and $c_2=\frac{P_{max}}{|\mathbf{h}_{22}^H \mathbf{w}_2|^2P_{max}+1}$ and we rewrite $T_1$ and $D_1$ to the following:
$\tilde{D}_1=2^{D_1}-1=c_1 |\mathbf{h}_{11}^H \mathbf{w}_1|^2$ and $\tilde{T}_1=2^{T_1}-1=c_2 |\mathbf{h}_{21}^H \mathbf{w}_1|^2$. Thus, the maximin problem in \eqref{eqt:maximin} is equivalent to 
\begin{equation}\label{eqt:maximin2}
\max_{\mathbf{w}_1\in \mathcal{W}_1} \min \{c_1 |\mathbf{h}_{11}^H \mathbf{w}_1|^2,c_2 |\mathbf{h}_{21}^H \mathbf{w}_1|^2\}, \hspace{0.5cm} \text{for arbitrary fixed } \mathbf{w}_2 \in \mathcal{S}.
\end{equation}

\begin{Lem}\label{thm:maximin_sol}
The candidate set of the maximin problem in \eqref{eqt:maximin2} and therefore maximum sum rate problem in \eqref{eqt:maximin} can be reduced from $\mathcal{W}_1 \in \Omega^{nd}$ to $\tilde{\mathcal{W}}_1$ which is a candidate set with cardinality three containing at least the maximum sum rate solution: $|\tilde{\mathcal{W}}_1|=3$,
\begin{equation}
\tilde{\mathcal{W}}_1= \left\{ \frac{\mathbf{h}_{11}}{||\mathbf{h}_{11}||}, \frac{\mathbf{h}_{21}}{||\mathbf{h}_{21} ||}, \mathbf{w}_1(\lambda_1^{(b)}) \right\}
\end{equation} with
\begin{equation}\label{eqt:lambda_bal_1}
\lambda_1^{(b)}=\frac{c_1 ||\Pi_{21}^\perp \mathbf{h}_{11} ||^2}{c_2||\mathbf{h}_{21}||^2- 2 \sqrt{c_1c_2} |\mathbf{h}_{21}^H\mathbf{h}_{11} | + c_1 ||\mathbf{h}_{11} ||^2}.
\end{equation}
\end{Lem}
\begin{proof}
Because of the formulation of $\mathbf{w}_1 \in \mathcal{W}_1$, we can write the beamforming vector as a function of a real-valued parameter $\lambda_1$, $\mathbf{w}_1(\lambda_1)$. Using the result in Lemma \ref{lem:pow_gain_gen}, we define $f_1(\lambda_1)=c_1 |\mathbf{h}_{11}^H \mathbf{w}_1(\lambda_1)|^2$ and $f_2(\lambda_1)=c_2 |\mathbf{h}_{21}^H \mathbf{w}_1(\lambda_1)|^2$. It is easy to see that $f_1(\lambda_1)$ is concave in $\lambda_1$ and $f_2(\lambda_1)$ is linearly increasing with $\lambda_1$. The function $f_1(\lambda_1)=c_1 |\mathbf{h}_{11}^H \mathbf{w}_1(\lambda_1)|^2$ attains maximum when $w_1(\lambda_1)=\frac{\mathbf{h}_{11}}{\| \mathbf{h}_{11}\|}$. Similarly, $f_2(\lambda_2)$ attains maximum when $w_1(\lambda_1)=\frac{\mathbf{h}_{21}}{\| \mathbf{h}_{21}\|}$. Now we compute $\lambda_1^{(b)} $ which satisfies $c_1 |\mathbf{h}_{11}^H \mathbf{w}_1(\lambda_1^{(b)})|^2= c_2 |\mathbf{h}_{21}^H \mathbf{w}_1(\lambda_1^{(b)})|^2$.

To proceed, we compute the channel powers $|\mathbf{h}_{11}^H \mathbf{w}_1|^2 = \left(  \sqrt{\lambda_1^{(b)}} \| \Pi_{21} \mathbf{h}_{11} \| + \sqrt{1-\lambda_1^{(b)}} \| \Pi_{21}^\perp \mathbf{h}_{11} \| \right)^2 $ and 
$|\mathbf{h}_{21}^H \mathbf{w}_1|^2 = \lambda_1^{(b)} \|\mathbf{h}_{21}\|^2$.
Notice that  $\| \Pi_{21} \mathbf{h}_{11} \|= \|\mathbf{h}_{21} \| \cos (\phi)$ and 
 $ \| \Pi_{21}^\perp \mathbf{h}_{11} \|= \|\mathbf{h}_{21} \| \sin (\phi)$
where $|\mathbf{h}_{21}^H \mathbf{h}_{11}|= \|\mathbf{h}_{21} \| \|\mathbf{h}_{11} \| \cos(\phi)$ and by definition $\cos(\phi)$ is positive.
Rewrite $\lambda_1^{(b)}= \cos^2(\theta)$ where $0 \leq \theta \leq \pi/2$.
Thus, we can rewrite the channel powers to
\begin{equation}
\begin{aligned}
|\mathbf{h}_{11}^H \mathbf{w}_1|^2 &=\|\mathbf{h}_{11} \|^2 \cos^2(\theta-\phi),\\
|\mathbf{h}_{21}^H \mathbf{w}_1|^2 &= \cos^2(\theta) \|\mathbf{h}_{21}\|^2.
\end{aligned}
\end{equation}
Thus, $c_1 |\mathbf{h}_{11}^H \mathbf{w}_1(\lambda_1^{(b)})|^2= c_2 |\mathbf{h}_{21}^H \mathbf{w}_1(\lambda_1^{(b)})|^2$ is equivalent to
$\sqrt{c_1} \|\mathbf{h}_{11} \| \cos(\theta-\phi) =
\sqrt{c_2} \cos(\theta) \|\mathbf{h}_{21}\|$ which is due to the fact that $\cos(\theta-\phi)$ and $\cos(\theta)$ are by definition positive. Putting the sinusoids one side and we obtain $\frac{\cos(\theta-\phi)}{ \cos(\theta) } = 
\frac{\sqrt{c_2}\|\mathbf{h}_{21}\|}{\sqrt{c_1} \|\mathbf{h}_{11} \| }$. If we expand  $\cos(\theta-\phi)= \cos(\theta)\cos(\phi)+ \sin(\theta)\sin(\phi)$, we have
\begin{equation}
\tan(\theta)=\frac{\sqrt{c_2}\|\mathbf{h}_{21}\|- \sqrt{c_1} \|\mathbf{h}_{11} \|\cos(\phi)}{\sqrt{c_1} \|\mathbf{h}_{11} \|\sin(\phi)}.
\end{equation}
Use the Pythagorus theorem, if $\tan(\theta)=\frac{a}{b}$ then $\cos(\theta)= \frac{b}{\sqrt{a^2+b^2}}$ and therefore
\begin{equation}
\lambda_1^{(b)}=\cos^2(\theta)=\frac{c_1 \|\Pi_{21}^\perp \mathbf{h}_{11} \|^2}{c_2\|\mathbf{h}_{21}\|^2- 2 \sqrt{c_1c_2} |\mathbf{h}_{21}^H\mathbf{h}_{11} | + c_1 \|\mathbf{h}_{11} \|^2}.
\end{equation}
\end{proof}
If we reverse the maximization order in \eqref{eqt:rnd_maxsumrate_def2}, we obtain:
\begin{eqnarray}\label{eqt:sumrate_rnd_w2}
\bar{R}^{nd}(\mathbf{w}_1, \mathbf{w}_2)&=&\max_{\mathbf{w}_1 \in \tilde{\mathcal{W}}_1 } \max_{\mathbf{w}_2 \in \mathcal{W}_2} \min \left\{ C_2(\mathbf{w}_2) + T_1(\mathbf{w}_1,\mathbf{w}_2),  C_2(\mathbf{w}_2)+ D_1(\mathbf{w}_1, \mathbf{w}_2)\right\}.
\end{eqnarray}

\begin{Lem}\label{lem:rndw2} 
The optimal solutions to $\bar{R}^{nd}$ in \eqref{eqt:sumrate_rnd_w2} can be reduced from $\mathcal{W}_2 \in \Omega^{nd}$ to $\tilde{\mathcal{W}}_2$, a set of beamforming vectors that includes the beamforming vector towards the desired channel $\mathbf{h}_{22}$,
\begin{equation}
\tilde{\mathcal{W}}_2= \left\{ \mathbf{w}_2 \in \mathcal{S}: \mathbf{w}_2= \sqrt{\lambda_2} \frac{\Pi_{12}\mathbf{h}_{22}}{\|\Pi_{12}\mathbf{h}_{22} \|} + \sqrt{1-\lambda_2} \frac{\Pi_{12}^\perp \mathbf{h}_{22}}{\| \Pi_{12}^\perp \mathbf{h}_{22} \|} \; ; \; \lambda_2^{(b)} \leq \lambda_2 \leq \lambda_2^{\mrt} \right\}
\end{equation} 
where $\lambda_2^{\mrt}=\frac{|\mathbf{h}_{12}^H \mathbf{h}_{22}|}{||\mathbf{h}_{12} || ||\mathbf{h}_{22} ||}$ is parameter that gives the beamforming solution towards channel $\mathbf{h}_{22}$ and 
\begin{equation}
\mathbf{w}_2(\lambda_2^{(b)})=\frac{\tilde{b}}{\sqrt{\tilde{a}+\tilde{b}}}\mathbf{v}_a + \frac{e^{j \phi}\tilde{a}}{\sqrt{\tilde{a}+\tilde{b}}} \mathbf{v}_b
\end{equation} for some eigenvectors $\mathbf{v}_a, \mathbf{v}_b$ and positive scalars $\tilde{a}, \tilde{b}$. The vectors $\mathbf{v}_{a}, \mathbf{v}_b$ are the most and least dominant eigenvectors of the matrix $\mathbf{S}=\mathbf{h}_{22}\mathbf{h}_{22}^H -\frac{g_{21}}{g_{11}} \mathbf{h}_{12}\mathbf{h}_{12}^H$.
\end{Lem}
\begin{proof}
$\bar{R}^{nd}$ in \eqref{eqt:sumrate_rnd_w2} is equivalent to compute
\begin{equation}\label{eqt:maxmin_rnd_w2}
\max_{\mathbf{w}_2 | \mathbf{w}_1} \min \left\{ 1+ g_{21} + |\mathbf{h}_{22}^H \mathbf{w}_2 |^2 P_2, (1+ |\mathbf{h}_{22}^H \mathbf{w}_2|^2P_{max})\left( 1+ \frac{g_{11} }{1+ |\mathbf{h}_{12}^H \mathbf{w}_2|^2 P_{max}}\right)\right\},
\end{equation} where the notation $\max_{\mathbf{w}_2 | \mathbf{w}_1}$ denotes maximization over $\mathbf{w}_2$ for some given $\mathbf{w}_1$
and can be decomposed into the following two subproblems:
\begin{equation}\label{rndw2_sub}
\left\{ \begin{array}{cc}
\max_{\mathbf{w}_2 | \mathbf{w}_1} 1+ g_{21} + |\mathbf{h}_{22}^H \mathbf{w}_2|^2 P_{max} & \text{ if } |\mathbf{h}_{22}^H \mathbf{w}_2|^2 P_{max} \geq \frac{g_{21}}{g_{11}}-1 + \frac{g_{21}}{g_{11}}|\mathbf{h}_{12}^H \mathbf{w}_2|^2 P_{max}\\
\max_{\mathbf{w}_2 | \mathbf{w}_1} (1+ |\mathbf{h}_{22}^H \mathbf{w}_2|^2P_2)\left( 1+ \frac{g_{11}}{1+ |\mathbf{h}_{12}^H \mathbf{w}_2|^2 P_{max}}\right) & \text{ if } |\mathbf{h}_{22}^H \mathbf{w}_2|^2 P_{max} \leq \frac{g_{21}}{g_{11}}-1 + \frac{g_{21}}{g_{11}}|\mathbf{h}_{12}^H \mathbf{w}_2|^2 P_{max}
\end{array} \right. .
\end{equation}
The first subproblem has optimum solution $\frac{\mathbf{h}_{22}}{\|\mathbf{h}_{22} \|}$. For the second subproblem, the optimal $\lambda_2$ must be in the region $\Lambda=\{\lambda_2: \lambda_2^{(b)} \leq \lambda_2 \leq \lambda_2^{\mrt}\}$ where $\mathbf{w}_2(\lambda_2^{\mrt})=\frac{\mathbf{h}_{22}}{\|\mathbf{h}_{22} \|}$ and $|\mathbf{h}_{22}^H \mathbf{w}_2(\lambda_2^{(b)})|^2 P_{max} = \frac{g_{21}}{g_{11}}-1 + \frac{g_{21}}{g_{11}}|\mathbf{h}_{12}^H \mathbf{w}_2(\lambda_2^{(b)})|^2 P_{max}$. To see this, we write the metric in the second subproblem as a function of $\lambda_2$:
\begin{equation*}
F(\lambda_2)=(1+ |\mathbf{h}_{22}^H \mathbf{w}_2(\lambda_2)|^2P_{max})\left( 1+ \frac{g_{11}}{1+ |\mathbf{h}_{12}^H \mathbf{w}_2(\lambda_2)|^2 P_{max}}\right)
\end{equation*}
 Assume $\lambda^+, \lambda^- \notin \Lambda$, in particular,  $\lambda^+ \geq \lambda_2^{\mrt}$ and $\lambda^- \leq \lambda_2^{(b)}$,  we have
\begin{eqnarray}
|\mathbf{h}_{12}^H \mathbf{w}_2(\lambda^+)|^2 &=& \lambda^+ \|\mathbf{h}_{12} \|^2 \geq \lambda^{\mrt} \| \mathbf{h}_{12}\|^2\\
|\mathbf{h}_{22}^H \mathbf{w}_2(\lambda^+)|^2&\leq& |\mathbf{h}_{22}^H \mathbf{w}_2(\lambda^{\mrt})|^2 = \|\mathbf{h}_{22}\|^2.
\end{eqnarray}
Thus, any $\lambda^+$ achieves $F(\lambda^+)$ smaller than $F(\lambda^{\mrt})$. 
Note that for any $\lambda^- \leq \lambda_2^{(b)}$, $\mathbf{w}_2(\lambda_2^{-})$ is not in the constraint set of the second optimization problem. It is because $|\mathbf{h}_{22}^H \mathbf{w}_2(\lambda)|^2$ is concave in $\lambda$ and attains the maximum at $\lambda_2^{\mrt}$ and $|\mathbf{h}_{12}^H \mathbf{w}_2(\lambda)|^2$ is linearly increasing with $\lambda$. Since $\lambda_2^{(b)}< \lambda_2^{\mrt}$, for points $\lambda^- \leq \lambda_2^{(b)}$, we have $|\mathbf{h}_{22}^H \mathbf{w}_2(\lambda^-)|^2 > \frac{g_{21}}{g_{11}}-1+ \frac{g_{21}}{g_{11}} |\mathbf{h}_{12}^H \mathbf{w}_2(\lambda^-)|^2$.

Unfortunately, the direct computation of $\lambda_2^{(b)}$ is tedious and does not give much insight. Here, we provide a cleaner method of computing $\mathbf{w}_2(\lambda^{(b)})$ directly. Denote $g=\frac{g_{21}}{g_{11}}$. We compute the beamforming vector $\mathbf{w}_2$ such that 
\begin{equation}\label{eqt:balancepowergains1}
|\mathbf{h}_{22}^H \mathbf{w}_2|^2 P_2= g-1 + g |\mathbf{h}_{12}^H \mathbf{w}_2|^2 P_2.
\end{equation} Define $\mathbf{S}= \mathbf{h}_{22} \mathbf{h}_{22}^H - g \mathbf{h}_{12} \mathbf{h}_{12}$, $\tilde{\mathbf{S}}= \mathbf{S}- (g-1)\mathbf{I}$. 
\begin{equation}\label{eqt:matrixs}
\mathbf{w}_{2}^{H} \tilde{\mathbf{S}} \mathbf{w}_{2}=0
\end{equation} is a necessary condition for satisfying \eqref{eqt:balancepowergains1}.

From the definition of $\mathbf{S}$, we know that $\mathbf{S}$ is rank two with one positive eigenvalue and one negative eigenvalue \cite{Mochaourab2010}. Denote the non-zero eigenvalues of $\mathbf{S}$ by $a$ and $-b$ where $a,b >0$. Employ eigenvalue decomposition on $\mathbf{S}$ and we have
\begin{equation}
\mathbf{S}= \left[ \mathbf{v}_a  \mathbf{v}_b | \mathbf{V} \right] \left[\begin{array}{ccc}
a & 0 & \mathbf{0}_{1 \times (N-2)}\\
0 & -b & \mathbf{0}_{1 \times (N-1)}\\
0 & 0 & \mathbf{0}_{(N-2) \times (N-2)}
\end{array} \right] \left[\mathbf{v}_a  \mathbf{v}_b | \mathbf{V}\right] ^H
\end{equation} where $\mathbf{V}$ is the $N$ by $N-2$ matrix with column vectors of eigenvectors of $\mathbf{S}$ that are orthogonal to $\mathbf{v}_a$, $\mathbf{v}_b$.
With the same eigenvectors, we can write $\tilde{\mathbf{S}}$ as the following:
\begin{equation}
\tilde{\mathbf{S}}= \left[ \mathbf{v}_a  \mathbf{v}_b | \mathbf{V} \right] \left[\begin{array}{ccc}
a-(g-1) & 0 & \mathbf{0}_{1 \times (N-2)}\\
0 & -b-(g-1) & \mathbf{0}_{1 \times (N-1)}\\
0 & 0 & -(g-1)\mathbf{I}_{(N-2) \times (N-2)}
\end{array} \right] \left[\mathbf{v}_a  \mathbf{v}_b | \mathbf{V}\right] ^H
\end{equation} Let $\tilde{a}=a-(g-1), \tilde{b}=b-(g-1)$. The beamforming vector $\mathbf{w}_2$ of the following form
\begin{equation}\label{eqt:bf_family}
\mathbf{w}_2=\frac{1}{\sqrt{\tilde{a}}}\mathbf{v}_a + \frac{e^{j \phi}}{\sqrt{\tilde{b}}} \mathbf{v}_b
\end{equation} satisfies $\mathbf{w}_2^{H} \tilde{\mathbf{S}} \mathbf{w}_2=0$, where $j= \sqrt{-1}$ and $\phi$ is a phase angle between 0 to $\pi$. Notice that if $\mathbf{w}_2$ has any power on the remaining orthogonal subspace spanned by $\mathbf{V}$, $\mathbf{w}_2$ also satisfies \eqref{eqt:balancepowergains1} but the value of $|\mathbf{h}_{22}^H \mathbf{w}_2|^2$ is smaller and therefore cannot achieve the maximum sum rate. It is easy to see the result by direct computation:
$\mathbf{w}_2^{H} \mathbf{\tilde{S}} \mathbf{w}_2
=\left(\frac{1}{\sqrt{\tilde{a}}}\mathbf{v}_a^H + \frac{e^{-j \phi}}{\sqrt{\tilde{b}}} \mathbf{v}_b^H \right)  \left[ \mathbf{v}_a  \mathbf{v}_b \right] \left[\begin{array}{cc}
\tilde{a} & 0 \\
0 & -\tilde{b}
\end{array} \right] \left[\mathbf{v}_a  \mathbf{v}_b \right] ^H
\left(\frac{1}{\sqrt{\tilde{a}}}\mathbf{v}_a + \frac{e^{j \phi}}{\sqrt{\tilde{b}}} \mathbf{v}_b \right)= 0$ for given angle $\phi$.  The formulation in \eqref{eqt:bf_family} gives a family of beamforming vectors, each with a different value of $\phi$. To fine the unique $\phi$ and therefore $\mathbf{w}_2$ that maximizes sum rate, we rewrite the optimization problem in \eqref{eqt:maximin2} to
\begin{eqnarray}\label{eqt:projection_norm}
\max_{\phi} &   |\mathbf{h}_{22}^H \mathbf{w}_2(\phi) |^2\\
\nonumber \text{such that} & \mathbf{w}_2=\frac{1}{\sqrt{\tilde{a}}}\mathbf{v}_a + \frac{e^{j \phi}}{\sqrt{\tilde{b}}} \mathbf{v}_b.
\end{eqnarray} 
Define the following phase angles,  $\phi_a= \arg(\mathbf{h}_{22}^H \mathbf{v}_a), \phi_b= \arg(\mathbf{h}_{22}^H \mathbf{v}_b)$ and therefore 
$$\phi_m= \arctan \left(\frac{Im(\mathbf{h}_{22}^H \mathbf{v}_m)}{Re(\mathbf{h}_{22}^H \mathbf{v}_m)} \right) + \left\{\begin{array}{cc} \pi & Re(\mathbf{h}_{22}^H \mathbf{v}_m)<0\\ 0 & \text{ otherwise.} \end{array}\right. , \hspace{0.5cm} m=a,b.
$$ The optimization problem in  \eqref{eqt:projection_norm} is therefore equivalent to $$\max_{\phi}  |\mathbf{h}_{22}^H \mathbf{w}_2 |^2  = \max_{\phi}   \left| \frac{1}{\sqrt{\tilde{a}}} |\mathbf{h}_{22}^H \mathbf{v}_a| e^{j \phi_a} + \frac{e^{j \phi}}{\sqrt{\tilde{b}}} |\mathbf{h}_{22}^H \mathbf{v}_b| e^{j \phi_b} \right|^2 = \max_{\phi}  \left| \frac{1}{\sqrt{\tilde{a}}} |\mathbf{h}_{22}^H \mathbf{v}_a| + \frac{e^{j (\phi +\phi_b -\phi_a)}}{\sqrt{\tilde{b}}} |\mathbf{h}_{22}^H \mathbf{v}_b|  \right|^2.$$ Thus, the optimal phase angle $\phi$ is 
\begin{equation}
\phi=\phi_a-\phi_b.
\end{equation}

To satisfy the norm constraint, we can scale the beamforming vector with a positive scalar, which does not change the direction of the vector.
\begin{equation}
\mathbf{w}_2=\frac{\tilde{b}}{\sqrt{\tilde{a}+\tilde{b}}}\mathbf{v}_a + \frac{e^{j \phi}\tilde{a}}{\sqrt{\tilde{a}+\tilde{b}}} \mathbf{v}_b.
\end{equation}

\end{proof}
Note that if $\lambda_2^{b}>\lambda_2^{mrt}$, then $\tilde{\mathcal{W}}_2$ has only one  element, i.e. $\frac{\mathbf{h}_{22}}{\|\mathbf{h}_{22} \|}$. Lemma \ref{thm:maximin_sol}  gives the optimal candidate set for $\mathbf{w}_1$ for arbitrary fixed $\mathbf{w}_2$ whereas Lemma \ref{lem:rndw2} gives the optimal candidate set for $\mathbf{w}_2$ for arbitrary fixed $\mathbf{w}_1$. Combining both Lemmas, we obtain the maximum sum rate candidate sets for both $\mathbf{w}_1$ and $\mathbf{w}_2$.

\begin{Remark}
The authors in \cite{Tomecki2010} provided a general solution of \eqref{eqt:maximin2}, in the context of a multicast SNR balancing problem. The authors transformed the channel powers balancing problem to a weighted sum channel powers maximization problem for some positive weights $w_1, w_2$. The optimal beamforming vector is then characterized as a dominant eigenvector of some matrices, depending on $w_1,w_2$. However, the computation of such weights $w_1,w_2$ is not provided or trivial. In this paper, due to the beamforming vectors parameterization proposed in Thm. \ref{thm:pareto_rnd} and \ref{Thm:ParetoRdd}, we obtained the closed form solution of such channel powers balancing beamforming vectors.  
\end{Remark}

\subsection{Proof of Thm. \ref{thm:rddsumrate}}\label{appendix:sumraterdd}
In this section, we provide the proof of Thm. \ref{thm:rddsumrate}. We start by identifying four constraint sets of beamforming vectors $\Omega^{00}, \Omega^{01}, \Omega^{10}, \Omega^{11}$ where the sum rate function is a different function in each set. In other words, when the beamforming vectors vary, the sum rate function being a sum of two minimum of rate functions, may change from one rate expression to another rate expression. The constraint sets are the set of beamforming vectors for which the sum rate function remains at one rate function. In the following, we provide the analysis for $\mathbf{w}_2$ but the sum rate function $\bar{R}^{dd}$ is symmetric with Tx 1 and 2. Therefore, we can exchange the role of Tx 1 and 2 and obtain the candidate sets for $\mathbf{w}_1$. The sum rate in the DD region is: 
\begin{equation}\label{eqt:def_rdd}
\begin{aligned}
\bar{R}^{dd}&= \min \{ C_1, T_1\} + \min \{ C_2+ T_2\}\\
&= \min\{ \underbrace{C_1+C_2}_{Z_1}, \underbrace{C_1+T_2}_{Z_2}, 
\underbrace{T_1+C_2}_{Z_3}, \underbrace{T_1+T_2}_{Z_4}\}.
\end{aligned}
\end{equation}
We analyze each term and define $\tilde{Z}_i=2^{Z_i}, i=1,\ldots, 4$ 
\begin{equation}
\begin{aligned}
\tilde{Z}_1&= (1+ |\mathbf{h}_{11}^H \mathbf{w}_1|^2P) (1+ |\mathbf{h}_{22}^H \mathbf{w}_2|^2P)\\
\tilde{Z}_2&= 1+ |\mathbf{h}_{11}^H \mathbf{w}_1|^2 P + |\mathbf{h}_{12}^H \mathbf{w}_2|^2P\\
\tilde{Z}_3&= 1+ |\mathbf{h}_{21}^H \mathbf{w}_1|^2 P + |\mathbf{h}_{22}^H \mathbf{w}_2|^2P\\
\tilde{Z}_4&= \left( 1+ \frac{ |\mathbf{h}_{21}^H \mathbf{w}_1|^2P}{1+  |\mathbf{h}_{22}^H \mathbf{w}_2|^2P} \right) \left(1 + \frac{ |\mathbf{h}_{12}^H \mathbf{w}_2|^2P}{1+ |\mathbf{h}_{11}^H \mathbf{w}_1|^2 P} \right)\\
&= \frac{\tilde{Z}_2 \tilde{Z}_3}{\tilde{Z}_1}.
\end{aligned}
\end{equation}

Notice that $\tilde{Z}_1< \tilde{Z}_2$ is equivalent to $\tilde{Z}_3<\tilde{Z}_4$ and $\tilde{Z}_1< \tilde{Z}_3$ is equivalent to $\tilde{Z}_2<\tilde{Z}_4.$ We summarize to the following lemma.
\begin{Lem}\label{eqt:ineq_rdd}
Let $g_{11}=|\mathbf{h}_{11}^H \mathbf{w}_1|^2 P_1$ and $g_{21}=|\mathbf{h}_{21}^H \mathbf{w}_1|^2 P_1$.
\begin{eqnarray*}
\tilde{Z}_1\leq \tilde{Z}_2 & \Leftrightarrow & \tilde{Z}_3 \leq \tilde{Z}_4\\
& \Leftrightarrow &  (1+ g_{11}) |\mathbf{h}_{22}^H \mathbf{w}_2|^2  \leq |\mathbf{h}_{12}^H \mathbf{w}_2|^2 \\
\tilde{Z}_1\leq \tilde{Z}_3 & \Leftrightarrow & \tilde{Z}_2 \leq \tilde{Z}_4\\
& \Leftrightarrow & |\mathbf{h}_{22}^H \mathbf{w}_2|^2 P_2 \leq \frac{g_{21}}{g_{11}}-1
\end{eqnarray*}
\end{Lem}
\begin{proof}
It is by direct manipulation of the definitions.
\end{proof} 
To facilitate representation, we denote the following two indicators and the corresponding candidate sets:
\begin{eqnarray}\label{eqt:rdd_constraints}
A&=& \left\{ \begin{array}{cc} 1 & \text{ if     } (1+ g_{11}) |\mathbf{h}_{22}^H \mathbf{w}_2|^2 P_2 \leq |\mathbf{h}_{12}^H \mathbf{w}_2|^2 P_2\\
0 & \text{ otherwise. }\end{array} \right. \\
B&=& \left\{ \begin{array}{cc} 1 & \text{ if     } |\mathbf{h}_{22}^H \mathbf{w}_2|^2 P_2 \leq \frac{g_{21}}{g_{11}}-1 \\
0& \text{ otherwise.} \end{array} \right.\\
\label{eqt:omega} \Omega^{ab}&=&\left\{ \mathbf{w}_2: \| \mathbf{w}_2\|=1, A=a,B=b\right\}
\end{eqnarray}
Notice that $\Omega^{ab}, a,b= 0,1,$ gives four constraint sets. By Lem. \ref{eqt:ineq_rdd}, we can decompose the optimization problem in \eqref{eqt:def_rdd} to the following:
\begin{eqnarray*}
\bar{R}^{dd}&=& \log_2 \left( \max_{\mathbf{w}_1 \in \mathcal{S}} \max_{\mathbf{w}_2 \in \mathcal{S} | \mathbf{w}_1} \min \{ \tilde{Z}_1, \tilde{Z}_2,\tilde{Z}_3,\tilde{Z}_4 \} \right)\\
&=& \left\{ \begin{array}{cc}
\log_2 \left( \max_{\mathbf{w}_1} \max_{\mathbf{w}_2 \in \Omega^{11}} \tilde{Z}_1 \right) & \text{ if }  A=1, B=1\\
\log_2 \left(\max_{\mathbf{w}_1} \max_{\mathbf{w}_2 \in \Omega^{01}} \tilde{Z}_2 \right) & \text{ if } A=0, B=1\\
\log_2 \left(\max_{\mathbf{w}_1} \max_{\mathbf{w}_2 \in \Omega^{10}} \tilde{Z}_3 \right) & \text{ if } A=1, B=0\\
\log_2 \left(\max_{\mathbf{w}_1} \max_{\mathbf{w}_2\in \Omega^{00}} \tilde{Z}_4 \right) & \text{ if } A=0, B=0
\end{array} \right. .
\end{eqnarray*}

Now, we proceed with the proof of  Thm. \ref{thm:rddsumrate} in two parts: first, in Section \ref{section:rddoptsol}, we identify the  candidate sets for each of the subproblems $\tilde{Z}_i, i=1, \ldots, 4$; second, in Section \ref{thm:rddoptsolcom}, we combine these candidate sets to one superset by eliminating beamforming vectors, which are on the boundary of the constraint sets, if they achieve smaller sum rate than other beamforming vectors. 
\subsubsection{The candidate sets of subproblems $\tilde{Z}_i$} \label{section:rddoptsol}
The candidate sets for each of the subproblem is as follows:
\begin{itemize}
\item If $\mathbf{w}= \arg \max_{\mathbf{w}_2 \in \Omega^{11}} \tilde{Z}_1$, then $\mathbf{w} \in \Omega^{A} \bigcup \Omega^{B} \bigcup \frac{\mathbf{h}_{22}}{\|\mathbf{h}_{22} \|}$
\item If $\mathbf{w}= \arg \max_{\mathbf{w}_2 \in \Omega^{01}} \tilde{Z}_2$ , then $\mathbf{w} \in \mathbf{w}_2(\lambda_2^A) \bigcup \mathbf{w}_2^{AB} \bigcup \frac{\mathbf{h}_{12}}{\|\mathbf{h}_{12} \|}$.
\item If $\mathbf{w}= \arg \max_{\mathbf{w}_2 \in \Omega^{10}} \tilde{Z}_3$, then
$\mathbf{w} \in \mathbf{w}_2(\lambda_2^A) \bigcup  \frac{\mathbf{h}_{22}}{\|\mathbf{h}_{22} \|}$ .
\item If $\mathbf{w}= \arg \max_{\mathbf{w}_2 \in \Omega^{00}} \tilde{Z}_4$, then
$\mathbf{w} \in \Omega^{A} \bigcup \tilde{\mathcal{W}}_{2}^{dd}$
\end{itemize}
where the constraints $\Omega^A$ and $\Omega^B$ are the set of beamforming vectors that satisfy the constraints by equality $\Omega^{A}= \left\{ \mathbf{w}_2: (1+ g_{11}) |\mathbf{h}_{22}^H \mathbf{w}_2|^2= |\mathbf{h}_{12}^H \mathbf{w}_2|^2 \right\}$ and 
 $\Omega^{B}= \left\{ \mathbf{w}_2:  |\mathbf{h}_{22}^H \mathbf{w}_2|^2 P_2 = \frac{g_{21}}{g_{11}}-1 \right\}$. The beamforming vector $\mathbf{w}_2(\lambda_2^A)$ is the beamforming vector in $\Omega^A$ that maximizes the desired channel power, $\mathbf{w}_2(\lambda^{A}_2)= \arg \max_{\mathbf{w}_2 \in \Omega^{A}} (1+ g_{11}) |\mathbf{h}_{22}^H \mathbf{w}_2|^2$. The beamforming vector $\mathbf{w}_2^{AB}$ is a unique vector that is a member of both $\Omega^A $ and $\Omega^B$, $\mathbf{w}_2^{AB}= \Omega^{A} \bigcap \Omega^{B}$. Lastly, we have $\tilde{\mathcal{V}}_2$ which is a subset of $\mathcal{V}_2$,
$\tilde{\mathcal{V}}_2=\left\{ \mathbf{w}_2: \sqrt{\lambda_2} \frac{\Pi_{22}\mathbf{h}_{12}}{\|\Pi_{22}\mathbf{h}_{12} \|}+ \sqrt{1-\lambda_2} \frac{\Pi_{22}^\perp \mathbf{h}_{12}}{\|\Pi_{22}^\perp \mathbf{h}_{12} \|} , \lambda_2^{A} \leq \lambda_2 \leq \lambda_2^{\mrt} \right\} $ where $\lambda_2^{\mrt}=\frac{|\mathbf{h}_{22}^H\mathbf{h}_{12}|^2}{\|\mathbf{h}_{22}\|^2 \|\mathbf{h}_{12}\|^2}$.

Notice that beamforming vectors in $\Omega^{A}, \Omega^{B}$ satisfy the constraints in  \eqref{eqt:rdd_constraints} with equality. To see this, we have the following observations:

1. $\tilde{Z}_1$ is monotonically increasing with $|\mathbf{h}_{22}^H \mathbf{w}_2|^2$. If the constraints are not active, the optimal solution is $\frac{\mathbf{h}_{22}}{\|\mathbf{h}_{22} \|}$. If the constraints are active, $\tilde{Z}_1$ is maximized over constraint set $\Omega^{11}$ and therefore the optimal solutions are in $\Omega^{A} \bigcup \Omega^{B}$. 

2. $\tilde{Z}_2$ is monotonically increasing with $|\mathbf{h}_{12}^H \mathbf{w}_2|^2$in constraint set $\Omega^{01}$. If the constraints are not active, the optimal solution is $\frac{\mathbf{h}_{12}}{\|\mathbf{h}_{12} \|}$.  There are an upper bound on $|\mathbf{h}_{12}^H \mathbf{w}_2|^2$ and an upper bound on $|\mathbf{h}_{22}^H \mathbf{w}_2|^2$ which in turn upper bound $|\mathbf{h}_{12}^H \mathbf{w}_2|^2$. If the constraint for $|\mathbf{h}_{12}^H \mathbf{w}_2|^2$ is active, the solution is $\mathbf{w}_2(\lambda_2^A)$. If both are active, the solution is $\mathbf{w}_2^{AB}$. Thus, the candidate set is $\left\{ \mathbf{w}_2(\lambda_2^A), \mathbf{w}_2^{AB}, \frac{\mathbf{h}_{12}}{\|\mathbf{h}_{12} \|}\right\}.$

3. $\tilde{Z}_3$ is monotonically increasing with $|\mathbf{h}_{22}^H \mathbf{w}_2|^2$ in constraint set $\Omega^{10}$. There is only one constraint that upper bound the value of $|\mathbf{h}_{22}^H \mathbf{w}_2|^2$, the optimal solution is in set $\Omega^{A}$ which at the same time maximizes $|\mathbf{h}_{22}^H \mathbf{w}_2|^2$, denote as $\mathbf{w}_2(\lambda_2^A)$. If the constraints are not active, we have $\frac{\mathbf{h}_{22}}{\| \mathbf{h}_{22}\|}$.

4. $\tilde{Z}_4$ is monotonically increasing with $|\mathbf{h}_{12}^H \mathbf{w}_2|^2$ and decreasing with $|\mathbf{h}_{22}^H \mathbf{w}_2|^2$ in constraint set $\Omega^{00}$. If the constraints are active, the optimal solutions are in $\Omega^{A}$. If the constraints are not active, the optimal solutions are in $\tilde{\mathcal{V}}_2$. Similar to the case of $\bar{\mathcal{R}}^{nd}$, any $\lambda \notin (\lambda_2^{A} \leq \lambda_2 \leq \lambda_2^{\mrt})$ cannot attain maximum value of $\tilde{Z}_4$ in $\Omega^{00}$. Notice that for any $\lambda \geq \lambda^{mrt}$, $|\mathbf{h}_{12}^H \mathbf{w}_2(\lambda)|^2 \leq |\mathbf{h}_{12}^H \mathbf{w}_2(\lambda_2^{\mrt})|^2$ and $|\mathbf{h}_{22}^H \mathbf{w}_2(\lambda)|^2 =\lambda \| \mathbf{h}_{12}\|^2 \geq \lambda_2^{\mrt} \| \mathbf{h}_{12}\|^2 = |\mathbf{h}_{12}^H \mathbf{w}_2(\lambda_2^{mrt})|^2$ . Thus, $\tilde{Z}_4(\lambda)\leq \tilde{Z}_4(\lambda_2^{\mrt}),$  for any $ \lambda \geq \lambda_2^{\mrt}$.
Also, for any $\lambda < \lambda_2^{A}$, $(1+g_{11})|\mathbf{h}_{22}^H \mathbf{w}_2(\lambda)|^2 \leq |\mathbf{h}_{12}^H \mathbf{w}_2(\lambda)|^2$. It is because $|\mathbf{h}_{12}^H \mathbf{w}_2(\lambda)|^2$ is concave in $\lambda$ and attains maximum at $\lambda_2^{\mrt}$ whereas $|\mathbf{h}_{22}^H \mathbf{w}_2(\lambda)|^2$ is linearly increasing with $\lambda$. Since the intersection point $\lambda_2^A \leq \lambda_2^{mrt}$, we have $(1+g_{11})|\mathbf{h}_{22}^H \mathbf{w}_2(\lambda)|^2 \leq |\mathbf{h}_{12}^H \mathbf{w}_2(\lambda)|^2$ which violates the constraint set requirement.

\subsubsection{Eliminating non-sum-rate optimal solutions}\label{thm:rddoptsolcom}
Now, we  combine the above results. For any solutions $\omega \in \Omega^A$, we have $\omega \in \Omega^{0b}$ and $\omega \in \Omega^{1b}$ for $b=0,1$. This is because $\omega$ is on the boundary separating $\Omega^{0b}, \Omega^{1b}$. Thus, if $\omega$ is sum rate optimal in $\Omega^{0b}$ but  \emph{not} sum rate optimal in  $\Omega^{1b}$, then the sum rate optimal solutions in $\Omega^{1b}$ achieves a higher sum rate than $\omega$ which is in the \emph{same} constraint set. Then, $\omega$ can be removed from the candidate sets of the maximum sum rate point over constraint sets $\Omega^{0b} \bigcup \Omega^{1b}$. Applying this argument, we combine the following:

For any $\omega \in \Omega^A$ which maximizes $\tilde{Z}_4$ in $\Omega^{00}$ is also in $\Omega^{10}$ and achieves a \emph{smaller} $\tilde{Z}_3$ than other sum rate optimal solutions in $\Omega^{10}$, namely $\mathbf{w}_2(\lambda_2^A) \bigcup \frac{\mathbf{h}_{22}}{\|\mathbf{h}_{22} \|}$. Thus, we have: 
\begin{equation}
\text{ If } \omega= \arg \max_{\mathbf{w}_2 \in \Omega^{10} \bigcup \Omega^{00}} \min\left\{ \tilde{Z}_3,\tilde{Z}_4\right\},   \text{ then } \omega \in \mathbf{w}_2(\lambda_2^A) \bigcup \frac{\mathbf{h}_{22}}{\|\mathbf{h}_{22} \|} \bigcup \tilde{\mathcal{V}}_2^{dd}
\end{equation}

Similarly, for any $\omega \in \Omega^B$ which maximizes $\tilde{Z}_2$ in the constraint set $\Omega^{01}$ is also in $\Omega^{00}$ and achieves a smaller $\tilde{Z}_4$ than other solutions. Thus, we have:
 \begin{equation}
\text{ If } \omega= \arg \max_{\mathbf{w}_2 \in \Omega^{01} \bigcup \Omega^{10} \bigcup \Omega^{00}} \min\left\{ \tilde{Z}_2,\tilde{Z}_3,\tilde{Z}_4\right\}, \text{ then }  \omega \in \mathbf{w}_2(\lambda_2^A) \bigcup \frac{\mathbf{h}_{22}}{\|\mathbf{h}_{22} \|} \bigcup \tilde{\mathcal{V}}_2^{dd}.
\end{equation} The candidate set remains unchanged because $\mathbf{w}_2^{AB}$ performs worse than other solutions and $\frac{\mathbf{h}_{12}}{\| \mathbf{h}_{12}\|} \in \tilde{\mathcal{V}}_2^{dd}$.

 Lastly,  for any $\omega_b \in \Omega^B$ which maximizes $\tilde{Z}_1$ in constraint set $\Omega^{11}$ is also in $\Omega^{10}$ and achieves a smaller $\tilde{Z}_3$ than other solutions. And for $\omega_a \in \Omega^A$ which maximizes $\tilde{Z}_1$ in constraint set $\Omega^{11}$ is also in $\Omega^{01}$ and achieves a less $\tilde{Z}_2$ than other solutions. Thus, we have:
 \begin{equation}
\text{ If } \omega= \arg \max_{\mathbf{w}_2 \in \Omega^{11} \bigcup \Omega^{01} \bigcup \Omega^{10} \bigcup \Omega^{00}} \min\left\{ \tilde{Z}_1, \tilde{Z}_2,\tilde{Z}_3,\tilde{Z}_4\right\},  \text{ then }  \omega \in \mathbf{w}_2(\lambda_2^A) \bigcup \frac{\mathbf{h}_{22}}{\|\mathbf{h}_{22} \|} \bigcup \tilde{\mathcal{V}}_2^{dd}.
\end{equation} 

Therefore, the final candidate sets are
\begin{equation}
\mathcal{V}_i^{dd}=\left\{  \frac{\mathbf{h}_{ii}}{\|\mathbf{h}_{ii} \|}, \tilde{\mathcal{V}}_i^{dd}, \mathbf{w}_i(\lambda_i^A)\right\}.
\end{equation} This result holds for $\mathbf{w}_1$ because the optimization problem is symmetric.

To compute the closed form $\lambda_2^A$ where $\mathbf{w}_2(\lambda_2^A)$ balances channel powers, we can use the same approach as before to obtain
$\lambda_2^{A}=\frac{\|\Pi_{22}^\perp \mathbf{h}_{12} \|}{\|\mathbf{h}_{12}\|^2 + (1+g_{11})\|\mathbf{h}_{22} \|^2-2 |\mathbf{h}_{22}^H \mathbf{h}_{12}| \sqrt{1+g_{11}}}$.

\subsection{Proof of MRT optimality conditions in the ND region}\label{app:mrtopt}

Define the following two functions in $0 \leq \lambda_2 \leq 1$:
\begin{itemize}
\item $F_1(\lambda_2)$ is a concave function in $\lambda_2$ and attains maximum at $\lambda_2^{\mrt}$.
\item $F_2(\lambda_2)$ is an arbitrary function in $\lambda_2$ but satisfy the following properties:
\begin{itemize}
\item for any $\lambda>\lambda_2^{\mrt}$, $F_2(\lambda)< F_2(\lambda_2^{\mrt})$.
\item $\left. \frac{\partial F_2(\lambda)}{\partial \lambda_2} \right|_{\lambda_2=\lambda_2^{\mrt}}<0 $.
\end{itemize}
\end{itemize}
\begin{Lem}\label{lem:rnd_mrt_opt}
The condition $F_2(\lambda_2^{\mrt})\geq F_1(\lambda_2^{\mrt})$ is a necessary and sufficient condition for 
\begin{equation}
\lambda_2^{\mrt}=\argmax_{0 \leq \lambda_2 \leq 1} \min\{ F_1(\lambda_2), F_2(\lambda_2)\}.
\end{equation}
\end{Lem}

\begin{proof}
\begin{description}
\item[ \text{``}$\Rightarrow$\text{''}:] Let $\lambda_2^{*}=\argmax_{0 \leq \lambda_2 \leq 1} \min\{ F_1(\lambda_2), F_2(\lambda_2)\}$. Denote a set $\Lambda$
\begin{equation}
\Lambda=\{ \lambda_2: F_1(\lambda_2)\leq F_2(\lambda_2) \}
\end{equation} Let $\lambda_2'\in \Lambda$ and therefore  $F_1(\lambda_2')\leq F_1(\lambda_2^*)\leq F_1(\lambda_2^{\mrt})$. The second inequality is due to the fact that $F_1(\lambda_2)$ attains maximum at $\lambda_2^{\mrt}$. Note that $F_1(\lambda_2^{\mrt}) \leq F_2(\lambda_2^{\mrt})$ by assumption, thus $\lambda_2^{\mrt}\in \Lambda$ and we can write
\begin{equation}
F_1(\lambda_2^{\mrt}) \leq F_1(\lambda_2^*) \leq F_1(\lambda^{\mrt}_2).
\end{equation}
Since $F_1(\lambda_2)$ is a concave and has unique maximum at $\lambda^{\mrt}_2$, we have $\lambda_2^*=\lambda_2^{\mrt}$.
\item[ \text{``}$\Leftarrow$\text{''}:] We start with $\lambda_2^{\mrt}=\argmax_{\lambda_2} \min\{ F_1(\lambda_2), F_2(\lambda_2)\}$ and we proceed with contradiction. Assume $F_1(\lambda_2^{\mrt})>F_2(\lambda_2^{\mrt})$. Since $\left. \frac{\partial F_2(\lambda)}{\partial \lambda_2}  \right|_{\lambda_2=\lambda_2^{\mrt}}<0 $, there exist $\lambda_2^- =\lambda_2^{\mrt}-\epsilon$ with arbitrary small $\epsilon>0$, which satisfies $F_1(\lambda_2^{-})\geq F_2(\lambda_2^-)\geq F_2(\lambda_2^{\mrt})$ and contradicts  to the assumption that $\lambda_2^{\mrt}=\argmax_{\lambda_2} \min\{ F_1(\lambda_2), F_2(\lambda_2)\}$ .  
\end{description}
\end{proof}

Note that the sum rate in the ND region is
\begin{equation}
\bar{R}^{nd}=\max_{0 \leq \lambda_1,\lambda_2 \leq 1} \min\{  C_2(\lambda_2)+T_1(\lambda_1,\lambda_2), C_2(\lambda_2)+D_1(\lambda_1,\lambda_2)\}.
\end{equation} From Thm. \ref{thm:pareto_rnd}, we can write the Pareto optimal beamforming vectors $\mathbf{w}_i$ in the ND region as a function of the real valued parameter $\lambda_i$ in \eqref{eqt:bf_set_w}. Hence, we can rewrite the rate expressions in \eqref{eqt:ch4_rate_def} as functions of $\lambda_i, i=1,2$:
\begin{equation}
\begin{aligned}
2^{C_2(\lambda_2)+ T_1(\lambda_1,\lambda_2)}&= 1+ g_2(\lambda_2) P_{max} + \lambda_1 \| \mathbf{h}_{21} \|^2 P_{max}\\
2^{C_2(\lambda_2)+ D_1(\lambda_1,\lambda_2)}&= \left(1+ g_2(\lambda_2) P_{max} \right) \left( 1+ \frac{g_1(\lambda_1) P_{max}}{1+ \lambda_2 \| \mathbf{h}_{12}\|^2 P_{max}}\right)\\
\end{aligned}
\end{equation} where
$g_1(\lambda_1)= \left( \sqrt{\lambda_1} \|\Pi_{21}\mathbf{h}_{11}\| +\sqrt{1-\lambda_1} \|\Pi_{21}^\perp \mathbf{h}_{11} \| \right)^2$ and
$g_2(\lambda_2)= \left( \sqrt{\lambda_2} \|\Pi_{12}\mathbf{h}_{22} \|+\sqrt{1-\lambda_2} \|\Pi_{12}^\perp \mathbf{h}_{22} \| \right)^2$.
Since logarithm function is monotonic, it does not change the maximization solution and from now on, we consider maximizing the minimum of the following two functions,
\begin{eqnarray}
F_1(\lambda_1,\lambda_2)&=&1+ g_2(\lambda_2) P_{max} + \lambda_1 \| \mathbf{h}_{21} \|^2 P_{max}\\
F_2(\lambda_1,\lambda_2)&=& \left(1+ g_2(\lambda_2)  P_{max}\right) \left( 1+ \frac{g_1(\lambda_1) P_{max}}{1+ \lambda_2 \| \mathbf{h}_{12}\|^2 P_{max}}\right).
\end{eqnarray}

\begin{Lem}
It can be shown that the function $g_i(\lambda_i)$ is concave in $\lambda_i$ for $i=1,2$. $F_1(\lambda_1,\lambda_2)$ is concave in $\lambda_2$ and attains its maximum at $\lambda_2^{\mrt}=\frac{\| \Pi_{21}\mathbf{h}_{11}\|^2}{ \|\mathbf{h}_{21} \|^2 \| \mathbf{h}_{11}\|^2}$.
\end{Lem}
\begin{proof}
Note that the first and second derivatives of $g_i(\lambda_i)$ with respect to $\lambda_i$ are $\frac{\partial }{\partial \lambda_i} g_i(\lambda_i)= \| \Pi_{ji}\mathbf{h}_{ii}\|^2 - \| \Pi_{ji}^\perp\mathbf{h}_{ii}\|^2 + \| \Pi_{ji}\mathbf{h}_{ii}\| \| \Pi_{ji}^\perp \mathbf{h}_{ii}\| \left( \frac{1-2 \lambda_i}{\sqrt{\lambda_i} \sqrt{1-\lambda_i}}\right)$ and $ \frac{\partial^2 }{\partial \lambda_i^2} g_i(\lambda_i)= -\frac{\| \Pi_{ji}\mathbf{h}_{ii}\| \| \Pi_{ji}^\perp \mathbf{h}_{ii}\|}{2 \lambda_i^{3/2} (1-\lambda_i)^{3/2}}$.
Since $\lambda_i$ is  between zero and one, the second derivative of $g_i(\lambda_i)$ is always negative, for all $\lambda_i$. Set the first derivative to zero and we obtain the maximum $\lambda_2^{\mrt}=\frac{\| \Pi_{21}\mathbf{h}_{11}\|^2}{ \|\mathbf{h}_{21} \|^2 \| \mathbf{h}_{11}\|^2}$.
\end{proof}
\begin{Lem}
$F_2(\lambda_1,\lambda_2)$ satisfies 
\begin{equation*}
\left. \frac{\partial F_2(\lambda_1,\lambda_2)}{\partial \lambda_2}  \right|_{\lambda_2=\lambda_2^{\mrt}}<0 .
\end{equation*}
\end{Lem}
\begin{proof}
\begin{eqnarray*}
 \left. \frac{\partial }{\partial \lambda_2} F_2(\lambda_1,\lambda_2) \right|_{\lambda_2=\lambda_2^{\mrt}} 
&=& \left. \left( \frac{\partial}{\partial \lambda_2} \left(1+ g_2(\lambda_2) P_{max} \right) \right)   \left( 1+ \frac{g_1(\lambda_1) P_{max}}{1+ \lambda_2 \| \mathbf{h}_{12}\|^2 P_{max}}\right) \right|_{\lambda_2=\lambda_2^{\mrt} }\\
& &  + \left.  \left(1+ g_2(\lambda_2) P_{max}\right) \left(\frac{\partial}{\partial \lambda_2} \left( 1+ \frac{g_1(\lambda_1)P_{max}}{1+ \lambda_2 \| \mathbf{h}_{12}\|^2 P_{max}} \right) \right) \right|_{\lambda_2=\lambda_2^{\mrt}} \\
&\overset{(a)}{=}&  \left. \left(1+ g_2(\lambda_2) P_{max}\right) \left(\frac{\partial}{\partial \lambda_2} \left( 1+ \frac{g_1(\lambda_1)P_{max}}{1+ \lambda_2 \| \mathbf{h}_{12}\|^2P_{max}}\right) \right)\right|_{\lambda_2=\lambda_2^{\mrt}} \\
&=& -\left(1+ \| \mathbf{h}_{22}\|^2 P_{max}\right) \frac{g_1(\lambda_1) \| \mathbf{h}_{12}\|^2 P_{max}}{(1+ \lambda_2^{\mrt} \| \mathbf{h}_{12}\|^2 P_{max} )^2 }\\
&<& 0 \text{ for any } 0\leq \lambda_1 \leq 1
\end{eqnarray*}
where $(a)$ is due to the fact that $g(\lambda_2)$ is concave and attains its maximum at $\lambda_2^{\mrt}$.
\end{proof}
\begin{Lem}
For any $\lambda> \lambda_2^{\mrt}$, $F_2(\lambda)< F_2(\lambda_2^{\mrt})$.
\end{Lem}
\begin{proof}
For any $\lambda$, $g_2(\lambda) \leq g_2(\lambda_2^{\mrt})$ because $g_2(.)$ is a concave function and attains maximum at $\lambda_2^{\mrt}$. Also, for any $\lambda > \lambda_2^{\mrt}$, the denominator of $F_2(\lambda_2)$ increases. Thus, for any $\lambda> \lambda_2^{\mrt}$, $F_2(\lambda)< F_2(\lambda_2^{\mrt})$.
\end{proof}

By Lemma \ref{lem:rnd_mrt_opt}, $\bar{R}^{nd}$ is maximized by $\lambda_2^{\mrt}$ for arbitrary fixed $\lambda_1$ if and only if $F_1(\lambda_1,\lambda_2^{\mrt}) \leq F_2(\lambda_1,\lambda_2^{\mrt})$ which is equivalent to the following:
\begin{eqnarray}
F_1(\lambda_1,\lambda_2^{\mrt}) \leq F_2(\lambda_1,\lambda_2^{\mrt})
&\Leftrightarrow & \lambda_1 \| \mathbf{h}_{21}\|^2 \leq \frac{1+ \| \mathbf{h}_{22}\|^2 P_{max}}{1+ \| \mathbf{h}_{12}\|^2 \cos^2(\theta_2)P_{max}} g_1(\lambda_1).
\end{eqnarray}
Also, for arbitrary fixed $\lambda_2$, $F_1(\lambda_1,\lambda_2)$ is linearly increasing with $\lambda_1$ and $F_2(\lambda_1,\lambda_2)$ is concave in $\lambda_1$ and attains maximum at $\lambda_1^{\mrt}$. Similar to the argument before, there are at most 2 intersection points between $F_1(\lambda_1,\lambda_2)$ and $F_2(\lambda_1,\lambda_2)$. We observe that  
\begin{equation}
\begin{aligned}
F_1(0,\lambda_2)&= 1+ g_2(\lambda_2)P_{max}\\
&< \left(1+g_2(\lambda_2)P_{max} \right) \left( 1+ \frac{g_1(0)P_{max}}{1+ \lambda_2 \|\mathbf{h}_{12} \|^2 P_{max}}\right)\\
&= F_2(0,\lambda_2).
\end{aligned}
\end{equation} Note that $g_1(0)>0$ except when $\mathbf{h}_{11}$ is orthogonal to $\mathbf{h}_{21}$ whose probability is zero almost surely.
Since $F_1(0,\lambda_2) < F_2(0, \lambda_2)$ for any $\lambda_2$, there is at most 1 intersection point.
If there is no intersection point, the curve $F_2(\lambda_1, \lambda_2)$ is above $F_1(\lambda_1, \lambda_2)$ for any $\lambda_1$ and therefore the optimal value of $\lambda_1$ which maximizes $F_1(\lambda_1,\lambda_2)$ at $\lambda_1=1$. If there is 1 intersection point, denote the intersection solution as $\lambda_1^{(b)}$. Graphically, it is clear to see that if and only if $\lambda_1^{(b)}<\lambda_1^{\mrt}$,  the optimal solution is $\lambda_1=\lambda_1^{\mrt}$.
Thus, we have for arbitrary fixed $\lambda_2$,
\begin{eqnarray*}
\lambda_1=1 &\text{ is optimal if and only if  }& F_1(1,\lambda_2) < F_2(1, \lambda_2)\\
\lambda_1= \lambda_1^{\mrt} &\text{ is optimal if and only if } &
 \lambda_1^{(b)} \leq \lambda_1^{\mrt}
  \end{eqnarray*} where $\lambda_1^{(b)}$ is given in \eqref{eqt:lambda_bal_1}. Now we combine the conditions for $\lambda_1$ and $\lambda_2$ and after some manipulations, we obtain the following:
  \begin{eqnarray}
  && \nonumber \left( \lambda_1^{\mrt}, \lambda_2^{\mrt}\right) \text{ is optimal if and only if }  \\
  & & \frac{c_1 \| \Pi_{21}^\perp \mathbf{h}_{11}\|^2}{c_2 \| \mathbf{h}_{21}\|^2 - 2 \sqrt{c_1 c_2} |\mathbf{h}_{21}^H \mathbf{h}_{11} | + c_1 \| \mathbf{h}_{11}\|^2}  < \cos^2(\theta_1) \leq \frac{(1+ \| \mathbf{h}_{22}\|^2 P_{max}) \| \mathbf{h}_{11}\|^2}{(1+ \| \mathbf{h}_{12} \|^2 \cos^2(\theta_2)P_{max} ) \| \mathbf{h}_{21}\|^2 } \\
&&  \left( \lambda_1=1, \lambda_2^{\mrt} \right) \text{ is optimal if and only if  }  \| \mathbf{h}_{21}\|^2 \leq (1+ \| \mathbf{h}_{22}\|^2 P_{max}) \frac{\| \Pi_{21} \mathbf{h}_{11}\|^2 }{1+ \| \mathbf{h}_{12}\|^2 \cos^2(\theta_2) P_{max}}
  \end{eqnarray} where $c_1=\frac{P_{max}}{\| \mathbf{h}_{12}\|^2 \cos^2(\theta_2) P_{max} + 1}$ and $c_2=\frac{P_{max}}{\|\mathbf{h}_{22} \|^2 P_{max}+1}$.

\subsection{Proof of MRT optimality in the DD region}\label{app:mrtopt_rdd}
In this section, we provide the sum rate optimality conditions for two MRT strategies, namely: interference amplifying beamforming $\mathbf{w}_i=\frac{\mathbf{h}_{ji}}{\| \mathbf{h}_{ji}\|}$, in Section \ref{section:mrtoptrddh21} and direct channel beamforming $\mathbf{w}_i=\frac{\mathbf{h}_{ii}}{\| \mathbf{h}_{ii}\|}$, in Section \ref{section:mrtoptrddh11}.

\subsubsection{Optimality conditions of amplifying interference in the DD region}\label{section:mrtoptrddh21} 
We aim to prove that the beamforming vector $\mathbf{w}_i=\frac{\mathbf{h}_{ji}}{\|\mathbf{h}_{ji} \|}$ is sum rate optimal in $\mathbf{R}^{dd}$ if and only if $(1+g_{jj}) \|  \mathbf{h}_{ii}\|^2 \cos^2(\theta_i)  \geq  \| \mathbf{h}_{ji}\|^2$ and 
$ \|\mathbf{h}_{ii}\|^2\cos^2(\theta_i) P_{max} \leq \frac{g_{ij}}{g_{jj}}-1$, where $g_{km}=\| \mathbf{h}_{km}^H \mathbf{w}_m\|^2$.

Due to symmetry of the problem, the proof for $\mathbf{w}_1$ and $\mathbf{w}_2$ is similar and we only give the proof for $\mathbf{w}_2$ here. First, by the definition of $\Omega^{01}$ \eqref{eqt:omega}, we observe that $\mathbf{w}_2(\lambda_2^{\mrt})=\frac{\mathbf{h}_{12}}{\|\mathbf{h}_{12} \|} $ is in the constraint set $\Omega^{01}$ if and only if the following constraints are satisfied:
\begin{equation*}
\left\{
\begin{aligned}
(1+g_{11}) \left| \mathbf{h}_{22}^H \mathbf{w}_2(\lambda_2^{\mrt}) \right|^2 & \geq  | \mathbf{h}_{12}^H \mathbf{w}_2(\lambda_2^{\mrt}) |^2\\
 \left| \mathbf{h}_{22}^H \mathbf{w}_2(\lambda_2^{\mrt}) \right|^2 P_{max} &\leq \frac{g_{21}}{g_{11}}-1
 \end{aligned} \right.
\end{equation*} 
which are equivalent to 
\begin{equation}\label{eqt:mrtoptrdd}
\left\{
\begin{aligned}
(1+g_{11}) \|  \mathbf{h}_{22}\|^2 \cos^2(\theta_2) & \geq  \| \mathbf{h}_{12}\|^2\\
 \|\mathbf{h}_{22}\|^2\cos^2(\theta_2) P_{max} &\leq \frac{g_{21}}{g_{11}}-1
 \end{aligned} \right. .
\end{equation} 

Thus, if and only if \eqref{eqt:mrtoptrdd} is satisfied, the beamforming vector $\mathbf{w}_2(\lambda_2^{\mrt})=\frac{\mathbf{h}_{12}}{\|\mathbf{h}_{12} \|}$ is in $\Omega^{01}$. Now we establish that this is the sum rate optimal solution.

$\tilde{Z}_1$ is monotonically increasing with $g_{22}$ in constraint set $\Omega^{11}$.  Note that $\mathbf{w}_2(1)=\frac{\mathbf{h}_{22}}{\| \mathbf{h}_{22}\|}$ \footnote{From now on, we write $\mathbf{w}_2(\lambda_2=1)$ as $\mathbf{w}_2(1)$. We must not confuse this with the first element of vector $\mathbf{w}_2$.} is \emph{not} in the constraint set $\Omega^{11}$. We see this by observing the constraint set of $\Omega^{11}$ requires:
\begin{itemize}
 \item $|\mathbf{h}_{22}^H \mathbf{w}_2(\lambda_2) |^2= \lambda_2 \| \mathbf{h}_{22}\|^2$. Thus, we have $|\mathbf{h}_{22}^H \mathbf{w}_2(\lambda_2^{\mrt}) |^2 \leq |\mathbf{h}_{22}^H \mathbf{w}_2(1) |^2$.
 \item $|\mathbf{h}_{12}^H \mathbf{w}_2(\lambda_2) |^2= \left( \sqrt{\lambda_2} \| \Pi_{22} \mathbf{h}_{12}\| + \sqrt{1-\lambda_2} \| \Pi_{22}^\perp \mathbf{h}_{12}\|\right)^2$ is concave in $\lambda_2$ and attains maximum at $\lambda_2^{\mrt}$ where $\mathbf{w}_2(\lambda_2^{\mrt})=\frac{\mathbf{h}_{12}}{\| \mathbf{h}_{12}\|}$. Thus, $|\mathbf{h}_{12}^H \mathbf{w}_2(1)|^2 \leq |\mathbf{h}_{12}^H \mathbf{w}_2(\lambda_2^{\mrt}) |^2$.
 \end{itemize}
 If the conditions in \eqref{eqt:mrtoptrdd} are satisfied, we have
 \begin{equation*}
 (1+g_{11})|\mathbf{h}_{22}^H \mathbf{w}_2(1) |^2 \geq (1+g_{11}) \|  \mathbf{h}_{22}\|^2 \cos^2(\theta_2)  \geq  \| \mathbf{h}_{12}\|^2 \geq |\mathbf{h}_{12}^H \mathbf{w}_2(1) |^2.
 \end{equation*} To satisfy the constraints of both $\Omega^{01}$ and $\Omega^{11}$ the sum rate optimal solution lies on the boundary between $\Omega^{01}$ and $\Omega^{11}$, namely $\Omega^{A}$.

The constraints set $\Omega^{10}$ is empty. Using the same argument as in the case of $\tilde{Z}_1$, for any $\lambda_2$ that satisfies $|\mathbf{h}_{22}^H \mathbf{w}_2(\lambda_2)|^2 \geq \frac{g_{21}}{g_{11}}-1$ must satisfy $\lambda_2 \geq \lambda_2^{\mrt}$. Also, any $\lambda_2 \geq \lambda_2^{\mrt}$ satisfies
 \begin{equation}
 (1+g_{11})|\mathbf{h}_{22}^H \mathbf{w}_2(\lambda_2) |^2 \geq (1+g_{11}) \|  \mathbf{h}_{22}\|^2 \cos^2(\theta_2)  \geq  \| \mathbf{h}_{12}\|^2 \geq |\mathbf{h}_{12}^H \mathbf{w}_2(\lambda_2) |^2.
 \end{equation} Thus, for any $\lambda_2$ that satisfies $B=0$ must have $A=0$, whih indicates that the constraint is empty.

Similar to the argument before, $|\mathbf{h}_{22}^H \mathbf{w}_2(\lambda_2)|^2 \geq \frac{g_{21}}{g_{11}}-1$ is a tighter constraint for $|\mathbf{h}_{22}^H \mathbf{w}_2(\lambda_2)|^2 $ than $(1+ g_{11})|\mathbf{h}_{22}^H \mathbf{w}_2(\lambda_2) |^2 \geq |\mathbf{h}_{12}^H \mathbf{w}_2(\lambda_2)|^2$ in $\Omega^{00}$.  As $\tilde{Z}_4$ is monotonically decreasing with $|\mathbf{h}_{22}^H \mathbf{w}_2(\lambda_2)|^2$ and increasing with $|\mathbf{h}_{12}^H \mathbf{w}_2(\lambda_2)|^2$ in $\Omega^{00}$ , the sum rate optimal solution in this case is the beamforming vector which satisfies:
 \begin{equation*}\left\{
 \begin{aligned}
 |\mathbf{h}_{22}^H \mathbf{w}_2(\lambda_2)|^2 &= \frac{g_{21}}{g_{11}}-1\\
 (1+g_{11}) |\mathbf{h}_{22}^H \mathbf{w}_2(\lambda_2)|^2 &= |\mathbf{h}_{12}^H \mathbf{w}_2(\lambda_2) |^2
 \end{aligned} \right.
 \end{equation*} which is in $\Omega^B$.
 
 Now we show that $\mathbf{w}_2(\lambda_2^{\mrt})=\frac{\mathbf{h}_{12}}{\| \mathbf{h}_{12}\|}$ is the optimal solution in $\Omega^{01}$. $\tilde{Z}_2$ is monotonically increasing with $g_{12}$. Since we assumed that $\frac{\mathbf{h}_{12}}{\| \mathbf{h}_{12}\|}$ is in $\Omega^{01}$, it is the optimal solution.

Finally, we notice that for any sum rate optimal solutions in $\Omega^A$ which maximizes $\tilde{Z}_1$ in $\Omega^{11}$, it is also in the constraint set $\Omega^{01}$ and therefore achieves a smaller sum rate than $\frac{\mathbf{h}_{12}}{\| \mathbf{h}_{12}\|}$. Similarly, any solution in $\Omega^B$ that maximizes $\tilde{Z}_4$ is also in constraint set $\Omega^{01}$ and therefore achieves a smaller sum rate than $\frac{\mathbf{h}_{12}}{\| \mathbf{h}_{12}\|}$.

\subsubsection{Optimality conditions of direct channel beamforming in the DD region}\label{section:mrtoptrddh11}

Now we prove that the beamforming vector $\mathbf{w}_i(\lambda_i)=\frac{\mathbf{h}_{ii}}{\| \mathbf{h}_{ii}\|}$ attains the maximum sum rate in the DD region, for arbitrary fixed $\mathbf{w}_j$ if $(1+g_{jj}) | \mathbf{h}_{ii}^H \mathbf{w}_i(1)|^2 \leq |\mathbf{h}_{ji}^H \mathbf{w}_i(1) |^2  $.

We provide the proof for $\mathbf{w}_2$ for simplicity as by reversing the role Tx 1 and Tx 2, the proof for $\mathbf{w}_1$ can be obtained.
Notice that if the optimality condition is true: $(1+g_{11}) | \mathbf{h}_{22}^H \mathbf{w}_2(1)|^2 \leq |\mathbf{h}_{12}^H \mathbf{w}_2(1) |^2  $, then the following arguments are true.

The constraint sets $\Omega^{01}$ and $\Omega^{00}$ are empty. 
This is because 
$$(1+g_{11}) | \mathbf{h}_{22}^H \mathbf{w}_2(0)|^2 =(1+g_{11}) \left| \mathbf{h}_{22}^H \frac{\Pi_{22}^\perp \mathbf{h}_{12}}{\|\Pi_{22}^\perp \mathbf{h}_{12} \|}\right|^2 = 0 \leq  |\mathbf{h}_{12}^H \mathbf{w}_2(0)|^2.$$ Thus, together with the assumption above: $(1+g_{11}) | \mathbf{h}_{22}^H \mathbf{w}_2(1)|^2 \leq |\mathbf{h}_{12}^H \mathbf{w}_2(1) |^2  $, we have
\begin{equation*}\left\{
\begin{aligned}
|\mathbf{h}_{12}^H \mathbf{w}_2(1)|^2 & \geq (1+g_{11}) |\mathbf{h}_{22}^H \mathbf{w}_2(1)|^2\\
|\mathbf{h}_{12}^H \mathbf{w}_2(0)|^2 & \geq (1+g_{11}) |\mathbf{h}_{22}^H \mathbf{w}_2(0)|^2.
\end{aligned}\right. 
\end{equation*} Since $|\mathbf{h}_{22}^H \mathbf{w}_2(\lambda)|^2$ is linearly increasing with $\lambda$ and $|\mathbf{h}_{12}^H \mathbf{w}_2(\lambda)|^2$ is concave in $\lambda$, we draw the conclusion that for all $\lambda$, $|\mathbf{h}_{12}^H \mathbf{w}_2(\lambda)|^2 \geq (1+g_{11}) |\mathbf{h}_{22}^H \mathbf{w}_2(\lambda)|^2$. Thus, for all $\lambda$, $A=1$.

 The sum rate optimal solution in $\Omega^{11}$ is either $\Omega^B$ or $\frac{\mathbf{h}_{22}}{\|\mathbf{h}_{22} \|}$. If the optimal solution is $\frac{\mathbf{h}_{22}}{\|\mathbf{h}_{22} \|}$ then we know that $\Omega^{10}$ is empty and $\frac{\mathbf{h}_{22}}{\|\mathbf{h}_{22} \|}$ is sum rate optimal. If the sum rate optimal solution in $\Omega^{11}$ is in $\Omega^B$, then these solutions are also in constraint set $\Omega^{10}$ which achieve a smaller sum rate of $\tilde{Z}_3$ than $\frac{\mathbf{h}_{22}}{\|\mathbf{h}_{22} \|}$.

Now, we obtained the MRT optimality conditions for each transmit beamformer $\mathbf{w}_i$ given $\mathbf{w}_j$. Apply the same approach and reverse the role of Tx 1 and 2, we obtain the conditions for $\mathbf{w}_j$. Combine both inequalities to obtain the conditions as shown in Theorem \ref{thm:mrtopt_rdd}.

\bibliographystyle{IEEEbib}
\bibliography{transit2011.bib} 

\end{document}